\documentclass[a4paper]{quantumarticle}
\pdfoutput=1
\usepackage[T1]{fontenc}

\usepackage[dvipsnames,table]{xcolor}
\usepackage{hyperref}
\hypersetup{
  breaklinks=true,
  colorlinks=true,
  allcolors=BlueViolet,
}

\usepackage[numbers,sort&compress]{natbib}
\usepackage{graphicx}
\graphicspath{{./figures/}} %

\usepackage[inline]{enumitem}
\setlist[enumerate]{noitemsep}

\usepackage{amsmath,amssymb,amsfonts,bm}
\DeclareMathOperator{\AcceptProb}{AcceptProb}
\DeclareMathOperator{\StabCode}{StabCode}
\DeclareMathOperator{\CSS}{CSS}
\DeclareMathOperator{\TrivCSS}{TrivCSS}
\DeclareMathOperator{\TrivSWEL}{TrivSWEL}
\DeclareMathOperator{\CNOT}{CNOT}

\DeclareMathOperator{\rank}{rank}
\DeclareMathOperator{\rowspan}{rowspan}
\DeclareMathOperator{\rows}{rows}
\DeclareMathOperator*{\argmin}{argmin}

\DeclareMathOperator{\GL}{GL}
\DeclareMathOperator{\Sp}{Sp}
\let\O\relax
\DeclareMathOperator{\O}{O}

\newcommand{\F}{\mathbb{F}}
\newcommand{\Z}{\mathbb{Z}}
\newcommand{\T}{\intercal}

\usepackage{braket}
\newcommand{\bk}{\braket}
\newcommand{\op}[2]{\ket{#1}\!\!\bra{#2}}

\usepackage{mathtools}
\DeclarePairedDelimiter{\norm}{\lVert}{\rVert}
\DeclarePairedDelimiter{\abs}{\lvert}{\rvert}
\let\set\relax
\DeclarePairedDelimiter{\set}{\{}{\}}

\usepackage{dsfont}
\newcommand{\1}{\mathds{1}}

\usepackage{stmaryrd}
\newcommand{\params}[1]{\llbracket #1 \rrbracket}

\usepackage{amsthm,thmtools,thm-restate}

\declaretheorem{lemma}
\declaretheorem{definition}

\usepackage[linesnumbered,ruled,vlined]{algorithm2e}
\SetAlgoCaptionSeparator{:}
\SetAlCapFnt{\normalsize}
\SetAlCapNameFnt{\normalsize}
\newcommand{\oper}[1]{\operatorname{\mathtt{#1}}}

\usepackage{tcolorbox}

\usepackage{tikz}
\usetikzlibrary{quantikz2}

\usepackage{nicefrac}

\usepackage{rotating}
\usepackage{subcaption}

\usepackage{makecell}

\usepackage[capitalize]{cleveref}

\usepackage{listings}
\definecolor{blue}{rgb}{0, 0.35, 0.577}
\definecolor{orange}{rgb}{0.858, 0.374, 0}
\definecolor{green}{rgb}{0, 0.501, 0.018}
\definecolor{red}{rgb}{0.702, 0, 0.008}
\definecolor{purple}{rgb}{0.459, 0.286, 0.61}
\newcommand{\mhl}[2]{
  {\setlength{\fboxsep}{1pt}
  \colorbox{#1}{\ensuremath{#2}}}
}

\title{Quantum codes from classical annealing}
\author{Michael A. Perlin}
\email{michael.perlin@jpmchase.com}
\author{Matthew Steinberg}
\affiliation{Global Technology Applied Research, JPMorganChase, New York, NY 10017, USA}
\author{Ben Criger}
\affiliation{Quantinuum, Terrington House, Hills Road, Cambridge CB2 1NL, UK}

\begin{document}

\begin{abstract}
  We introduce an adaptive simulated annealing algorithm to search for moderately-sized quantum error-correcting codes with high encoding rates and large distances.
  Our search targets two classes of stabilizer codes: (1) CSS codes, and (2) a subclass of CSS codes that we call ``self-dual with equivalent logicals'' (SWEL) codes, the latter of which admit transversal implementations of logical Hadamard and phase gates that can be leveraged to construct fault-tolerant gate sets.
  The search is guided by an energy function that acts as a surrogate for the logical error rate in a code-capacity noise model, combining code distance with a count of minimum-weight logical operators to resolve the discrete plateaux that impede na\"ive distance optimization.
  For block lengths of up to $56$ physical qubits, our search finds state-of-the-art CSS and SWEL codes whose distances frequently meet or exceed the variants of the quantum Gilbert-Varshamov bound.
  In addition to providing favorable seed codes for fault-tolerant architectures based on code concatenation, the codes found in this work are promising candidates for high-rate code demonstrations on near-term quantum computing hardware.
\end{abstract}

\maketitle
\tableofcontents

\section{Introduction}

Quantum error correction (QEC) is necessary for large-scale fault-tolerant (FT) quantum computation ~\cite{lidar2013quantum, nielsen2010quantum}.
While the threshold theorem guarantees the scalability of FT quantum computation on theoretical grounds~\cite{steane2003overhead, lidar2013quantum}, the overheads of QEC for known code families push the resource requirements for practical FT quantum computation beyond the capabilities of current and near-term quantum computers~\cite{preskill2018quantum}.
As such, much recent work has been devoted towards the exploration of new code families amenable to near-term implementation~\cite{he2025performance,jin2025iceberg,self2024protecting}.

Current quantum code design research largely follows three tracks.
First, there are topological stabilizer codes, which include the surface and color codes~\cite{lidar2013quantum}.
Despite great theoretical success and a clear roadmap to achieving universal FT quantum computation, topological codes incur exceptionally large qubit overheads due to the Bravyi-Poulin-Terhal (BPT) bound, which severely constrains achievable encoding rates and distances~\cite{bravyi2010tradeoffs}.
Since the surface code saturates the BPT bound, further optimization at the level of abstract code construction is largely ineffectual for topological codes.

A second stream of research focuses on quantum low-density parity-check (qLDPC) codes, which come equipped with guarantees of fault tolerance and have attracted particular attention since the discovery of asymptotically ``good'' qLDPC codes~\cite{panteleev2022asymptotically, dinur2023good}.
However, qLDPC codes with favorable code parameters necessarily require long-range interactions~\cite{baspin2022quantifying, dai2025locality}, as well as complicated decoding~\cite{koutsioumpas2025colour, panteleev2021degenerate} and logical computation~\cite{williamson2026lowoverhead, swaroop2024universal, he2025extractors, baspin2025fast} techniques.
Furthermore, known asymptotically good code families cannot be used at the scale of current and near-term hardware due to large sub-leading-order overheads that dominate at small block sizes.

Finally, code concatenation has experienced a recent resurgence of interest, due mainly to the construction of architectures with constant spacetime overheads and efficient decoding algorithms~\cite{goto2024highperformance, yamasaki2024timeefficient, yoshida2025concatenate, nakai2026subsystem}.
Although such codes admit high-rate and high-distance constructions, there remains much work at the level of low-overhead FT error correction and logical gates, among other engineering challenges at finite size, as has been explored in the holographic code concatenation literature~\cite{pastawski2015holographic, harris2018calderbankshorsteane, harris2020decoding, jahn2021holographic, farrelly2022parallel, steinberg2025far, steinberg2025universal}.

While qLDPC and concatenated codes are incompatible with solid-state hardware that have nearest-neighbor qubit connectivity, trapped ion and neutral atom hardware allows qubits to be routed across a device while maintaining coherence, thereby facilitating long-range interactions~\cite{bluvstein2022quantum, ransford2025helios}.
Notwithstanding recent work in designing large-scale codes with favorable code parameters~\cite{bravyi2024highthreshold, cain2026shors, zhao2026ultrahighrate}, brute-force methods still permit the discovery of moderately-sized codes with state-of-the-art code parameters.
Moderately-sized codes may also be used in larger-scale future devices with a QEC architecture based on concatenated codes.

In this work, we perform a numerical search for moderately-sized quantum stabilizer codes that encode many logical qubits.
Specifically, we consider $\params{n,k}$ codes with $n\le56$ physical qubits and $k\ge4$ logical qubits.
We first perform a search over CSS codes, which admit relatively simple FT logical state preparation circuits, QEC cycles, and decoders.
We then search over a subclass of CSS codes that we call SWEL codes, which admit transversal logical Hadamard and phase gates ($\overline{H}^{\otimes k}$ and $\overline{S}^{\otimes k}$) that can be leveraged to construct FT gate sets~\cite{tansuwannont2025clifford, sullivan2026injection}.
Despite the restriction to CSS and SWEL codes, most of the codes we discover achieve, for a given choice of physical qubits $n$ and logical qubits $k$, a code distance $d$ that meets or exceeds the distance $d_{\mathrm{QGV}}(n,k)$ guaranteed to be achievable by the quantum Gilbert--Varshamov bound for \emph{arbitrary} (non-CSS) stabilizer codes~\cite{gottesman1997stabilizer, gottesman2024surviving}.
Our search algorithm is based on simulated annealing~\cite{vanlaarhoven1987simulated}, with an energy function designed from a surrogate of the logical error rate in a code-capacity model, and an adaptive temperature schedule designed to prevent freezing in barren plateaux and suppress temperature fluctuations around deep energy basins.

The remainder of this paper is organized as follows.
We review some relevant background and establish notation in \cref{sec:background}.
We describe our numerical methods and the reasoning that underpins them in \cref{sec:methods}, showing in particular that our search covers the entire space of CSS codes in \cref{sec:search_css} and SWEL codes in \cref{sec:search_swel}.
We describe our construction of an annealing energy function in \cref{sec:energy_function}, and an adaptive temperature schedule in \cref{sec:adaptive_annealing}.
We describe how we find minimum-weight stabilizer generators of numerically optimized codes in \cref{sec:reducing_weight}.
We present and discuss the best codes that we discover in \cref{sec:results}.
Finally, we provide closing remarks and discuss future directions, such as strategies for scaling up our code search, in \cref{sec:discussion}.

\section{Background}
\label{sec:background}

Here we review some relevant background and establish notation used in this work.
Unless stated otherwise, arithmetic operations in this work occur over the binary field $\F_2$, meaning that addition is taken modulo 2 and minus signs are ignored.
A notable exception to this rule is that if $v \in \F_2^n$, then the (Hamming) \emph{weight} $\norm{v} = \sum_i v_i \in \mathbb{N}_0$ of $v = (v_i)_i$ is the integer number of ones in $v$.
The occasional treatment of numbers in $\F_2$ as integers should be clear from context.

We denote the transpose of a matrix $A$ by $A^\T$, and if $A$ is invertible we let $A^{-\T} = (A^{-1})^\T = (A^\T)^{-1}$.
All one-dimensional vectors $\bm v\in\F_2^n$ are treated as column vectors unless explicitly indicated otherwise by a transpose symbol, $\bm v^\T$.
Finally, we let $\Z_N = \set{1, 2, \cdots, N}$, denote the $m \times m$ identity matrix by $\1_m$, and denote the $m\times n$ zero matrix by $\bm{0}_{m\times n}$, though to declutter notation we may at times write $\1$ or $\bm{0}$ when matrix dimensions are clear from context.

\subsection{Pauli strings}
\label{sec:pauli_strings}

Let $X = \op{1}{0} + \op{0}{1}$ and $Z = \op{0}{0} - \op{1}{1}$ be single-qubit Pauli-$X$ and Pauli-$Z$ operators.
Given $n$-bit vectors $\bm x,\bm z \in \F_2^n$, we define the Pauli-$X$ and Pauli-$Z$ strings
\begin{align}
  X(\bm x) = \bigotimes_{j=1}^n X^{x_j},
  &&
  Z(\bm z) = \bigotimes_{j=1}^n Z^{z_j},
\end{align}
and more generally the Pauli string
\begin{align}
  P(\bm x, \bm z)
  = i^{\norm{\bm x\land\bm z}} \bigotimes_{j=1}^n X^{x_j} Z^{z_j}
  \cong X(\bm x) Z(\bm z),
  \label{eq:pauli_string}
\end{align}
where $\bm x\land\bm z$ is the bitwise AND of $\bm x$ and $\bm z$, and $\cong$ denotes equality up to a phase factor.
If $\bm x,\bm z$ are collected into the $2n$-bit vector $\bm p = (\bm x, \bm z) \in \F_2^{2n}$, then we additionally define $P(\bm p) = P(\bm x,\bm z)$.
For simplicity, we mod out all Pauli strings by their phase prefactor, allowing us to uniquely identify each $n$-qubit Pauli string with an element of $\F_2^{2n}$; this simplification has no consequences for the results and discussions in this work.

The \emph{weight} of a Pauli string is the number of qubits that it addresses nontrivially:
\begin{align}
  \norm{P(\bm x,\bm z)} = \norm{\bm x\lor\bm z},
\end{align}
where $\bm x\lor\bm z$ is the bitwise OR of $\bm x$ and $\bm z$.
The symplectic inner product of Pauli strings $P(\bm p_1)$ and $P(\bm p_2)$ with $\bm p_1 = (\bm x_1, \bm z_1)$ and $\bm p_2 = (\bm x_2, \bm z_2)$ is
\begin{align}
  \bk{P(\bm p_1), P(\bm p_2)}
  = \bm p_1^\T \Omega_n \bm p_2
  = \bm x_1^\T \bm z_2 + \bm z_1^\T \bm x_2,
\end{align}
where
\begin{align}
  \Omega_n =
  \begin{pmatrix}
    \bm{0} & \1_n \\
    \1_n & \bm{0}
  \end{pmatrix}
\end{align}
is the binary symplectic form.
The symplectic inner product $\bk{P(\bm p_1), P(\bm p_2)} = 0$ if $P(\bm p_1)$ and $P(\bm p_2)$ commute, and $1$ otherwise.

\subsection{Classical linear codes}
\label{sec:classical_linear_codes}

An $[n,k]$ linear code $C$ is a subspace of $\F_2^n$ spanned by $k$ linearly independent bitstrings.
Here $n$ is called the \emph{block length} of $C$, and $\nicefrac{k}{n}$ is the \emph{encoding rate}.
Every linear code $C$ is the kernel of some parity check matrix $H\in\F_2^{r\times n}$ (for some integer $r$); that is,
\begin{align}
  C = \ker(H) = \set{\bm x \in \F_2^n: H\bm x = \bm{0}}.
\end{align}
If $C$ is an $[n,k]$ code with parity check matrix $H$, then $\rank(H) = n-k$.

\subsection{Quantum stabilizer codes}
\label{sec:quantum_stab_codes}

Let $S$ be a set of pair-wise commuting $n$-qubit Pauli strings, which is to say that if $a, b\in S$ then $\bk{a, b} = 0$.
The stabilizer code $\StabCode(S)$ is the subspace of the $n$-qubit Hilbert space $\mathcal{H}_n$ that is stabilized by all Pauli strings in $S$:
\begin{align}
  \StabCode(S) = \set{\ket\psi \in \mathcal{H}_n: s\ket\psi = \ket\psi ~\forall~ s\in S}.
\end{align}
A stabilizer code $C$ is said to be an $\params{n,k}$ code to mean that $C$ addresses $n$ \emph{physical} qubits (that is, its stabilizers are $n$-qubit strings), and $C$ encodes $k$ \emph{logical} qubits, with Hilbert space dimension $2^k$.
In this case $n$ is the block length of $C$, and $\nicefrac{k}{n}$ is its encoding rate.
A parity check matrix of $C = \StabCode(S)$ is a matrix $H\in \F_2^{r\times 2n}$ (for some integer $r$) whose rows generate the stabilizer group of $C$,
\begin{align}
  \set{P(\bm p):\bm p\in\rowspan(H)} = \bk{S},
\end{align}
where $\bk{S}$ is the commutative group generated by $S$.
If $\StabCode(S)$ is an $\params{n,k}$ code with parity check matrix $H$, then $\rank(H) = n-k$; equivalently, the stabilizer group $\bk{S}$ is generated by $n-k$ linearly independent Pauli strings.

\subsection{CSS codes}
\label{sec:css_codes_bkgd}

Let $H_X$ and $H_Z$ be parity check matrices of classical linear codes.
These parity check matrices and the respective codes are said to be \emph{CSS-compatible} iff $H_X H_Z^\T = \bm{0}$.
In this case, we can define $\CSS(H_X, H_Z)$ as the stabilizer code with parity check matrix
\begin{align}
  H =
  \begin{pmatrix}
    H_X & \bm{0} \\ \bm{0} & H_Z
  \end{pmatrix}.
\end{align}
We will refer to $(H_X,H_Z)$ as the parity check matrices of $\CSS(H_X, H_Z)$.
In words, $H_X$ and $H_Z$ define the support of $X$-type and $Z$-type stabilizers of $\CSS(H_X, H_Z)$, meaning
\begin{align}
  \CSS(H_X,H_Z) = \StabCode(X(H_X) \cup Z(H_Z)),
\end{align}
where
\begin{align}
  X(H_X) &= \set{X(\bm x): \bm x\in \rows(H_X)},
  \\
  Z(H_Z) &= \set{Z(\bm z): \bm z\in \rows(H_Z)}.
\end{align}
The condition $H_X H_Z^\T = \bm{0}$ ensures that elements of $X(H_X)$ and $Z(H_Z)$ commute.
If $\ker(H_X)$ is an $[n,k_X]$ code and $\ker(H_Z)$ is an $[n,k_Z]$ code, then $\CSS(H_X,H_Z)$ is an $\params{n,k}$ code with $k = k_X + k_Z - n$.
The code $\CSS(H_X, H_Z)$ is called \emph{self-dual} if $H_X$ and $H_Z$ correspond to the same linear code: $\ker(H_X) = \ker(H_Z)$.

\subsection{Logical operators and code distance}
\label{sec:logical_ops}

Let $C = \StabCode(S)$ be an $\params{n,k}$ stabilizer code.
Any $n$-qubit operator $O$ that preserves the code space of $C$, meaning $O\ket\psi\in C$ for all $\ket\psi\in C$, is called a \emph{logical operator} of $C$.
For simplicity, we restrict our discussion to logical operators that are elements of the Pauli group $\mathcal{P}_n = \set{P(\bm p):\bm p\in\F_2^{2n}}$.
The logical operators of a code must commute with its stabilizers; indeed, the set of all logical operators of $C$ is the normalizer of of its stabilizer group, $N(S) = \set{P\in\mathcal{P}_n:\bk{P,s}=0~\text{for all}~s\in S}$.

Elements of the stabilizer group $\bk{S}$ are trivial logical operators that act as the identity operator on $C$.
If $k>0$, then there additionally exist logical operators with nontrivial action on $C$.
The \emph{distance} of $C$ is the minimum weight of its nontrivial logical operators:
\begin{align}
  d(C) = \min\set{\norm{P}:P\in N(S), P\notin\bk{S}}.
\end{align}
An $\params{n,k}$ code with distance $d$ is called an $\params{n,k,d}$ code.

A \emph{symplectic basis of logical operators} for $C$ is a choice of logical operators that act as single-qubit Pauli $X$ and $Z$ operators on each of $k$ logical qubits encoded by $C$.
Such a basis can be represented by matrix $L\in\F_2^{2k\times 2n}$ in which rows $q\in\Z_k$ and $q+k$ represent, respectively, Pauli $X$ and $Z$ operators for logical qubit $q$ of $C$, meaning $L\Omega_n L^\T = \Omega_k$.
Every logical operator $Q$ of $C$ can be decomposed as $Q = P(\bm\ell) s$ for some vector $\bm\ell\in\rowspan(L)\subset\F_2^{2n}$ and stabilizer $s\in\bk{S}$.

If $C = \CSS(H_X, H_Z)$ is an $\params{n,k}$ CSS code, then it admits a symplectic basis of logical operators that is block-diagonal,
\begin{align}
  L =
  \begin{pmatrix}
    L_X & \bm{0} \\
    \bm{0} & L_Z
  \end{pmatrix},
\end{align}
where the matrices $L_X,L_Z\in\F_2^{k\times n}$ indicate, respectively, the support of Pauli-$X$ and Pauli-$Z$ strings that come in pair-wise anti-commuting pairs, meaning $L_X L_Z^\T = \1_k$.
In this case, we will refer to $(L_X, L_Z)$ as a symplectic basis of logical operators for $C$.

\subsection{Clifford gates}
\label{sec:clifford_gates}

Clifford gates are unitary operations that map Pauli strings to Pauli strings.
Modulo global phases and Paulis, the set of all $n$-qubit Clifford gates forms the \emph{symplectic group} $\Sp(2n,\F_2)$.
A matrix $\tau\in\F_2^{2n\times 2n}$ belongs to $\Sp(2n,\F_2)$ iff it preserves the binary symplectic form,
\begin{align}
  \tau^\T \Omega_n \tau = \Omega_n.
\end{align}
Every $n$-qubit Clifford gate $\tau$ can be expressed in a $2 \times 2$ block form called a \emph{Clifford tableau},
\begin{align}
  \tau =
  \begin{pmatrix}
    \tau_{XX} & \tau_{XZ} \\
    \tau_{ZX} & \tau_{ZZ}
  \end{pmatrix},
\end{align}
where $\tau_{XX}, \tau_{XZ}, \tau_{ZX}, \tau_{ZZ} \in \F_2^{n \times n}$.
The action of a Clifford gate $\tau$ on a Pauli string $P(\bm p) = P(\bm x, \bm z)$ is induced by matrix multiplication,
\begin{multline}
  \tau(P(\bm p))
  = P(\tau \bm p) \\
  = P(\tau_{XX} \bm x + \tau_{XZ} \bm z, \tau_{ZX} \bm x + \tau_{ZZ} \bm z).
\end{multline}
Preservation of the symplectic form implies that $\tau$ preserves the symplectic inner product (that is, commutation relations) between Pauli strings,
\begin{align}
  \bk{\tau(P(\bm p_1)), \tau(P(\bm p_2))} = \bk{P(\bm p_1), P(\bm p_2)}.
\end{align}
The action of a Clifford gate $\tau$ on a stabilizer code is induced by its action on the code's stabilizers,
\begin{align}
  \tau(\StabCode(S)) = \StabCode(\tau(S)),
\end{align}
where $\tau(S) = \set{\tau(s) : s \in S}$.
If $H$ is a matrix with $2n$ columns, such as the parity check matrix of an $n$-qubit stabilizer code, then the action of $\tau$ on $H$ is $\tau(H) = H \tau^\T$.

\subsection{Coding bounds}
\label{sec:coding_bounds}

The question of how efficiently a block of $n$ qubits can encode $k$ logical qubits at some distance $d$ is a foundational question in quantum error correction~\cite{gottesman1997stabilizer}.
Here, we briefly review some coding bounds that are relevant to this work: the quantum Hamming bound, and three variants of the quantum Gilbert-Varshamov bound.
We refer interested readers to Ref.~\cite{gottesman1997stabilizer} for further details.

\subsubsection{Quantum Hamming bound}

The quantum Hamming bound is the quantum analogue of the classical Hamming bound, also known as the sphere-packing bound, on the parameters of error-correcting codes.
Consider an $\params{n,k,d}$ stabilizer code that can correct all Pauli errors up to weight $t = \lfloor\frac{d-1}{2}\rfloor$.
The number of distinct Pauli errors of weight $w$ is $3^w \binom{n}{w}$, where $\binom{n}{w}$ reflects the choice of support for a $w$-qubit Pauli string, and there are $3$ choices of non-identity Pauli operator per addressed qubit.
The total number of correctable errors (including the trivial error with $w=0$) is then
\begin{align}
  \sum_{w=0}^t 3^w \binom{n}{w}.
\end{align}
If the code is \emph{non-degenerate}, meaning that each correctable error produces a distinct syndrome, then each error maps the $2^k$-dimensional code space to a unique orthogonal copy of the code space.
These subspaces must fit within the full $2^n$-dimensional Hilbert space without overlapping, giving rise to the constraint
\begin{align}
  2^k \sum_{w=0}^t 3^w \binom{n}{w} \le 2^n,
  &&
  \text{or}
  &&
  \sum_{w=0}^t 3^w \binom{n}{w} \le 2^{n-k}.
  \label{eq:bound_ham}
\end{align}
\cref{eq:bound_ham} provides an upper bound on the distance $d$ achievable by a non-degenerate $\params{n,k}$ stabilizer code.
Degenerate codes, in which distinct correctable errors can produce the same syndrome (due to equality modulo stabilizers), can in principle exceed the quantum Hamming bound~\cite{sarvepalli2010degenerate}.

\subsubsection{Quantum Gilbert-Varshamov bounds}

While the Hamming bound provides a constraint on code parameters from above, the quantum Gilbert-Varshamov (QGV) bound guarantees the existence of codes from below.
By a counting argument over all $\params{n,k}$ stabilizer codes~\cite[Section 2.3]{gottesman2004co}, or a greedy construction of a random stabilizer code~\cite[Section 7.1]{gottesman1997stabilizer}, it can be shown that an $\params{n,k,d}$ stabilizer code exists as long as the number of Pauli strings with weight $w<d$ does not exceed the size of the stabilizer group,
\begin{align}
  \sum_{w=0}^{d-1} 3^w \binom{n}{w} \le 2^{n-k}.
  \label{eq:bound_qgv}
\end{align}
The largest distance $d$ that satisfies \cref{eq:bound_qgv} for a fixed choice of $\params{n,k}$ is called the \emph{stabilizer QGV bound} $d_{\mathrm{QGV}}(n,k)$~\cite{gottesman1997stabilizer, lidar2013quantum}.

The QGV bound guarantees the existence of a stabilizer code with parameters $\params{n,k,d}$.
However, it does not guarantee the existence of an $\params{n,k,d}$ CSS code, which requires additional structure.
Nonetheless, a counting argument analogous to that for the QGV bound~\cite[Theorem 2]{matsumoto2017two} guarantees the existence of an $\params{n,k,d}$ CSS code $\CSS(H_X,H_Z)$ for which $\ker(H_X)$ and $\ker(H_Z)$ are, respectively, $[n,k_X]$ and $[n,k_Z]$ codes, provided that
\begin{align}
  2 \sum_{w=1}^{d-1} \binom{n}{w} < \frac{2^n-1}{2^{k_x} - 2^{n - k_x} + 2^{k_z} - 2^{n-k_z}},
  \label{eq:bound_css}
\end{align}
\cref{eq:bound_css} is obtained by setting $d_X = d_Z$ in Theorem 2 of Ref.~\cite{matsumoto2017two}.
The right-hand side of \cref{eq:bound_css} is maximized by choosing
\begin{align}
  \set{k_X, k_Z}
  = \left\{\left\lfloor \frac{n+k}{2} \right\rfloor, \left\lceil \frac{n+k}{2} \right\rceil\right\}.
  \label{eq:bound_css_k}
\end{align}
We refer to the largest distance $d$ that satisfies \cref{eq:bound_css} subject to \cref{eq:bound_css_k} for a fixed choice of $\params{n,k}$ as the \emph{CSS GV bound} $d_{\mathrm{CSS}}$.

\subsection{Simulated annealing}
\label{sec:simulated_annealing}

Our objective is nominally to find stabilizer codes with good encoding rates and large distances.
This is a challenging task: the space of stabilizer codes with block length $n$ is exponentially large in $n$, and computing the distance of a stabilizer code takes exponential time in $n$.
At a high level, we address this problem by fixing a choice of physical and logical qubit number, $\params{n,k}$, and formulating our task as a search problem targeting a code with a given distance.
To this end, we use an adaptive simulated annealing algorithm to optimize over suitably defined spaces of stabilizer codes.
Here we briefly review simulated annealing in general terms, leaving the details of our implementation to \cref{sec:methods}.

We first describe a Metropolis-Hastings algorithm~\cite{metropolis1953equation, hastings1970monte, tierney1994markov} that underpins simulated annealing.
The Metropolis-Hastings algorithm is a Markov-chain Monte Carlo algorithm for sampling from a set $\mathcal{X}$.
This algorithm requires four ingredients.
\begin{enumerate}[label=(\Alph*)]
  \item\label{input:state} An initial \emph{state} $x_0\in\mathcal{X}$.
  \item A \emph{proposal distribution} $g(y|x)$ for $x,y\in\mathcal{X}$.
  \item An \emph{energy function} $E:\mathcal{X}\to\mathbb{R}$.
  \item\label{input:temperature} A \emph{temperature} $T$.
\end{enumerate}
The proposal distribution $g(y|x)$ can be thought of as a probability distribution over \emph{neighbors} $y$ of $x$.
The algorithm works as follows.
\begin{enumerate}[label=(\arabic*)]
  \item\label{step:propose_a_move}
    For a given current state $x$, sample a neighbor $y\sim g(y|x)$, thereby proposing a \emph{move} $x\to y$.
  \item Letting $\Delta = E(y) - E(x)$ denote the energy difference of the proposed move, accept the move with probability
    \begin{align}
      \AcceptProb(\Delta,T)
      =
      \begin{cases}
        e^{-\Delta/T} & \Delta > 0 \\
        1 & \Delta \leq 0
      \end{cases}.
      \label{eq:accept_prob}
    \end{align}
  \item Return to Step \ref{step:propose_a_move}.
\end{enumerate}
The algorithm continues until some termination condition is satisfied, such as performing a maximum number of repetitions.
In the limit $T\to\infty$ this algorithm becomes a random walk, whereas in the limit $T\to 0$ it becomes a greedy algorithm to find a local minimum of $E$.
Under suitable conditions~\cite{metropolis1953equation, hastings1970monte, tierney1994markov}, the distribution of states observed at Step \ref{step:propose_a_move} of this algorithm converges to the Boltzmann distribution
\begin{align}
  \pi(x) \propto e^{-E(x)/T}.
\end{align}

Simulated annealing~\cite{kirkpatrick1983optimization, vanlaarhoven1987simulated} extends the Metropolis-Hastings algorithm by introducing a \emph{schedule} $T(t)$ for the temperature $T$ at iteration $t$.
The notion of a temperature schedule is inspired by the physical process of metallurgical annealing, wherein a metal is first heated and then gradually cooled to remove defects, resulting in a low-energy crystalline state.
Simulated annealing can thereby be thought of as an optimization algorithm to minimize the energy $E$ over $\mathcal{X}$.
Indeed, simulated annealing can be proven to find the global optimum of suitably well-behaved energy landscapes with a logarithmic cooling schedule $T(t)\propto 1/\log(t)$~\cite{geman1984stochastic}.
In practice, logarithmic cooling is too slow and global optimality is too stringent of a requirement, so simulated annealing is used as a heuristic optimization algorithm with faster cooling schedules, such as a geometric schedule $T(t)\propto e^{-t}$.

Good annealing schedules can be difficult to find because they are intrinsically problem-dependent.
Moreover, a reasonable choice of temperature at a given time throughout the algorithm may depend on the local energy landscape and history of observed states.
Rather than setting a predetermined annealing schedule, \emph{adaptive} simulated annealing adjusts the temperature throughout the algorithm according to some rule.
For example, the algorithm may consider the current time $t$ and the scale $\delta$ of recently observed energy differences, and make the temperature large relative to $\delta$ at early times, and small relative to $\delta$ at late times.
Such a rule can enforce early-time exploration and late-time refinement without detailed knowledge of the energy landscape.
Adaptive simulated annealing thereby requires specifying a rule and setting associated hyper-parameters, but may nonetheless be easier to tune than a fixed temperature schedule for simulated annealing.

\section{Methods}
\label{sec:methods}

Here we describe our numerical methods and the theory that motivates them.
We consider searching over two classes of codes: CSS codes, and the SWEL subclass of CSS codes, which were previously studied in Ref.~\cite{tansuwannont2025clifford}.
Each of our annealing runs fixes a choice of code class (CSS or SWEL), code parameters $\params{n,k}$, and the number of $X$-type and $Z$-type stabilizers of the code.
The annealing run then tries to minimize an energy function that decreases with code distance.

Altogether, ingredients \ref{input:state}--\ref{input:temperature} for the Metropolis-Hastings algorithm that underpins our adaptive simulated annealing algorithm are provided as follows:
\begin{enumerate}[label=(\Alph*)]
  \item The choice of initial state is always a random code, sampled according to a procedure described in \cref{sec:search_css} for CSS codes and \cref{sec:search_swel} for SWEL codes.
  \item The proposal distribution in our simulations is induced by a random choice of elementary gate by which to deform a code, namely a random CNOT gate for CSS codes (\cref{sec:search_css}) and a random four-qubit SWEL-preserving gate for SWEL codes (\cref{sec:search_swel}).
  \item The energy function in our simulations is motivated and defined in \cref{sec:energy_function}.
  \item The temperature throughout our simulations is set according to the adaptive annealing schedule described in \cref{sec:adaptive_annealing}.
\end{enumerate}
Finally, we describe how we find minimum-weight stabilizer generators of a code in \cref{sec:reducing_weight}.

\subsection{CSS code search}
\label{sec:search_css}

In order to describe our CSS code search, it will be useful to define a \emph{trivial} CSS code.
\begin{tcolorbox}
  \begin{definition}
    The \emph{trivial CSS code} $\TrivCSS(k, s_X, s_Z)$ for nonnegative integers $k, s_X, s_Z$ is
    \begin{align*}
      \TrivCSS(k, s_X, s_Z) = \CSS(H_X, H_Z),
    \end{align*}
    where
    \begin{align*}
      H_X &=
      \begin{pmatrix}
        \bm{0}_{s_X\times k} & \1_{s_X} & \bm{0}_{s_X\times s_Z}
      \end{pmatrix},
      \\
      H_Z &=
      \begin{pmatrix}
        \bm{0}_{s_Z\times k} & \bm{0}_{s_Z\times s_X} & \1_{s_Z}
      \end{pmatrix}.
    \end{align*}
    Here $\bm{0}_{a\times b}$ is the $a\times b$ zero matrix and $\1_c$ is the $c\times c$ identity matrix.
    \label{def:trivial_css}
  \end{definition}
\end{tcolorbox}
Equivalently,
\begin{multline}
  \TrivCSS(k, s_X, s_Z) = \\
  \set{
    \ket\psi \otimes\ket{+}^{\otimes s_X} \otimes \ket{0}^{\otimes s_Z}
    : \ket\psi\in\mathcal{H}_k
  },
\end{multline}
where $\mathcal{H}_k$ is the Hilbert space of $k$-qubit states.
The trivial CSS code is an $\params{n,k}$ code with $n = k + s_X + s_Z$.

In principle, a na\"ive CSS code annealing run could be performed with the following ingredients:
\begin{enumerate}[label=(\Alph*)]
  \item \textbf{Initial code}:
    Choose three positive integers $(k, s_X, s_Z)$, and initialize the code $\TrivCSS(k, s_X, s_Z)$.
  \item \textbf{Random move}:
    Choose two qubits $(a,b)$ with $a\ne b$ at random, and propose the move $C\to \CNOT_{a,b}(C)$.
\end{enumerate}
Such a code search is \emph{complete} in the sense that it covers the entire space of CSS codes with fixed $k,s_X,s_Z$.
\begin{tcolorbox}
  \begin{restatable}[CSS equivalence]{theorem}{CSSEquivalence}
    Let $C = \CSS(H_X, H_Z)$ and $C' = \CSS(H_X', H_Z')$.
    If
    \begin{enumerate}[label=(\roman*)]
      \item $\rank(H_X) = \rank(H_X')$, and
      \item $\rank(H_Z) = \rank(H_Z')$,
    \end{enumerate}
    then $C' = (G_N \cdots G_2 G_1)(C)$ for some finite sequence of CNOT gates $(G_1, G_2, \cdots, G_N)$.
    \label{thm:css_equivalence}
  \end{restatable}
  We prove \cref{thm:css_equivalence} in \cref{sec:css_equivalence} by reducing it to the decomposition of an $n$-qubit CSS-preserving gate $\tau$ into CNOT gates.
\end{tcolorbox}
However, the na\"ive CSS code search always starts at the same distance-one code for a given choice of $(k, s_X, s_Z)$.
In order to explore a larger space, and motivated by the goodness of random classical linear codes%
~\cite{barg2002random}, we can improve the CSS code search by starting with a random CSS code.

How do we produce a random CSS code?
We can produce a random $\params{n,k}$ stabilizer code by applying a random $n$-qubit Clifford gate to any $\params{n,k}$ stabilizer code.
To produce a random CSS code, we can similarly apply a random Clifford gate to a trivial CSS code, but restrict ourselves to the subset of Clifford gates that map CSS codes to CSS codes.
The following definition captures a necessary and sufficient condition for a Clifford gate to map CSS codes to CSS codes:
\begin{tcolorbox}
  \begin{definition}
    An $n$-qubit \emph{CSS-preserving gate} is an $n$-qubit Clifford gate $\tau$ with $\tau_{XZ} = \tau_{ZX} = 0$.
    If $\tau$ is a CSS-preserving gate, then we say that $\tau \simeq (\tau_X, \tau_Z)$ with $\tau_X = \tau_{XX}$ and $\tau_Z = \tau_{ZZ}$.
  \end{definition}
\end{tcolorbox}
Sampling $\tau$ from the entire space of $n$-qubit CSS-preserving gates and applying $\tau$ to the trivial code $\TrivCSS(k, s_X, s_Z)$ samples from the entire space of CSS codes with $k$ logical qubits and $(s_X,s_Z)$ linearly independent stabilizers of type $(X,Z)$.
\begin{tcolorbox}
  \begin{restatable}[CSS construction]{theorem}{CSSConstruction}
    If $C = \CSS(H_X, H_Z)$, then
    \begin{align*}
      C = \tau(\TrivCSS(k, \rank(H_X), \rank(H_Z)))
    \end{align*}
    for some CSS-preserving gate $\tau$.
    \label{thm:css_construction}
  \end{restatable}
  We prove \cref{thm:css_construction} in \cref{sec:css_construction} by constructing $\tau$ from standard-form parity check matrices of $C$.
\end{tcolorbox}
In turn, to sample the space of $n$-qubit CSS-preserving gates, we use the fact that the space of $n$-qubit CSS-preserving gates is isomorphic to the general linear group $\GL(n,\F_2)$ of invertible $n \times n$ matrices.
We use this isomorphism to sample $\tau_X \sim \GL(n,\F_2)$ and set $\tau \simeq (\tau_X, \tau_Z)$ with $\tau_Z = \tau_X^{-\T}$.
\begin{tcolorbox}
  \begin{lemma}
    If $\tau$ is a CSS-preserving gate, then $\tau_X^\T = \tau_Z^{-1}$ and $\tau_Z^\T = \tau_X^{-1}$ with $\tau_X\in\GL(n,\F_2)$.
    \label{lemma:css_clifford}
  \end{lemma}
  \begin{proof}
    Clifford gates preserve the symplectic inner product between Pauli strings, which implies that for all $\bm x, \bm z \in \F_2^n$ we must have
    \begin{align}
      \bk{X(\bm x), Z(\bm z)} &= \bk{\tau(X(\bm x)), \tau(Z(\bm z))}, \\
      \bm x^\T \bm z &= \bm x^\T \tau_X^\T \tau_Z \bm z,
    \end{align}
    so $\tau_X^\T = \tau_Z^{-1}$.
    An identical calculation with $Z(\bm z) \leftrightarrow X(\bm x)$ shows that $\tau_Z^\T = \tau_X^{-1}$.
  \end{proof}
\end{tcolorbox}
In total, our ingredients for a CSS code search are as follows:
\begin{enumerate}[label=(\Alph*)]
  \item \textbf{Initial code}:
    Choose three positive integers $(k, s_X, s_Z)$ and a random CSS-preserving gate $\tau \simeq (\tau_X, \tau_Z)$, obtained by sampling $\tau_X\sim\GL(n,\F_2)$ and setting $\tau_Z = \tau_X^{-\T}$.
    Initialize the code $\tau(\TrivCSS(k, s_X, s_Z))$.
  \item \textbf{Random move}:
    Choose two qubits $(a,b)$ with $a\ne b$ at random, and propose the move $C\to \CNOT_{a,b}(C)$.
\end{enumerate}

\subsection{SWEL code search}
\label{sec:search_swel}

Here we consider a particular subclass of CSS codes, previously studied in Ref.~\cite{tansuwannont2025clifford}.
\begin{tcolorbox}
  \begin{definition}
    A CSS code $C = \CSS(H_X, H_Z)$ is \emph{self-dual with equivalent logicals (SWEL)} iff
    \begin{enumerate}[label=(\roman*)]
      \item\label{prop:swel_1} $C$ is self-dual, $\ker(H_X) = \ker(H_Z)$, and
      \item\label{prop:swel_2} there exists a symplectic basis of logical operators $(L_X, L_Z)$ for $C$ with $L_X = L_Z$.
    \end{enumerate}
  \end{definition}
\end{tcolorbox}
Property \ref{prop:swel_1} (self-duality) ensures that $n$-fold physical Hadamard and phase gates, $\bigotimes_{j=1}^n H$ and $\bigotimes_{j=1}^n S$, each correspond to some logical gate of an $\params{n,k}$ SWEL code.
Property \ref{prop:swel_2} further guarantees that these gates implement $k$-fold \emph{logical} Hadamard and phase gates, $\bigotimes_{j=1}^k H$ and $\bigotimes_{j=1}^k S$ (up to Pauli corrections for the phase gate), which can be leveraged to construct fault-tolerant logical gate sets~\cite{tansuwannont2025clifford, sullivan2026injection}.
Every self-dual CSS code with odd block length $n$ is a SWEL code~\cite[Corollary 1]{tansuwannont2025clifford}, and being SWEL is both sufficient and necessary for a CSS code to have the global Hadamard and phase gates in its transversal gate set~\cite[Corollary 2]{tansuwannont2025clifford}.

In order to discuss our SWEL code search, we first define a trivial SWEL code:
\begin{tcolorbox}
  \begin{definition}
    The \emph{trivial SWEL code} $\TrivSWEL(k, s)$ for nonnegative integers $k, s$ is
    \begin{align*}
      \TrivSWEL(k, s) = \CSS(H_\star, H_\star),
    \end{align*}
    where
    \begin{align*}
      H_\star =
      \begin{pmatrix}
        \bm{0}_{s\times k} & \1_s & \1_s
      \end{pmatrix}.
    \end{align*}
    Here $\bm{0}_{s\times k}$ is the $s\times k$ zero matrix and $\1_s$ is the $s\times s$ identity matrix.
    \label{def:trivial_swel}
  \end{definition}
\end{tcolorbox}
Equivalently,
\begin{multline}
  \TrivSWEL(k, s) = \set{\ket\psi \otimes \ket{\Phi_s}: \ket\psi\in\mathcal{H}_k},
\end{multline}
where $\mathcal{H}_k$ is the Hilbert space of $k$-qubit states, and
\begin{align}
  \ket{\Phi_s} = \frac1{\sqrt{2^s}} \sum_{\bm a\in\F_2^s} \ket{\bm a} \otimes\ket{\bm a},
\end{align}
is the state of $s$ Bell pairs.
The trivial SWEL code is an $\params{n,k}$ code with $n = k + 2s$.

We next seek gates that map SWEL codes to SWEL codes.
Such gates must clearly be Clifford and CSS-preserving.
We can guarantee that the CSS-preserving gate $\tau \simeq (\tau_X, \tau_Z)$ maps SWEL codes to SWEL codes by enforcing that $\tau_X = \tau_Z$, so that the $X$ and $Z$ logical operators of a SWEL code transform identically.
Indeed, the condition $\tau_X = \tau_Z$ is both sufficient and necessary for our purposes.
\begin{tcolorbox}
  \begin{restatable}[SWEL stabilization]{theorem}{SWELStabilization}
    If $\tau \simeq (\tau_X, \tau_Z)$ is a CSS-preserving gate, then $\tau_X = \tau_Z$ is a necessary and sufficient condition for $\tau$ to map every SWEL code to a SWEL code.
    \label{thrm:swel_stabilization}
  \end{restatable}
  We prove \cref{thrm:swel_stabilization} in \cref{sec:swel_stabilization} by considering an arbitrary SWEL code $C$, the constraints on the parity check matrices and logical operators of $C$, and the requirements imposed on $\tau$ by analogous constraints for $\tau(C)$.
\end{tcolorbox}
\cref{thrm:swel_stabilization} thereby motivates the following definition:
\begin{tcolorbox}
  \begin{definition}
    An $n$-qubit \emph{SWEL-preserving gate} is an $n$-qubit CSS-preserving gate $\tau \simeq (\tau_X, \tau_Z)$ for which $\tau_X = \tau_Z$.
    For any SWEL-preserving gate $\tau$ we define $\tau_\star = \tau_X = \tau_Z$.
    \label{def:swel_preserving}
  \end{definition}
\end{tcolorbox}
Similarly to the case of CSS codes, we can produce a random SWEL code by applying a random SWEL-preserving gate to a trivial SWEL code.
But what is a random SWEL-preserving gate?
It helps to observe that the set of $n$-qubit SWEL-preserving gates is isomorphic to the orthogonal group $\O(n,\F_2)$ of $n\times n$ matrices $M$ with $M^\T = M^{-1}$.
\begin{tcolorbox}
  \begin{lemma}
    If $\tau \simeq (\tau_\star, \tau_\star)$ is a SWEL-preserving gate, then $\tau_\star^\T = \tau_\star^{-1}$.
    \label{lemma:swel_clifford}
  \end{lemma}
  \begin{proof}
    This follows directly from \cref{lemma:css_clifford} together with the fact that $\tau_\star = \tau_X = \tau_Z$.
  \end{proof}
\end{tcolorbox}
We can therefore produce a random SWEL-preserving gate $\tau \simeq (\tau_\star, \tau_\star)$ by sampling $\tau_\star \sim \O(n,\F_2)$.
This sampling is complete in the sense that it can produce every SWEL code.
\begin{tcolorbox}
  \begin{restatable}[SWEL construction]{theorem}{SWELConstruction}
    If $C = \CSS(H_\star, H_\star)$ is a SWEL code, then
    \begin{align*}
      C = \tau(\TrivSWEL(k, \rank(H_\star)))
    \end{align*}
    for some SWEL-preserving gate $\tau$.
    \label{thm:swel_construction}
  \end{restatable}
  We prove \cref{thm:swel_construction} in \cref{sec:swel_construction} by constructing $\tau$ from standard-form parity check matrices of $C$.
\end{tcolorbox}

In the case of CSS codes, we use the CNOT gate to move between codes for annealing.
The CNOT gate is the smallest CSS-preserving gate that is nontrivial in the sense that it is not equal to the identity gate modulo qubit permutations.
The CNOT gate is also the \emph{only} nontrivial two-qubit CSS-preserving gate.
The smallest nontrivial SWEL-preserving gate can be derived from first principles, or more simply found by a brute-force search over orthogonal matrices of increasing size.
This gate acts on four qubits, and is unique modulo qubit permutations.
For this reason, we call it the \emph{elementary} SWEL-preserving gate.
\begin{tcolorbox}
  \begin{definition}[Elementary SWEL-preserving gate]
    Let $n\ge 4$ be an integer, $e_q\in\F_2^n$ with $q\in\Z_n$ be a standard basis vector for $\F_2^n$, and for any subset $Q\subset\Z_n$ define
    \begin{align*}
      e_Q = \sum_{q\in Q} e_q
      &&
      \text{and}
      &&
      t(Q) = \1_n + e_Q e_Q^\T.
    \end{align*}
    An \emph{elementary SWEL-preserving gate} is any SWEL-preserving gate $T_Q \simeq (t(Q), t(Q))$ for which $Q\subset\Z_n$ is a subset of $\abs{Q}=4$ targets.
  \end{definition}
\end{tcolorbox}
In the case of $n=4$, the matrix $t(\Z_4)$ that induces the unique elementary SWEL-preserving gate is
\begin{align}
  t(\Z_4) =
  \begin{pmatrix}
    0 & 1 & 1 & 1 \\
    1 & 0 & 1 & 1 \\
    1 & 1 & 0 & 1 \\
    1 & 1 & 1 & 0
  \end{pmatrix}.
\end{align}
The corresponding gate is implemented by the circuit
\begin{align}
  \begin{quantikz}[row sep={1.5em,between origins},column sep={.5em}]
    & \ctrl{1} &          & \targ{}   &[-.25em]       &[-.25em] \targ{} &          & \ctrl{1}& \\
    & \targ{}  & \ctrl{1} &           &               &                 & \ctrl{1} & \targ{} & \\
    &          & \targ{}  & \ctrl{-2} & \permute{2,1} & \ctrl{-2}       & \targ{}  &         & \\
    &          &          &           &               &                 &          &         &
  \end{quantikz}
  =
  \begin{quantikz}[row sep={1.5em,between origins},column sep={.5em}]
    & \ctrl{1} & \targ{}   &[-.75em]  & \targ{}   &[-.75em]  &[-.25em] \ctrl{1} &[-.25em] \\
    & \targ{}  &           & \ctrl{1} &           & \ctrl{2} & \targ{}          & \\
    &          &           & \targ{}  & \ctrl{-2} &          & \permute{2,1}    & \\
    &          & \ctrl{-3} &          &           & \targ{}  &                  &
  \end{quantikz},
\end{align}
and is equal to the unitary
\begin{align}
  T_{\Z_4} = -e^{-i\frac{\pi}{4}} Y^{\otimes 4} S(Z^{\otimes 4}) S(Y^{\otimes 4}) S(X^{\otimes 4}),
\end{align}
where $S(P) = \exp(-i\frac{\pi}{4} P)$ is an $S$-like rotation about the Pauli string $P$.
For a fixed choice of parameters $\params{n,k}$, the use of the elementary SWEL-preserving gate to search for SWEL codes is \emph{complete} in the sense that it covers the entire space of $\params{n,k}$ SWEL codes.
\begin{tcolorbox}
  \begin{restatable}[SWEL equivalence]{theorem}{SWELEquivalence}
    If $C$ and $C'$ are both $\params{n,k}$ SWEL codes, then $C' = (\Pi T_N \cdots T_2 T_1)(C)$ for some permutation gate $\Pi$ and a finite sequence of elementary SWEL-preserving gates $(T_1, T_2, \cdots, T_N)$.
    \label{thm:swel_equivalence}
  \end{restatable}
  We prove \cref{thm:swel_equivalence} in \cref{sec:swel_equivalence} by reducing it to the decomposition of orthogonal matrices $U\in\O(n,\F_2)$ into transvections of the form $t(Q) = \1_n + e_Q e_Q^\T$ with $\abs{Q}\in\set{2,4}$, where $\abs{Q} = 2$ corresponds to a SWAP gate, and constructing an Gaussian-elimination-like algorithm to perform this decomposition.
\end{tcolorbox}
\cref{thm:swel_equivalence} implies that elementary SWEL-preserving gates suffice, in principle, to explore the entire space of SWEL codes.
In total, our ingredients for a SWEL code search are as follows:
\begin{enumerate}[label=(\Alph*)]
  \item \textbf{Initial code}:
    Choose two positive integers $(k, s)$ and a random SWEL-preserving gate $\tau \simeq (\tau_\star, \tau_\star)$, obtained as a product of $30n$ elementary SWEL-preserving gates, each on four qubits chosen at random.
    Initialize the code $\tau(\TrivSWEL(k, s))$.
  \item \textbf{Random move}:
    Choose four qubits at random, $Q\subset\Z_n$ with $\abs{Q} = 4$, and propose the move $C \to T_Q(C)$, where $T_Q$ is the elementary SWEL-preserving gate on qubits $Q$.
\end{enumerate}

\subsection{Energy function}
\label{sec:energy_function}

Once we know how to initialize a code and move between codes, we need to construct an energy function to minimize by simulated annealing.
We nominally wish to maximize the distance $d$ of a code $C$.
However, an energy function such as $E(C) = -d$ suffers from a discrete barren plateau: for almost any stabilizer code $C$, almost all Clifford gates that address a small number of qubits leave the code distance unchanged.
We therefore need an energy function that is sensitive to more granular data than code distance.

A candidate metric to consider is the logical error rate of a stabilizer code in a code-capacity model, in which each qubit is subject to an independent, identically distributed (IID) Pauli error with probability $p$.
Code-capacity logical error rates are typically estimated using Monte Carlo simulations.
However, Monte Carlo simulations would be prohibitively expensive to use at each step of a code search.
Moreover, the search process would suffer from statistical uncertainties in simulation results.

To develop a simple and stable surrogate for the logical error rate in a code-capacity model, we examine the probability that a qubitwise-IID error process results in a logical error for a code with distance $d$.
The lowest-weight error that may be decoded incorrectly by a perfect decoder is an error of weight $w_d = \left\lceil\frac{d+1}{2}\right\rceil$ with support contained \emph{within} that of a weight-$d$ logical operator\footnote{If $d$ is even, errors of weight $\nicefrac{d}{2}$ cause logical erasure.
These errors do not contribute to the error in, say, a quantity estimated from samples of a quantum circuit, because samples with weight $\nicefrac{d}{2}$ errors would be discarded.}.
Each weight-$d$ logical operator has ${d \choose w_d}$ candidate half-supports, where a half-support is a weight-$w_d$ subset of the full (weight-$d$) support.
If there are $m$ distinct minimum-weight nontrivial logical operators of $X$ type and $Z$ type, making the approximation that their half-supports are distinct yields a logical error rate of
\begin{align}
  p_L \approx m {d \choose w_d} p^{w_d}.
  \label{eq:error_estimate}
\end{align}
The distance $d$ and number $m$ of minimum-weight nontrivial $X$-type and $Z$-type logical operators can be computed exactly with the MacWilliams identities~\cite{macwilliams1977theory}, which require enumerating the $2^{\rank(H_X)} + 2^{\rank(H_Z)}$ stabilizers of $X$ and $Z$ type, or with the Brouwer-Zimmermann algorithm~\cite{hernando2019fast, hernando2026fast}, whose runtime grows with the distance of the code.
Our search uses whichever method is faster for a given code.
In general, both methods take time exponential in $n$.

The logical error rate in \cref{eq:error_estimate} can be extremely small, so we take its logarithm and rescale by $\log(1/p) = -\log(p)$ to arrive at the energy function
\begin{align}
  E_L(C)
  = \frac{\log p_L}{\log(1/p)}
  \approx \frac{\log m}{\log(1/p)} + \frac{\log {d \choose w_d}}{\log(1/p)} - w_d.
\end{align}
The first term, $\frac{\log m}{\log(1/p)}$, biases the annealer towards codes with fewer minimum-weight nontrivial logical operators.
This is the essential term that resolves codes with equal distance.
The remaining terms are constant in the space of codes with fixed $d$.
However, since $w_d$ is the same for an even distance $d$ and for $d+1$, these terms do not reward an increase from even to odd distance.
We therefore replace these terms by $-d$, which rewards every increase in distance, altogether arriving at the energy function
\begin{align}
  E(C) = \frac{\log m}{\log(1/p)} - d.
  \label{eq:energy_function}
\end{align}
We set $p = 10^{-3}$ to reflect realistic physical error rates.
The energy function in \cref{eq:energy_function} can alternatively be obtained by considering the probability of a spontaneous undetectable logical error, which is approximately $mp^d$.

\subsection{Adaptive annealing}
\label{sec:adaptive_annealing}

Here we describe our procedure for setting the temperature throughout an annealing run (see \cref{sec:simulated_annealing}).
The basic idea is to split an annealing run into $K$ \emph{epochs}, where each epoch consists of $N$ proposed moves, and set a \emph{target move acceptance rate} $r_t$ for epoch $t$.
We can then set the temperature $T_t$ in epoch $t$ to the value for which $r_t
= R_{\bm\Delta_{t-1}}(T_t)$, where $\bm\Delta_{t-1} = (\Delta_{t-1,1},\Delta_{t-1,2},\cdots,\Delta_{t-1,N})$ are the energy differences from the proposed moves in epoch $t-1$, and
\begin{align}
  R_{\bm\Delta}(T)
  = \mathop{\mathbb{E}}\limits_{\substack{\delta\in\bm\Delta}}[\AcceptProb(\delta,T)]
  \label{eq:temp_to_rate}
\end{align}
is the mean acceptance probability for proposed moves with energy differences $\bm\Delta$ at the temperature $T$.
As long as there is at least one positive value in $\bm\Delta$, the function $R_{\bm\Delta}(T)\in[0,1]$ is continuous and monotonically increasing with $T$, allowing $r = R_{\bm\Delta}(T)$ to be solved numerically by minimizing $(R_{\bm\Delta}(T) - r)^2$.
For the first epoch, we perform a ``dry run'' by proposing and rejecting $N$ moves to collect an initial vector of energy differences $\bm\Delta_0$.

In practice, we make two further modifications to the above procedure.
First, to prevent the annealer from freezing on plateaux or degenerate local minima, within which all moves are accepted, we set a target \emph{uphill} move acceptance rate, nominally considering the temperature
\begin{align}
  T_t^\star = \argmin_z \left[\left(R_{\bm\Delta_{t-1}^+}(z) - r_t\right)^2\right],
\end{align}
where $\bm\Delta^+ = (\delta\in\bm\Delta:\delta>0)$ are the positive values in $\bm\Delta$%
\footnote{%
  We restrict the minimization to $e^{-10} \le T_t^\star \le 1$.
  If $\bm\Delta_{t-1}^+$ is empty, as a fallback we set $T_t^\star = 1$.
}.
Second, to dampen temperature fluctuations that can occur, for example, when the annealer finds a deep basin, we introduce a \emph{memory factor} $\mu\in[0,1)$, and set the temperature $T_t$ in epoch $t$ according to
\begin{align}
  \log T_t = (1-\mu) \log T_t^\star + \mu \log T_{t-1}.
\end{align}
The final adaptive annealing algorithm is sketched in \cref{alg:adaptive_annealing}, where we set $\mu=0.5$.

\begin{algorithm*}
  \caption[]{%
    Adaptive annealing with target uphill acceptance rates and temperature memory.
    This algorithm uses the following definitions and conventions:
    \begin{enumerate}[label=(\alph*), nosep]
      \item $\mathcal{X}$ is the search space, such as the space of CSS or SWEL codes with fixed parameters.
      \item $\oper{get\_random\_neighbor}(x)$ samples a neighbor of $x$ from a predetermined proposal distribution.
      \item $\oper{get\_energy}(x) = E(x)$ for a predetermined energy function $E:\mathcal{X}\to\mathbb{R}$.
      \item An annealing run consists of $K$ epochs, where each epoch consists of $N$ proposed moves.
      \item Each epoch $t\in\Z_K$ has a target uphill acceptance rate $r_t$.
      \item $\oper{get\_temperature}(\bm\Delta,r) = \argmin_z (R_{\bm\Delta^+}(z) - r)^2$, where $\bm\Delta^+ = (\delta\in\bm\Delta:\delta>0)$ are the positive values in $\bm\Delta$, and $R_{\bm\Delta}(z)$ is defined in \cref{eq:temp_to_rate}.
      \item The memory factor $\mu \in [0,1)$ damps epoch-to-epoch temperature fluctuations.
      \item $\oper{rand}(a,b)$ returns a uniformly random real number from the interval $[a,b)$.
      \item $\AcceptProb(\Delta,T)$ is defined in \cref{eq:accept_prob}.
    \end{enumerate}
    In practice, we always set the number of iterations per epoch to $N = 10^4$, the memory factor to $\mu = 0.5$, and use a geometric schedule for the target acceptance rates, $\log_{10}\bm{r} = \oper{linspace}(-2, -4, K)$ (meaning $\log_{10}\bm{r}$ consists of $K$ values evenly spaced from $-2$ to $-4$), except for $n\le22$, for which the schedule ends at $-3$ instead of $-4$.
    We typically choose $K$ between $1000$ and $3000$.
  }
  \label{alg:adaptive_annealing}
  \DontPrintSemicolon
  \SetKwComment{Comment}{\# }{}
  \SetKwProg{Fn}{def}{\string:}{}
  \SetKwInOut{Input}{Input}
  \SetKwInOut{Output}{Output}
  \SetKwFor{For}{for}{:}{}
  \SetKwFor{While}{while}{:}{}
  \SetKwIF{If}{ElseIf}{Else}{if}{:}{elif}{else:}{}
  \BlankLine
  \Input{%
    (1) Initial state $x_0\in\mathcal{X}$ \\
    (2) Target uphill acceptance rates $(r_1, r_2, \ldots, r_K) \in (0, 1)^K$ \\
    (3) Number of iterations per epoch $N\in\Z_{>0}$ \\
    (4) Memory factor $\mu\in[0,1)$
  }
  \Output{Annealed state $x\in\mathcal{X}$}
  \BlankLine
  \SetKwFunction{Anneal}{anneal}
  \Fn{\Anneal{$x_0, \bm{r}, N, \mu$}}{
    \texttt{set} $x \leftarrow x_0$\;
    \texttt{set} $\bm\Delta \leftarrow \bm{0}_N$
    \Comment{all-zero vector}
    \Comment{dry run to set initial temperature}
    \For{$s \in \Z_N$}{
      \texttt{set} $y \leftarrow \oper{get\_random\_neighbor}(x)$\;
      \texttt{update} $\Delta_s \leftarrow \oper{get\_energy}(y) - \oper{get\_energy}(x)$\;
    }
    \texttt{set} $T \leftarrow \oper{get\_temperature}(\bm\Delta,r_1)$\;
    \Comment{main annealing loop}
    \For{$t \in \Z_K$}{
      \texttt{set} $T_\star \leftarrow \oper{get\_temperature}(\bm\Delta,r_t)$\;
      \texttt{update} $T \leftarrow \exp((1-\mu) \log T_\star + \mu \log T)$\;
      \For{$s \in \Z_N$}{
        \texttt{set} $y \leftarrow \oper{get\_random\_neighbor}(x)$\;
        \texttt{update} $\Delta_s \leftarrow \oper{get\_energy}(y) - \oper{get\_energy}(x)$\;
        \If{$\oper{rand}(0, 1) < \AcceptProb(\Delta_s,T)$}{
          \texttt{update} $x \leftarrow y$\;
        }
      }
    }
    \Return{$x$}\;
  }
\end{algorithm*}

\subsection{Reducing stabilizer weights}
\label{sec:reducing_weight}

Annealing with the energy function constructed in \cref{sec:energy_function} optimizes distance and number of minimum-weight logical operators of a code, but does not consider the weights of the stabilizers used to represent the code.
Finding low-weight stabilizer generators is desirable to reduce the complexity of syndrome measurements and decoding.
We therefore find stabilizer generators of minimum weight by brute force.
For a CSS code $\CSS(H_X, H_Z)$, we treat the $X$-type and $Z$-type stabilizer subgroups separately: we enumerate all elements of a subgroup, and take elements in order of increasing weight whenever they are independent of the elements already taken.
The resulting generators have the least possible total weight, and enumerating the at most $2^{22}$ elements of each subgroup is inexpensive for the codes in this work.

\section{Results}
\label{sec:results}

The CSS code search described in \cref{sec:search_css} takes as input a choice of logical qubits number $k$ and numbers $(s_X,s_Z)$ of $X$-type and $Z$-type stabilizer generators.
We set $s_X = s_Z = s$ for simplicity, parameterizing each CSS code search by a choice of $\params{n,k}$ for which $n-k = 2s$ is even.
Each SWEL code search can similarly be parameterized by a choice of $\params{n,k}$ with even $n-k = 2s$.
We warm-start every annealing run with the lowest-energy code of $3\times10^4$ randomly sampled SWEL codes, which each run samples independently.
For each $\params{n,k}$ code, we perform 100 independent annealing runs if $n<48$, and 50 if $n\ge48$; other settings, such as the number of epochs of a run, also depend on $\params{n,k}$, and are given with the code search in the Python library \texttt{qLDPC}~\cite{perlin2023qldpc}.
Each annealing run saves both the lowest-energy code and the code of greatest distance that it encounters.

Altogether, the best codes that we found by annealing are provided in \cref{sec:best_codes}, and summarized in \cref{tab:codes}.
In addition to the code parameters $\params{n,k,d}$, \cref{tab:codes} provides the following:
\begin{itemize}
  \item The number of minimum-weight nontrivial logical operators.
  \item The mean and max weight of the stabilizer generators after minimization according to the procedure discussed in \cref{sec:reducing_weight}.
  \item The excess $\delta_{(\mathrm{QGV},\mathrm{CSS})}(n,k) = d - d_{(\mathrm{QGV,CSS})}(n,k)$ of the distance $d$ beyond the stabilizer ($d_{\mathrm{QGV}}$) or CSS ($d_{\mathrm{CSS}}$) Gilbert-Varshamov bound for parameters $\params{n,k,d}$.
\end{itemize}
The distance bounds provide a theoretical yardstick by which to measure ``how good'' of a distance $d$ was obtained for any fixed choice $\params{n,k}$.
We remark that the stabilizer Gilbert-Varshamov bound that we compare against applies to \emph{arbitrary} Pauli stabilizer codes, including non-CSS codes, and the existence of an $\params{n,k,d}$ CSS code with $d \ge d_{\mathrm{QGV}}$ ($\delta_{\mathrm{QGV}}>0$) is generally not guaranteed.

Finally, \cref{tab:codes} compares the codes found in this work with the best code parameters indexed on \texttt{qecdb.org}.
Specifically, for each $\params{n,k,d}$ code in \cref{tab:codes_css}, we provide the smallest values of $n_\mathrm{best}^{\mathrm{qecdb}}$ and largest values of $k_\mathrm{best}^{\mathrm{qecdb}}$ and $d_\mathrm{best}^{\mathrm{qecdb}}$ for which \texttt{qecdb.org} indexed an $\params{n_\mathrm{best}^{\mathrm{qecdb}}, \ge k, \ge d}$, $\params{\le n,k_\mathrm{best}^{\mathrm{qecdb}},\ge d}$, or $\params{\le n,\ge k,d_\mathrm{best}^{\mathrm{qecdb}}}$ CSS code as of 2026-11-14.
We similarly compare the SWEL codes found in this work to the best self-dual code parameters on \texttt{qecdb.org} in \cref{tab:codes_swel}, though we did not verify whether the corresponding codes on \texttt{qecdb.org} are SWEL.

We remark that some of the codes in \cref{tab:codes} have only a few minimum-weight nontrivial logical operators (indicated by bold text).
Reducing this number to zero increases the distance of a code.
These codes are therefore enticing targets for further improvement.
We note that the $\params{22,4,4}$ SWEL code in \cref{tab:codes_swel} is reducible to $\params{20,4,4} \otimes \TrivSWEL(0,1)$; its 10th and 12th qubits are addressed only by their own $XX$ and $ZZ$ stabilizers (see \cref{sec:best_codes_swel}).

\begin{sidewaystable*}
  \footnotesize
  \caption{%
    CSS and SWEL codes found with adaptive annealing.
    $\delta_{(\mathrm{QGV},\mathrm{CSS})} = d - d_{(\mathrm{QGV},\mathrm{CSS})}$ is the amount by which the distance $d$ exceeds a stabilizer ($d_{\mathrm{QGV}}$) or CSS ($d_{\mathrm{CSS}}$) Gilbert-Varshamov bound; we leave dots ($\cdot$) in place of zeros for legibility.
    Minimum-weight nontrivial logical operator counts are split into $X$ / $Z$ type.
    Mean and maximum stabilizer generator weights are obtained by the minimization procedure discussed in \cref{sec:reducing_weight}.
    $n_\mathrm{best}^{\mathrm{qecdb}}$, $k_\mathrm{best}^{\mathrm{qecdb}}$, and $d_\mathrm{best}^{\mathrm{qecdb}}$ are the best parameters for which an $\params{n_\mathrm{best}^{\mathrm{qecdb}},\ge k,\ge d}$, $\params{\le n,k_\mathrm{best}^{\mathrm{qecdb}},\ge d}$, or $\params{\le n,\ge k,d_\mathrm{best}^{\mathrm{qecdb}}}$ code was indexed on \texttt{qecdb.org} (which, we note, does not imply the existence of an $\params{n_\mathrm{best}^{\mathrm{qecdb}},k_\mathrm{best}^{\mathrm{qecdb}},d_\mathrm{best}^{\mathrm{qecdb}}}$ code, as the best value of each parameter may be achieved by a different code).
    SWEL codes found in this work are compared to self-dual codes on \texttt{qecdb.org}, which is a broader class of codes that contains non-SWEL codes.
    Improvements over the codes indexed on \texttt{qecdb.org} are indicated by green (this work) and red (\texttt{qecdb.org}) highlighted cells.
    Codes are compared with \texttt{qecdb.org} as of 2026-11-14.
  }
  \setlength{\tabcolsep}{4pt}
  ~\hfill
  \begin{subtable}[t]{0.48\textwidth}
    \centering
    \caption{CSS}
    \begingroup
\providecommand{\hlred}[1]{\cellcolor{red!25}#1}%
\providecommand{\hlgreen}[1]{\cellcolor{green!25}#1}%
\makeatletter
\col@sep\tabcolsep
\makeatother
\begin{tabular}{c|c|c||c|c||c|c|c||c|c|c}
  $n$ & $k$ & $d$ & $\delta_{\mathrm{QGV}}$ & $\delta_{\mathrm{CSS}}$ & \thead{min-weight\\nontrivial\\logicals} & \thead{mean\\stab.\\weight} & \thead{max\\stab.\\weight} & $n_\mathrm{best}^{\mathrm{qecdb}}$ & $k_\mathrm{best}^{\mathrm{qecdb}}$ & $d_\mathrm{best}^{\mathrm{qecdb}}$ \\ \hline\hline
  20           & 4            & 4           & $\cdot$ & 2 & \textbf{5 / 5} &  8    &  8 & 15         & 8          & 4         \\ \hline
  20           & 6            & 4           & 1       & 2 & 36 / 36        &  8.14 &  9 & 16         & 8          & 4         \\ \hline
  \hlgreen{22} & \hlgreen{4}  & \hlgreen{5} & 1       & 3 & 63 / 63        &  6.78 &  7 & \hlred{24} & \hlred{3}  & \hlred{4} \\ \hline
  22           & 6            & 4           & $\cdot$ & 2 & 15 / 14        &  7.75 &  8 & 16         & 8          & 4         \\ \hline
  24           & 4            & 5           & $\cdot$ & 3 & 43 / 19        &  6.75 &  7 & 24         & 4          & 5         \\ \hline
  24           & 6            & 4           & $\cdot$ & 2 & \textbf{2 / 2} &  8    &  9 & 16         & 10         & 4         \\ \hline
  \hlgreen{25} & \hlgreen{5}  & \hlgreen{5} & 1       & 3 & 56 / 36        &  7.40 &  8 & \hlred{28} & \hlred{4}  & \hlred{4} \\ \hline
  \hlgreen{26} & \hlgreen{4}  & \hlgreen{6} & 1       & 3 & 259 / 254      &  7.73 &  8 & \hlred{28} & \hlred{2}  & \hlred{5} \\ \hline
  \hlgreen{26} & \hlgreen{6}  & \hlgreen{5} & 1       & 3 & 68 / 63        &  7.85 &  8 & \hlred{28} & \hlred{4}  & \hlred{4} \\ \hline
  30           & 4            & 6           & $\cdot$ & 3 & 81 / 80        &  7.35 &  8 & 28         & 8          & 6         \\ \hline
  30           & 6            & 6           & 1       & 3 & 307 / 307      &  8    &  8 & 30         & 8          & 6         \\ \hline
  32           & 6            & 6           & $\cdot$ & 3 & 134 / 222      &  7.96 &  8 & 30         & 8          & 6         \\ \hline
  32           & 8            & 6           & 1       & 3 & 484 / 482      &  9.25 & 10 & 30         & 8          & 6         \\ \hline
  36           & 6            & 6           & $\cdot$ & 2 & 21 / 28        &  8.77 &  9 & 30         & 8          & 6         \\ \hline
  36           & 8            & 6           & $\cdot$ & 3 & 127 / 122      &  9.04 & 10 & 30         & 8          & 6         \\ \hline
  38           & 8            & 6           & $\cdot$ & 3 & 49 / 45        &  9.70 & 10 & 30         & 8          & 6         \\ \hline
  \hlgreen{38} & \hlgreen{10} & \hlgreen{6} & $\cdot$ & 3 & 295 / 219      &  9.57 & 10 & \hlred{42} & \hlred{8}  & \hlred{5} \\ \hline
  \hlgreen{38} & \hlgreen{12} & \hlgreen{6} & 1       & 3 & 708 / 692      & 10.69 & 12 & \hlred{42} & \hlred{8}  & \hlred{4} \\ \hline
  \hlgreen{40} & \hlgreen{10} & \hlgreen{6} & $\cdot$ & 3 & 77 / 80        & 10.33 & 11 & \hlred{42} & \hlred{8}  & \hlred{5} \\ \hline
  \hlgreen{40} & \hlgreen{12} & \hlgreen{6} & $\cdot$ & 3 & 459 / 455      & 10.64 & 12 & \hlred{42} & \hlred{8}  & \hlred{4} \\ \hline
  42           & 8            & 6           & $-1$    & 2 & \textbf{3 / 3} &  9.79 & 10 & 30         & 16         & 6         \\ \hline
  42           & 10           & 6           & $-1$    & 2 & 30 / 32        & 10.44 & 11 & 42         & 16         & 6         \\ \hline
  42           & 12           & 6           & $\cdot$ & 3 & 139 / 139      & 11    & 12 & 42         & 16         & 6         \\ \hline
  \hlgreen{44} & \hlgreen{8}  & \hlgreen{7} & $\cdot$ & 3 & 129 / 123      &  9.94 & 11 & \hlred{47} & \hlred{6}  & \hlred{6} \\ \hline
  48           & 6            & 8           & $-1$    & 3 & 210 / 299      &  9.64 & 10 & 48         & 10         & 8         \\ \hline
  48           & 8            & 8           & $\cdot$ & 3 & 767 / 735      & 10    & 12 & 48         & 10         & 8         \\ \hline
  50           & 6            & 8           & $-1$    & 3 & 83 / 71        & 10    & 11 & 48         & 10         & 8         \\ \hline
  50           & 8            & 8           & $-1$    & 3 & 477 / 497      & 10.14 & 12 & 48         & 10         & 8         \\ \hline
  50           & 10           & 7           & $-1$    & 2 & 30 / 0         & 11.15 & 12 & 48         & 10         & 8         \\ \hline
  52           & 10           & 8           & $\cdot$ & 3 & 677 / 713      & 11.43 & 12 & 48         & 10         & 8         \\ \hline
  \hlgreen{52} & \hlgreen{12} & \hlgreen{7} & $-1$    & 3 & 0 / 79         & 11.93 & 12 & \hlred{55} & \hlred{10} & \hlred{6} \\ \hline
  56           & 12           & 8           & $-1$    & 3 & 647 / 653      & 11.86 & 12 & 56         & 16         & 8         \\ \hline
  56           & 14           & 7           & $-1$    & 2 & 0 / 45         & 12.38 & 14 & 55         & 16         & 8
\end{tabular}%
\endgroup%

    \label{tab:codes_css}
  \end{subtable}
  \hfill
  \begin{subtable}[t]{0.48\textwidth}
    \centering
    \caption{SWEL}
    \begingroup
\providecommand{\hlred}[1]{\cellcolor{red!25}#1}%
\providecommand{\hlgreen}[1]{\cellcolor{green!25}#1}%
\makeatletter
\col@sep\tabcolsep
\makeatother
\begin{tabular}{c|c|c||c|c||c|c|c||c|c|c}
  $n$ & $k$ & $d$ & $\delta_{\mathrm{QGV}}$ & $\delta_{\mathrm{CSS}}$ & \thead{min-weight\\nontrivial\\logicals} & \thead{mean\\stab.\\weight} & \thead{max\\stab.\\weight} & $n_\mathrm{best}^{\mathrm{qecdb}}$ & $k_\mathrm{best}^{\mathrm{qecdb}}$ & $d_\mathrm{best}^{\mathrm{qecdb}}$ \\ \hline\hline
  20           & 4            & 4           & $\cdot$ & 2 & \textbf{5 / 5} &  8    &  8 & 16         & 8          & 4         \\ \hline
  20           & 6            & 4           & 1       & 2 & 53 / 53        &  7.14 &  8 & 16         & 8          & 4         \\ \hline
  22           & 4            & 4           & $\cdot$ & 2 & \textbf{5 / 5} &  7.33 &  8 & 16         & 8          & 4         \\ \hline
  22           & 6            & 4           & $\cdot$ & 2 & 24 / 24        &  7.50 &  8 & 16         & 8          & 4         \\ \hline
  24           & 4            & 4           & $-1$    & 2 & \textbf{1 / 1} &  7.20 &  8 & 16         & 8          & 4         \\ \hline
  24           & 6            & 4           & $\cdot$ & 2 & 10 / 10        &  8.22 & 10 & 16         & 8          & 4         \\ \hline
  25           & 5            & 4           & $\cdot$ & 2 & \textbf{3 / 3} &  7.20 &  8 & 16         & 8          & 4         \\ \hline
  \hlgreen{26} & \hlgreen{4}  & \hlgreen{5} & $\cdot$ & 2 & 32 / 32        &  7.27 &  8 & \hlred{28} & \hlred{3}  & \hlred{4} \\ \hline
  26           & 6            & 4           & $\cdot$ & 2 & \textbf{3 / 3} &  8.20 & 10 & 16         & 8          & 4         \\ \hline
  30           & 4            & 5           & $-1$    & 2 & \textbf{4 / 4} &  8.15 & 10 & 28         & 4          & 5         \\ \hline
  \hlgreen{30} & \hlgreen{6}  & \hlgreen{5} & $\cdot$ & 2 & 28 / 28        &  8.17 & 10 & \hlred{31} & \hlred{4}  & \hlred{4} \\ \hline
  32           & 6            & 5           & $-1$    & 2 & 8 / 8          &  8.15 & 10 & 31         & 11         & 5         \\ \hline
  32           & 8            & 5           & $\cdot$ & 2 & 32 / 32        &  9.17 & 10 & 31         & 11         & 5         \\ \hline
  36           & 6            & 6           & $\cdot$ & 2 & 42 / 42        &  8.80 & 10 & 36         & 6          & 6         \\ \hline
  \hlgreen{36} & \hlgreen{8}  & \hlgreen{6} & $\cdot$ & 3 & 126 / 126      &  9.43 & 10 & \hlred{56} & \hlred{6}  & \hlred{5} \\ \hline
  \hlgreen{38} & \hlgreen{8}  & \hlgreen{6} & $\cdot$ & 3 & 61 / 61        &  9.87 & 10 & \hlred{56} & \hlred{6}  & \hlred{5} \\ \hline
  \hlgreen{38} & \hlgreen{10} & \hlgreen{6} & $\cdot$ & 3 & 229 / 229      &  9.86 & 10 & \hlred{56} & \hlred{6}  & \hlred{5} \\ \hline
  \hlgreen{38} & \hlgreen{12} & \hlgreen{5} & $\cdot$ & 2 & 24 / 24        & 10.31 & 12 & \hlred{45} & \hlred{11} & \hlred{4} \\ \hline
  \hlgreen{40} & \hlgreen{10} & \hlgreen{6} & $\cdot$ & 3 & 104 / 104      & 10    & 10 & \hlred{56} & \hlred{6}  & \hlred{5} \\ \hline
  \hlgreen{40} & \hlgreen{12} & \hlgreen{5} & $-1$    & 2 & \textbf{4 / 4} & 10.71 & 12 & \hlred{45} & \hlred{11} & \hlred{4} \\ \hline
  \hlgreen{42} & \hlgreen{8}  & \hlgreen{6} & $-1$    & 2 & 8 / 8          & 10    & 10 & \hlred{56} & \hlred{6}  & \hlred{5} \\ \hline
  \hlgreen{42} & \hlgreen{10} & \hlgreen{6} & $-1$    & 2 & 42 / 42        & 10.12 & 12 & \hlred{56} & \hlred{6}  & \hlred{5} \\ \hline
  \hlgreen{42} & \hlgreen{12} & \hlgreen{6} & $\cdot$ & 3 & 152 / 152      & 11.73 & 12 & \hlred{56} & \hlred{6}  & \hlred{4} \\ \hline
  \hlgreen{44} & \hlgreen{8}  & \hlgreen{7} & $\cdot$ & 3 & 140 / 140      & 10    & 10 & \hlred{56} & \hlred{2}  & \hlred{5} \\ \hline
  \hlgreen{48} & \hlgreen{6}  & \hlgreen{7} & $-2$    & 2 & \textbf{3 / 3} &  9.81 & 10 & \hlred{56} & \hlred{2}  & \hlred{6} \\ \hline
  \hlgreen{48} & \hlgreen{8}  & \hlgreen{7} & $-1$    & 2 & 30 / 30        &  9.90 & 10 & \hlred{56} & \hlred{2}  & \hlred{5} \\ \hline
  \hlgreen{50} & \hlgreen{6}  & \hlgreen{8} & $-1$    & 3 & 105 / 105      & 10    & 10 & \hlred{56} & \hlred{2}  & \hlred{6} \\ \hline
  \hlgreen{50} & \hlgreen{8}  & \hlgreen{7} & $-2$    & 2 & 7 / 7          & 10.57 & 12 & \hlred{56} & \hlred{2}  & \hlred{5} \\ \hline
  \hlgreen{50} & \hlgreen{10} & \hlgreen{7} & $-1$    & 2 & 56 / 56        & 11.50 & 12 & \hlred{56} & \hlred{2}  & \hlred{5} \\ \hline
  \hlgreen{52} & \hlgreen{10} & \hlgreen{7} & $-1$    & 2 & 18 / 18        & 11.71 & 12 & \hlred{56} & \hlred{2}  & \hlred{5} \\ \hline
  \hlgreen{52} & \hlgreen{12} & \hlgreen{7} & $-1$    & 3 & 92 / 92        & 11.70 & 12 & \hlred{56} & \hlred{2}  & \hlred{5} \\ \hline
  56           & 12           & 7           & $-2$    & 2 & 10 / 10        & 11.91 & 12 & 56         & 16         & 8         \\ \hline
  56           & 14           & 7           & $-1$    & 2 & 55 / 55        & 12.86 & 14 & 56         & 16         & 8
\end{tabular}%
\endgroup%

    \label{tab:codes_swel}
  \end{subtable}
  \hfill~
  \label{tab:codes}
\end{sidewaystable*}

\begin{figure*}[p]
  \centering
  \begin{subfigure}{\textwidth}
    \centering
    \includegraphics{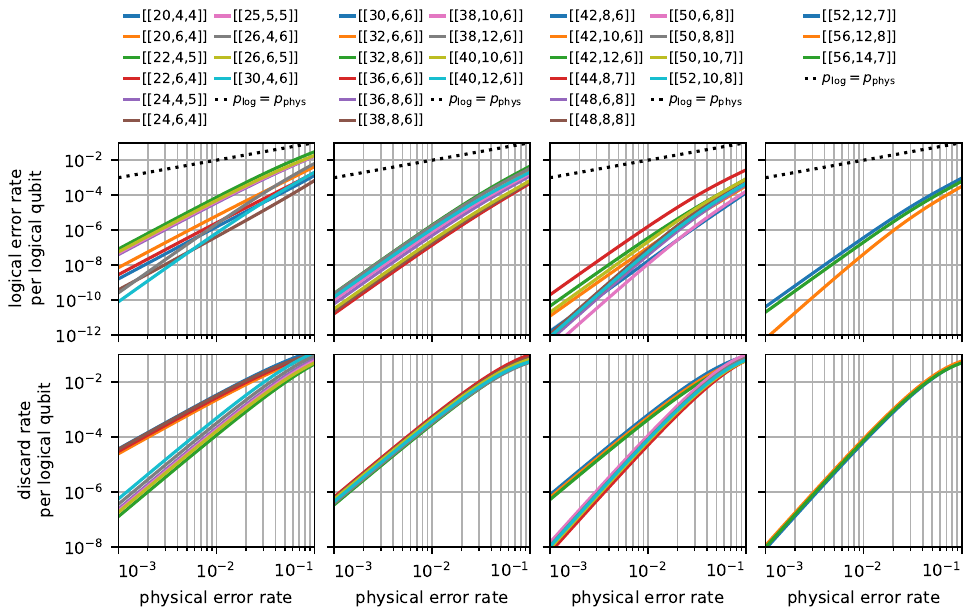}
    \caption{CSS}
  \end{subfigure}
  \\[.5em]
  \begin{subfigure}{\textwidth}
    \centering
    \includegraphics{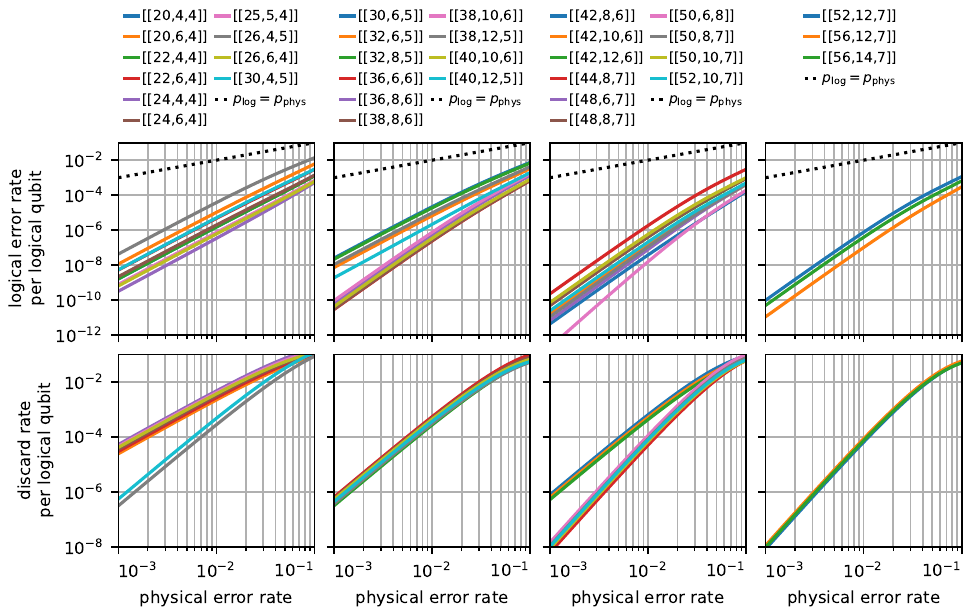}
    \caption{SWEL}
  \end{subfigure}%
  \caption{%
    Logical error rates (LER) and discard rates (DR) per logical qubit (LQ) for the codes in \cref{tab:codes} in a code-capacity noise model with a lookup table decoder.
    The decoder declares failure (triggering a discard) if it observes a syndrome that cannot be caused by any error of weight $w<\nicefrac{d}{2}$.
    The physical error rate corresponds to the IID probability that each qubit is subject to a randomly chosen $X$, $Y$, or $Z$ error.
  }
  \label{fig:error_rates}
\end{figure*}

\begin{figure*}
  \centering
  \includegraphics{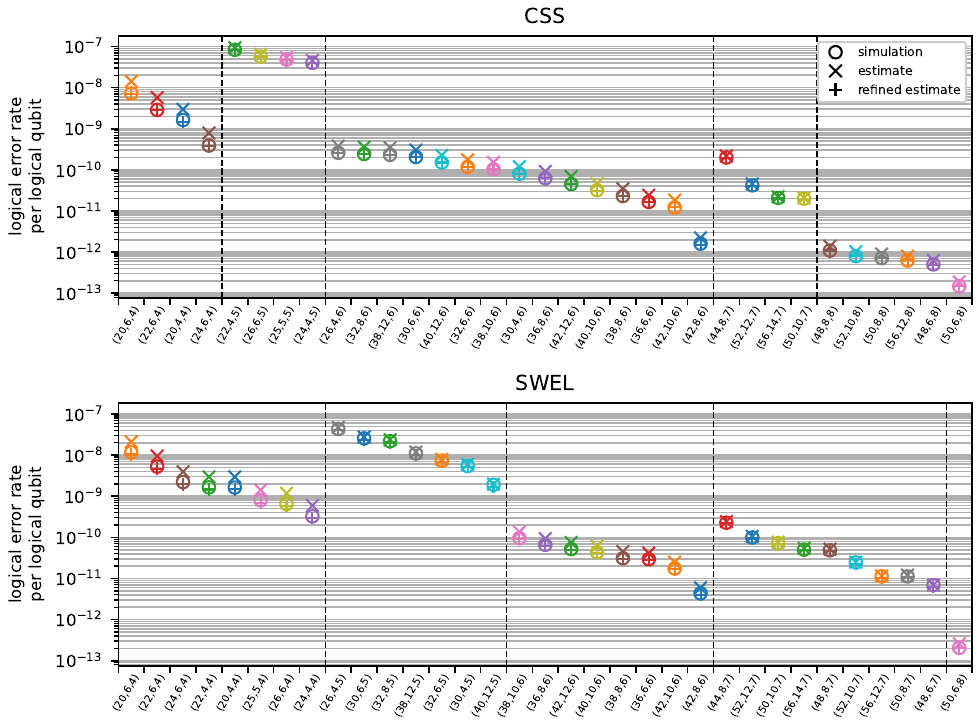}
  \caption{%
    Logical error rates from \cref{fig:error_rates} and corresponding estimates computed with \cref{eq:error_estimate} and \cref{eq:error_estimate_refined} at a physical error rate of $p_{\text{dep}}=10^{-3}$.
    Codes are sorted first by distance, and then by logical error rate.
    Colors identify the code size $\params{n,k}$ as in \cref{fig:error_rates}.
  }
  \label{fig:error_estimates}
\end{figure*}

\cref{fig:error_rates} shows the logical error rates of all codes in \cref{tab:codes} in a code-capacity model with qubitwise-IID depolarizing errors.
We compute code-capacity logical error rates with the importance sampling Monte Carlo method implemented in the Python library \texttt{qLDPC}~\cite{perlin2023qldpc}, sampling $10^7$ errors per code.
Errors are decoded with a lookup table that is built by considering all errors of integer weight $w < \nicefrac{d}{2}$, and identifying each observed syndrome with the lowest-weight error that induces it.
If the decoder encounters a syndrome that is not in the lookup table, it declares failure, which contributes to the discard rates in the bottom panels of \cref{fig:error_rates} rather than the logical error rates on top.
To meaningfully compare codes with different numbers of logical qubits, logical error rates and discard rates are divided by the number of logical qubits in a code block.
We remark that the discard rates in \cref{fig:error_rates} generally exceed the logical error rates.
For odd-distance codes, it may be possible to mitigate discard rates at the cost of increased logical error rates by using a decoder that is less eager to declare failure.
For even-distance codes, however, this discard rate reflects the unavoidable $O(p^{d/2})$ probability of uncorrectable errors that otherwise become logical errors with probability $\nicefrac{1}{2}$.

To validate our choice of annealing energy function, derived from the estimate of logical error rates in \cref{eq:error_estimate}, \cref{fig:error_estimates} compares the logical error rates in \cref{fig:error_rates} to their corresponding estimates at a physical error rate of $p_{\text{dep}}=10^{-3}$.
We set $p = \frac23 p_{\text{dep}}$ in the estimate to reflect the qubitwise-IID probability of an $X$ or $Z$ error.
To ensure that this estimate cannot be significantly improved upon by more careful counting arguments, we also plot a refined estimate, which we derive in \cref{sec:refined_estimate}.
The refined estimate (a) counts only \emph{distinct} half-supports, rather than assuming that all half-supports are distinct, and (b) accounts for $Y$ errors that would cause a sample to be discarded instead of being decoded incorrectly.
The refined estimate agrees with the simulated logical error rates to within $12\%$, whereas the simpler estimate in \cref{eq:error_estimate} exceeds them by up to a factor of $2$.
Accounting for erasure multiplies an estimate by a factor that depends only on the code distance $d$, so it does not affect a comparison between codes of equal distance.

Altogether, \cref{fig:error_estimates} shows that within each distance $d$, where codes are sorted by decreasing simulated logical error rate, the estimate in \cref{eq:error_estimate} decreases as well, with few exceptions.
This behavior validates using the estimate in \cref{eq:error_estimate} to disambiguate codes with equal distance.

\section{Discussion}
\label{sec:discussion}

In this work, we introduced an adaptive simulated annealing algorithm for finding moderate-sized QEC codes with favorable code parameters, with an eye towards near-term implementation on quantum hardware.
While the codes found by the search algorithm have high encoding rates and distances, further work is required to find codes suitable for use at commercial scale.
At face value, our search algorithm relies on subroutines whose runtime is exponential in the block length $n$, which presents challenges for scaling up (each annealing run in this work took between ten minutes and thirteen hours on a single CPU core).
Nonetheless, it may be fruitful to push the search to larger codes (we did not attempt block lengths greater than $n=56$), and there are several avenues for optimization to scale up the search algorithm.
Potential avenues for improvement include: parallelizing the evaluation of the energy function (which is the computational bottleneck), better annealing techniques (such as parallel tempering, or an improved adaptive schedule), and more sophisticated search techniques based on machine learning or artificial intelligence.
There may also be room for improvement by borrowing techniques from analogous efforts to find good classical error-correcting codes with simulated annealing~\cite{greenfield1993simulated, kanchi2026tunnelingaugmented}.

Any effort to apply adaptive annealing to larger codes will require faster evaluation of the energy function.
Our search already evaluates the energy function exactly with methods that are faster than brute-force enumeration (see \cref{sec:energy_function}).
However, their runtime grows exponentially with $n$, all the more so because our energy function counts minimum-weight nontrivial logical operators to avoid barren plateaux, rather than only computing code distance.
Scaling up the search therefore also requires faster counting of minimum-weight nontrivial logical operators (e.g., with parallelization), or a cheaper (possibly randomized or approximate) energy function to resolve codes of equal distance, as in Ref.~\cite{freire2025optimizing}.
It may also be possible to extend the applicability of adaptive annealing by imposing symmetries on the initial codes and move sets.
Imposing symmetries should, in turn, make the evaluation of the energy function cheaper by roughly the size of the corresponding symmetry group.
We remark that there exist many high-rate and high-distance code families with large symmetry groups, such as cyclic and Reed-Muller codes.

The final, and perhaps simplest, option to produce commercially applicable codes from the codes found here is through code concatenation.
A two-level scheme with the bottom level consisting of $\params{8,\, 6,\, 2}$ Iceberg codes~\cite{self2024protecting}, with $\params{38,\, 12,\, 6}$ codes on the top level would result in code blocks with parameters $\params{304,\, 72,\, 12}$, permitting hundreds of logical qubits to be stored on a device with thousands of physical qubits, at a distance likely sufficient to suppress errors at a commercially relevant scale.

Going beyond the use of error-correcting codes for quantum memory requires finding low-overhead logical gate sets.
Transversal logical Clifford gates of a stabilizer code can be derived from its automorphisms~\cite{sayginel2025faulttolerant}.
Though it may seem unlikely for numerically optimized codes to possess a large automorphism group, a non-exhaustive AI-assisted search for matrix automorphisms of small codes in \cref{tab:codes} found, for example, that the $\params{20,6,4}$ SWEL code has at least 2304 SWAP-transversal logical Clifford gates (after modding out by logical Pauli gates; see \cref{sec:automorphisms}).
It may be fruitful to fold the consideration of automorphisms into a code search by, for example, enforcing suitable symmetries on the search space, similarly to the case of the SWEL code search.
In the particular case of SWEL codes, there are additional avenues to construct fault-tolerant logical gate sets~\cite{tansuwannont2025clifford}.

For a fixed code, targeted logical gates may additionally be constructed with flag circuits, as in Ref.~\cite{dasu2026flagging}.
Logical gates can also be incorporated into Knill error correction~\cite{sullivan2026injection}, or performed on ancilla blocks as part of a von-Neumann-like architecture~\cite{brun2015teleportationbased}.
Both of these schemes require the fault-tolerant preparation of CSS states, a task for which automated general-purpose methods are currently under active development~\cite{forlivesi2025flag, peham2025automated, criger2026automated}.
However, it is likely that the lowest-overhead architecture will be code-dependent.

\section{Data and code availability}
The code search and scripts that reproduce our results are available in \texttt{experiments/code\_annealing} of \texttt{qLDPC}~\cite{perlin2023qldpc}, as of \texttt{qldpc==0.5.0}.

\section{AI usage statement}
We used AI coding assistants to write, test, and review the implementation of the code search and the scripts that reproduce our results in \texttt{qLDPC}, and to search for the automorphisms in \cref{sec:automorphisms}.

\section{Acknowledgments}
We thank Matt DeCross and Selwyn Simsek for their comments on early drafts of this manuscript, as well as Ruslan Shaydulin and Rob Otter for executive support of this work.

\bibliography{main.bib}

\begin{thebibliography}{65}%
\makeatletter
\providecommand \@ifxundefined [1]{%
 \@ifx{#1\undefined}
}%
\providecommand \@ifnum [1]{%
 \ifnum #1\expandafter \@firstoftwo
 \else \expandafter \@secondoftwo
 \fi
}%
\providecommand \@ifx [1]{%
 \ifx #1\expandafter \@firstoftwo
 \else \expandafter \@secondoftwo
 \fi
}%
\providecommand \natexlab [1]{#1}%
\providecommand \enquote  [1]{``#1''}%
\providecommand \bibnamefont  [1]{#1}%
\providecommand \bibfnamefont [1]{#1}%
\providecommand \citenamefont [1]{#1}%
\providecommand \href@noop [0]{\@secondoftwo}%
\providecommand \href [0]{\begingroup \@sanitize@url \@href}%
\providecommand \@href[1]{\@@startlink{#1}\@@href}%
\providecommand \@@href[1]{\endgroup#1\@@endlink}%
\providecommand \@sanitize@url [0]{\catcode `\\12\catcode `\$12\catcode
  `\&12\catcode `\#12\catcode `\^12\catcode `\_12\catcode `\%12\relax}%
\providecommand \@@startlink[1]{}%
\providecommand \@@endlink[0]{}%
\providecommand \url  [0]{\begingroup\@sanitize@url \@url }%
\providecommand \@url [1]{\endgroup\@href {#1}{\urlprefix }}%
\providecommand \urlprefix  [0]{URL }%
\providecommand \Eprint [0]{\href }%
\providecommand \doibase [0]{https://doi.org/}%
\providecommand \selectlanguage [0]{\@gobble}%
\providecommand \bibinfo  [0]{\@secondoftwo}%
\providecommand \bibfield  [0]{\@secondoftwo}%
\providecommand \translation [1]{[#1]}%
\providecommand \BibitemOpen [0]{}%
\providecommand \bibitemStop [0]{}%
\providecommand \bibitemNoStop [0]{.\EOS\space}%
\providecommand \EOS [0]{\spacefactor3000\relax}%
\providecommand \BibitemShut  [1]{\csname bibitem#1\endcsname}%
\let\auto@bib@innerbib\@empty
\bibitem [{\citenamefont {Lidar}\ and\ \citenamefont
  {Brun}(2013)}]{lidar2013quantum}%
  \BibitemOpen
  \bibfield  {author} {\bibinfo {author} {\bibfnamefont {D.~A.}\ \bibnamefont
  {Lidar}}\ and\ \bibinfo {author} {\bibfnamefont {T.~A.}\ \bibnamefont
  {Brun}},\ }\href@noop {} {\emph {\bibinfo {title} {Quantum {{Error
  Correction}}}}}\ (\bibinfo  {publisher} {Cambridge University Press},\
  \bibinfo {year} {2013})\BibitemShut {NoStop}%
\bibitem [{\citenamefont {Nielsen}\ and\ \citenamefont
  {Chuang}(2010)}]{nielsen2010quantum}%
  \BibitemOpen
  \bibfield  {author} {\bibinfo {author} {\bibfnamefont {M.~A.}\ \bibnamefont
  {Nielsen}}\ and\ \bibinfo {author} {\bibfnamefont {I.~L.}\ \bibnamefont
  {Chuang}},\ }\href@noop {} {\emph {\bibinfo {title} {Quantum {{Computation}}
  and {{Quantum Information}}: 10th {{Anniversary Edition}}}}}\ (\bibinfo
  {publisher} {Cambridge University Press},\ \bibinfo {year}
  {2010})\BibitemShut {NoStop}%
\bibitem [{\citenamefont {Steane}(2003)}]{steane2003overhead}%
  \BibitemOpen
  \bibfield  {author} {\bibinfo {author} {\bibfnamefont {A.~M.}\ \bibnamefont
  {Steane}},\ }\href {https://doi.org/10.1103/PhysRevA.68.042322} {\bibfield
  {journal} {\bibinfo  {journal} {Physical Review A}\ }\textbf {\bibinfo
  {volume} {68}},\ \bibinfo {pages} {042322} (\bibinfo {year}
  {2003})}\BibitemShut {NoStop}%
\bibitem [{\citenamefont {Preskill}(2018)}]{preskill2018quantum}%
  \BibitemOpen
  \bibfield  {author} {\bibinfo {author} {\bibfnamefont {J.}~\bibnamefont
  {Preskill}},\ }\href {https://doi.org/10.22331/q-2018-08-06-79} {\bibfield
  {journal} {\bibinfo  {journal} {Quantum}\ }\textbf {\bibinfo {volume} {2}},\
  \bibinfo {pages} {79} (\bibinfo {year} {2018})}\BibitemShut {NoStop}%
\bibitem [{\citenamefont {He}\ \emph {et~al.}(2025{\natexlab{a}})\citenamefont
  {He}, \citenamefont {Amaro}, \citenamefont {Shaydulin},\ and\ \citenamefont
  {Pistoia}}]{he2025performance}%
  \BibitemOpen
  \bibfield  {author} {\bibinfo {author} {\bibfnamefont {Z.}~\bibnamefont
  {He}}, \bibinfo {author} {\bibfnamefont {D.}~\bibnamefont {Amaro}}, \bibinfo
  {author} {\bibfnamefont {R.}~\bibnamefont {Shaydulin}},\ and\ \bibinfo
  {author} {\bibfnamefont {M.}~\bibnamefont {Pistoia}},\ }\href
  {https://doi.org/10.1038/s42005-025-02136-8} {\bibfield  {journal} {\bibinfo
  {journal} {Communications Physics}\ }\textbf {\bibinfo {volume} {8}},\
  \bibinfo {pages} {217} (\bibinfo {year} {2025}{\natexlab{a}})}\BibitemShut
  {NoStop}%
\bibitem [{\citenamefont {Jin}\ \emph {et~al.}(2025)\citenamefont {Jin},
  \citenamefont {He}, \citenamefont {Hao}, \citenamefont {Amaro}, \citenamefont
  {Tannu}, \citenamefont {Shaydulin},\ and\ \citenamefont
  {Pistoia}}]{jin2025iceberg}%
  \BibitemOpen
  \bibfield  {author} {\bibinfo {author} {\bibfnamefont {Y.}~\bibnamefont
  {Jin}}, \bibinfo {author} {\bibfnamefont {Z.}~\bibnamefont {He}}, \bibinfo
  {author} {\bibfnamefont {T.}~\bibnamefont {Hao}}, \bibinfo {author}
  {\bibfnamefont {D.}~\bibnamefont {Amaro}}, \bibinfo {author} {\bibfnamefont
  {S.}~\bibnamefont {Tannu}}, \bibinfo {author} {\bibfnamefont
  {R.}~\bibnamefont {Shaydulin}},\ and\ \bibinfo {author} {\bibfnamefont
  {M.}~\bibnamefont {Pistoia}},\ }\href
  {https://doi.org/10.48550/arXiv.2504.21172} {\bibinfo {title} {Iceberg
  {{Beyond}} the {{Tip}}: {{Co-Compilation}} of a {{Quantum Error Detection
  Code}} and a {{Quantum Algorithm}}}} (\bibinfo {year} {2025}),\ \Eprint
  {https://arxiv.org/abs/2504.21172} {arXiv:2504.21172 [quant-ph]} \BibitemShut
  {NoStop}%
\bibitem [{\citenamefont {Self}\ \emph {et~al.}(2024)\citenamefont {Self},
  \citenamefont {Benedetti},\ and\ \citenamefont {Amaro}}]{self2024protecting}%
  \BibitemOpen
  \bibfield  {author} {\bibinfo {author} {\bibfnamefont {C.~N.}\ \bibnamefont
  {Self}}, \bibinfo {author} {\bibfnamefont {M.}~\bibnamefont {Benedetti}},\
  and\ \bibinfo {author} {\bibfnamefont {D.}~\bibnamefont {Amaro}},\ }\href
  {https://doi.org/10.1038/s41567-023-02282-2} {\bibfield  {journal} {\bibinfo
  {journal} {Nature Physics}\ }\textbf {\bibinfo {volume} {20}},\ \bibinfo
  {pages} {219} (\bibinfo {year} {2024})}\BibitemShut {NoStop}%
\bibitem [{\citenamefont {Bravyi}\ \emph {et~al.}(2010)\citenamefont {Bravyi},
  \citenamefont {Poulin},\ and\ \citenamefont {Terhal}}]{bravyi2010tradeoffs}%
  \BibitemOpen
  \bibfield  {author} {\bibinfo {author} {\bibfnamefont {S.}~\bibnamefont
  {Bravyi}}, \bibinfo {author} {\bibfnamefont {D.}~\bibnamefont {Poulin}},\
  and\ \bibinfo {author} {\bibfnamefont {B.}~\bibnamefont {Terhal}},\ }\href
  {https://doi.org/10.1103/PhysRevLett.104.050503} {\bibfield  {journal}
  {\bibinfo  {journal} {Physical Review Letters}\ }\textbf {\bibinfo {volume}
  {104}},\ \bibinfo {pages} {050503} (\bibinfo {year} {2010})}\BibitemShut
  {NoStop}%
\bibitem [{\citenamefont {Panteleev}\ and\ \citenamefont
  {Kalachev}(2022)}]{panteleev2022asymptotically}%
  \BibitemOpen
  \bibfield  {author} {\bibinfo {author} {\bibfnamefont {P.}~\bibnamefont
  {Panteleev}}\ and\ \bibinfo {author} {\bibfnamefont {G.}~\bibnamefont
  {Kalachev}},\ }in\ \href {https://doi.org/10.1145/3519935.3520017} {\emph
  {\bibinfo {booktitle} {Proceedings of the 54th {{Annual ACM SIGACT
  Symposium}} on {{Theory}} of {{Computing}}}}},\ \bibinfo {series and number}
  {{{STOC}} 2022}\ (\bibinfo  {publisher} {Association for Computing
  Machinery},\ \bibinfo {address} {New York, NY, USA},\ \bibinfo {year}
  {2022})\ pp.\ \bibinfo {pages} {375--388}\BibitemShut {NoStop}%
\bibitem [{\citenamefont {Dinur}\ \emph {et~al.}(2023)\citenamefont {Dinur},
  \citenamefont {Hsieh}, \citenamefont {Lin},\ and\ \citenamefont
  {Vidick}}]{dinur2023good}%
  \BibitemOpen
  \bibfield  {author} {\bibinfo {author} {\bibfnamefont {I.}~\bibnamefont
  {Dinur}}, \bibinfo {author} {\bibfnamefont {M.-H.}\ \bibnamefont {Hsieh}},
  \bibinfo {author} {\bibfnamefont {T.-C.}\ \bibnamefont {Lin}},\ and\ \bibinfo
  {author} {\bibfnamefont {T.}~\bibnamefont {Vidick}},\ }in\ \href
  {https://doi.org/10.1145/3564246.3585101} {\emph {\bibinfo {booktitle}
  {Proceedings of the 55th {{Annual ACM Symposium}} on {{Theory}} of
  {{Computing}}}}},\ \bibinfo {series and number} {{{STOC}} 2023}\ (\bibinfo
  {publisher} {Association for Computing Machinery},\ \bibinfo {address} {New
  York, NY, USA},\ \bibinfo {year} {2023})\ pp.\ \bibinfo {pages}
  {905--918}\BibitemShut {NoStop}%
\bibitem [{\citenamefont {Baspin}\ and\ \citenamefont
  {Krishna}(2022)}]{baspin2022quantifying}%
  \BibitemOpen
  \bibfield  {author} {\bibinfo {author} {\bibfnamefont {N.}~\bibnamefont
  {Baspin}}\ and\ \bibinfo {author} {\bibfnamefont {A.}~\bibnamefont
  {Krishna}},\ }\href {https://doi.org/10.1103/PhysRevLett.129.050505}
  {\bibfield  {journal} {\bibinfo  {journal} {Physical Review Letters}\
  }\textbf {\bibinfo {volume} {129}},\ \bibinfo {pages} {050505} (\bibinfo
  {year} {2022})}\BibitemShut {NoStop}%
\bibitem [{\citenamefont {Dai}\ and\ \citenamefont
  {Li}(2025)}]{dai2025locality}%
  \BibitemOpen
  \bibfield  {author} {\bibinfo {author} {\bibfnamefont {S.}~\bibnamefont
  {Dai}}\ and\ \bibinfo {author} {\bibfnamefont {R.}~\bibnamefont {Li}},\ }in\
  \href {https://doi.org/10.1145/3717823.3718113} {\emph {\bibinfo {booktitle}
  {Proceedings of the 57th {{Annual ACM Symposium}} on {{Theory}} of
  {{Computing}}}}},\ \bibinfo {series and number} {{{STOC}} '25}\ (\bibinfo
  {publisher} {Association for Computing Machinery},\ \bibinfo {address} {New
  York, NY, USA},\ \bibinfo {year} {2025})\ pp.\ \bibinfo {pages}
  {677--688}\BibitemShut {NoStop}%
\bibitem [{\citenamefont {Koutsioumpas}\ \emph {et~al.}(2025)\citenamefont
  {Koutsioumpas}, \citenamefont {Noszko}, \citenamefont {Sayginel},
  \citenamefont {Webster},\ and\ \citenamefont
  {Roffe}}]{koutsioumpas2025colour}%
  \BibitemOpen
  \bibfield  {author} {\bibinfo {author} {\bibfnamefont {S.}~\bibnamefont
  {Koutsioumpas}}, \bibinfo {author} {\bibfnamefont {T.}~\bibnamefont
  {Noszko}}, \bibinfo {author} {\bibfnamefont {H.}~\bibnamefont {Sayginel}},
  \bibinfo {author} {\bibfnamefont {M.}~\bibnamefont {Webster}},\ and\ \bibinfo
  {author} {\bibfnamefont {J.}~\bibnamefont {Roffe}},\ }\href
  {https://doi.org/10.48550/arXiv.2508.15743} {\bibinfo {title} {Colour {{Codes
  Reach Surface Code Performance}} using {{Vibe Decoding}}}} (\bibinfo {year}
  {2025}),\ \Eprint {https://arxiv.org/abs/2508.15743} {arXiv:2508.15743
  [quant-ph]} \BibitemShut {NoStop}%
\bibitem [{\citenamefont {Panteleev}\ and\ \citenamefont
  {Kalachev}(2021)}]{panteleev2021degenerate}%
  \BibitemOpen
  \bibfield  {author} {\bibinfo {author} {\bibfnamefont {P.}~\bibnamefont
  {Panteleev}}\ and\ \bibinfo {author} {\bibfnamefont {G.}~\bibnamefont
  {Kalachev}},\ }\href {https://doi.org/10.22331/q-2021-11-22-585} {\bibfield
  {journal} {\bibinfo  {journal} {Quantum}\ }\textbf {\bibinfo {volume} {5}},\
  \bibinfo {pages} {585} (\bibinfo {year} {2021})}\BibitemShut {NoStop}%
\bibitem [{\citenamefont {Williamson}\ and\ \citenamefont
  {Yoder}(2026)}]{williamson2026lowoverhead}%
  \BibitemOpen
  \bibfield  {author} {\bibinfo {author} {\bibfnamefont {D.~J.}\ \bibnamefont
  {Williamson}}\ and\ \bibinfo {author} {\bibfnamefont {T.~J.}\ \bibnamefont
  {Yoder}},\ }\href {https://doi.org/10.1038/s41567-026-03220-8} {\bibfield
  {journal} {\bibinfo  {journal} {Nature Physics}\ }\textbf {\bibinfo {volume}
  {22}},\ \bibinfo {pages} {598} (\bibinfo {year} {2026})}\BibitemShut
  {NoStop}%
\bibitem [{\citenamefont {Swaroop}\ \emph {et~al.}(2024)\citenamefont
  {Swaroop}, \citenamefont {{Jochym-O'Connor}},\ and\ \citenamefont
  {Yoder}}]{swaroop2024universal}%
  \BibitemOpen
  \bibfield  {author} {\bibinfo {author} {\bibfnamefont {E.}~\bibnamefont
  {Swaroop}}, \bibinfo {author} {\bibfnamefont {T.}~\bibnamefont
  {{Jochym-O'Connor}}},\ and\ \bibinfo {author} {\bibfnamefont {T.~J.}\
  \bibnamefont {Yoder}},\ }\href {https://doi.org/10.48550/arXiv.2410.03628}
  {\bibinfo {title} {Universal adapters between quantum {{LDPC}} codes}}
  (\bibinfo {year} {2024}),\ \Eprint {https://arxiv.org/abs/2410.03628}
  {arXiv:2410.03628} \BibitemShut {NoStop}%
\bibitem [{\citenamefont {He}\ \emph {et~al.}(2025{\natexlab{b}})\citenamefont
  {He}, \citenamefont {Cowtan}, \citenamefont {Williamson},\ and\ \citenamefont
  {Yoder}}]{he2025extractors}%
  \BibitemOpen
  \bibfield  {author} {\bibinfo {author} {\bibfnamefont {Z.}~\bibnamefont
  {He}}, \bibinfo {author} {\bibfnamefont {A.}~\bibnamefont {Cowtan}}, \bibinfo
  {author} {\bibfnamefont {D.~J.}\ \bibnamefont {Williamson}},\ and\ \bibinfo
  {author} {\bibfnamefont {T.~J.}\ \bibnamefont {Yoder}},\ }\href
  {https://doi.org/10.48550/arXiv.2503.10390} {\bibinfo {title} {Extractors:
  {{QLDPC Architectures}} for {{Efficient Pauli-Based Computation}}}} (\bibinfo
  {year} {2025}{\natexlab{b}}),\ \Eprint {https://arxiv.org/abs/2503.10390}
  {arXiv:2503.10390} \BibitemShut {NoStop}%
\bibitem [{\citenamefont {Baspin}\ \emph {et~al.}(2025)\citenamefont {Baspin},
  \citenamefont {Berent},\ and\ \citenamefont {Cohen}}]{baspin2025fast}%
  \BibitemOpen
  \bibfield  {author} {\bibinfo {author} {\bibfnamefont {N.}~\bibnamefont
  {Baspin}}, \bibinfo {author} {\bibfnamefont {L.}~\bibnamefont {Berent}},\
  and\ \bibinfo {author} {\bibfnamefont {L.~Z.}\ \bibnamefont {Cohen}},\ }\href
  {https://doi.org/10.48550/arXiv.2510.04521} {\bibinfo {title} {Fast surgery
  for quantum {{LDPC}} codes}} (\bibinfo {year} {2025}),\ \Eprint
  {https://arxiv.org/abs/2510.04521} {arXiv:2510.04521 [quant-ph]} \BibitemShut
  {NoStop}%
\bibitem [{\citenamefont {Goto}(2024)}]{goto2024highperformance}%
  \BibitemOpen
  \bibfield  {author} {\bibinfo {author} {\bibfnamefont {H.}~\bibnamefont
  {Goto}},\ }\href {https://doi.org/10.1126/sciadv.adp6388} {\bibfield
  {journal} {\bibinfo  {journal} {Science Advances}\ }\textbf {\bibinfo
  {volume} {10}},\ \bibinfo {pages} {eadp6388} (\bibinfo {year}
  {2024})}\BibitemShut {NoStop}%
\bibitem [{\citenamefont {Yamasaki}\ and\ \citenamefont
  {Koashi}(2024)}]{yamasaki2024timeefficient}%
  \BibitemOpen
  \bibfield  {author} {\bibinfo {author} {\bibfnamefont {H.}~\bibnamefont
  {Yamasaki}}\ and\ \bibinfo {author} {\bibfnamefont {M.}~\bibnamefont
  {Koashi}},\ }\href {https://doi.org/10.1038/s41567-023-02325-8} {\bibfield
  {journal} {\bibinfo  {journal} {Nature Physics}\ }\textbf {\bibinfo {volume}
  {20}},\ \bibinfo {pages} {247} (\bibinfo {year} {2024})}\BibitemShut
  {NoStop}%
\bibitem [{\citenamefont {Yoshida}\ \emph {et~al.}(2025)\citenamefont
  {Yoshida}, \citenamefont {Tamiya},\ and\ \citenamefont
  {Yamasaki}}]{yoshida2025concatenate}%
  \BibitemOpen
  \bibfield  {author} {\bibinfo {author} {\bibfnamefont {S.}~\bibnamefont
  {Yoshida}}, \bibinfo {author} {\bibfnamefont {S.}~\bibnamefont {Tamiya}},\
  and\ \bibinfo {author} {\bibfnamefont {H.}~\bibnamefont {Yamasaki}},\ }\href
  {https://doi.org/10.1038/s41534-025-01035-8} {\bibfield  {journal} {\bibinfo
  {journal} {npj Quantum Information}\ }\textbf {\bibinfo {volume} {11}},\
  \bibinfo {pages} {88} (\bibinfo {year} {2025})}\BibitemShut {NoStop}%
\bibitem [{\citenamefont {Nakai}\ and\ \citenamefont
  {Goto}(2026)}]{nakai2026subsystem}%
  \BibitemOpen
  \bibfield  {author} {\bibinfo {author} {\bibfnamefont {R.}~\bibnamefont
  {Nakai}}\ and\ \bibinfo {author} {\bibfnamefont {H.}~\bibnamefont {Goto}},\
  }\href {https://doi.org/10.1103/xbzn-vn37} {\bibfield  {journal} {\bibinfo
  {journal} {Physical Review Applied}\ }\textbf {\bibinfo {volume} {25}},\
  \bibinfo {pages} {014032} (\bibinfo {year} {2026})}\BibitemShut {NoStop}%
\bibitem [{\citenamefont {Pastawski}\ \emph {et~al.}(2015)\citenamefont
  {Pastawski}, \citenamefont {Yoshida}, \citenamefont {Harlow},\ and\
  \citenamefont {Preskill}}]{pastawski2015holographic}%
  \BibitemOpen
  \bibfield  {author} {\bibinfo {author} {\bibfnamefont {F.}~\bibnamefont
  {Pastawski}}, \bibinfo {author} {\bibfnamefont {B.}~\bibnamefont {Yoshida}},
  \bibinfo {author} {\bibfnamefont {D.}~\bibnamefont {Harlow}},\ and\ \bibinfo
  {author} {\bibfnamefont {J.}~\bibnamefont {Preskill}},\ }\href
  {https://doi.org/10.1007/JHEP06(2015)149} {\bibfield  {journal} {\bibinfo
  {journal} {Journal of High Energy Physics}\ }\textbf {\bibinfo {volume}
  {2015}},\ \bibinfo {pages} {149} (\bibinfo {year} {2015})}\BibitemShut
  {NoStop}%
\bibitem [{\citenamefont {Harris}\ \emph {et~al.}(2018)\citenamefont {Harris},
  \citenamefont {McMahon}, \citenamefont {Brennen},\ and\ \citenamefont
  {Stace}}]{harris2018calderbankshorsteane}%
  \BibitemOpen
  \bibfield  {author} {\bibinfo {author} {\bibfnamefont {R.~J.}\ \bibnamefont
  {Harris}}, \bibinfo {author} {\bibfnamefont {N.~A.}\ \bibnamefont {McMahon}},
  \bibinfo {author} {\bibfnamefont {G.~K.}\ \bibnamefont {Brennen}},\ and\
  \bibinfo {author} {\bibfnamefont {T.~M.}\ \bibnamefont {Stace}},\ }\href
  {https://doi.org/10.1103/PhysRevA.98.052301} {\bibfield  {journal} {\bibinfo
  {journal} {Physical Review A}\ }\textbf {\bibinfo {volume} {98}},\ \bibinfo
  {pages} {052301} (\bibinfo {year} {2018})}\BibitemShut {NoStop}%
\bibitem [{\citenamefont {Harris}\ \emph {et~al.}(2020)\citenamefont {Harris},
  \citenamefont {Coupe}, \citenamefont {McMahon}, \citenamefont {Brennen},\
  and\ \citenamefont {Stace}}]{harris2020decoding}%
  \BibitemOpen
  \bibfield  {author} {\bibinfo {author} {\bibfnamefont {R.~J.}\ \bibnamefont
  {Harris}}, \bibinfo {author} {\bibfnamefont {E.}~\bibnamefont {Coupe}},
  \bibinfo {author} {\bibfnamefont {N.~A.}\ \bibnamefont {McMahon}}, \bibinfo
  {author} {\bibfnamefont {G.~K.}\ \bibnamefont {Brennen}},\ and\ \bibinfo
  {author} {\bibfnamefont {T.~M.}\ \bibnamefont {Stace}},\ }\href
  {https://doi.org/10.1103/PhysRevA.102.062417} {\bibfield  {journal} {\bibinfo
   {journal} {Physical Review A}\ }\textbf {\bibinfo {volume} {102}},\ \bibinfo
  {pages} {062417} (\bibinfo {year} {2020})}\BibitemShut {NoStop}%
\bibitem [{\citenamefont {Jahn}\ and\ \citenamefont
  {Eisert}(2021)}]{jahn2021holographic}%
  \BibitemOpen
  \bibfield  {author} {\bibinfo {author} {\bibfnamefont {A.}~\bibnamefont
  {Jahn}}\ and\ \bibinfo {author} {\bibfnamefont {J.}~\bibnamefont {Eisert}},\
  }\href {https://doi.org/10.1088/2058-9565/ac0293} {\bibfield  {journal}
  {\bibinfo  {journal} {Quantum Science and Technology}\ }\textbf {\bibinfo
  {volume} {6}},\ \bibinfo {pages} {033002} (\bibinfo {year}
  {2021})}\BibitemShut {NoStop}%
\bibitem [{\citenamefont {Farrelly}\ \emph {et~al.}(2022)\citenamefont
  {Farrelly}, \citenamefont {Milicevic}, \citenamefont {Harris}, \citenamefont
  {McMahon},\ and\ \citenamefont {Stace}}]{farrelly2022parallel}%
  \BibitemOpen
  \bibfield  {author} {\bibinfo {author} {\bibfnamefont {T.}~\bibnamefont
  {Farrelly}}, \bibinfo {author} {\bibfnamefont {N.}~\bibnamefont {Milicevic}},
  \bibinfo {author} {\bibfnamefont {R.~J.}\ \bibnamefont {Harris}}, \bibinfo
  {author} {\bibfnamefont {N.~A.}\ \bibnamefont {McMahon}},\ and\ \bibinfo
  {author} {\bibfnamefont {T.~M.}\ \bibnamefont {Stace}},\ }\href
  {https://doi.org/10.1103/PhysRevA.105.052446} {\bibfield  {journal} {\bibinfo
   {journal} {Physical Review A}\ }\textbf {\bibinfo {volume} {105}},\ \bibinfo
  {pages} {052446} (\bibinfo {year} {2022})}\BibitemShut {NoStop}%
\bibitem [{\citenamefont {Steinberg}\ \emph
  {et~al.}(2025{\natexlab{a}})\citenamefont {Steinberg}, \citenamefont {Fan},
  \citenamefont {Harris}, \citenamefont {Elkouss}, \citenamefont {Feld},\ and\
  \citenamefont {Jahn}}]{steinberg2025far}%
  \BibitemOpen
  \bibfield  {author} {\bibinfo {author} {\bibfnamefont {M.}~\bibnamefont
  {Steinberg}}, \bibinfo {author} {\bibfnamefont {J.}~\bibnamefont {Fan}},
  \bibinfo {author} {\bibfnamefont {R.~J.}\ \bibnamefont {Harris}}, \bibinfo
  {author} {\bibfnamefont {D.}~\bibnamefont {Elkouss}}, \bibinfo {author}
  {\bibfnamefont {S.}~\bibnamefont {Feld}},\ and\ \bibinfo {author}
  {\bibfnamefont {A.}~\bibnamefont {Jahn}},\ }\href
  {https://doi.org/10.22331/q-2025-08-08-1826} {\bibfield  {journal} {\bibinfo
  {journal} {Quantum}\ }\textbf {\bibinfo {volume} {9}},\ \bibinfo {pages}
  {1826} (\bibinfo {year} {2025}{\natexlab{a}})}\BibitemShut {NoStop}%
\bibitem [{\citenamefont {Steinberg}\ \emph
  {et~al.}(2025{\natexlab{b}})\citenamefont {Steinberg}, \citenamefont {Fan},
  \citenamefont {Eisert}, \citenamefont {Feld}, \citenamefont {Jahn},\ and\
  \citenamefont {Cao}}]{steinberg2025universal}%
  \BibitemOpen
  \bibfield  {author} {\bibinfo {author} {\bibfnamefont {M.}~\bibnamefont
  {Steinberg}}, \bibinfo {author} {\bibfnamefont {J.}~\bibnamefont {Fan}},
  \bibinfo {author} {\bibfnamefont {J.}~\bibnamefont {Eisert}}, \bibinfo
  {author} {\bibfnamefont {S.}~\bibnamefont {Feld}}, \bibinfo {author}
  {\bibfnamefont {A.}~\bibnamefont {Jahn}},\ and\ \bibinfo {author}
  {\bibfnamefont {C.}~\bibnamefont {Cao}},\ }\href
  {https://arxiv.org/abs/2504.10386} {\bibinfo {title} {Universal
  fault-tolerant logic with heterogeneous holographic codes}} (\bibinfo {year}
  {2025}{\natexlab{b}}),\ \Eprint {https://arxiv.org/abs/2504.10386}
  {arXiv:2504.10386 [quant-ph]} \BibitemShut {NoStop}%
\bibitem [{\citenamefont {Bluvstein}\ \emph {et~al.}(2022)\citenamefont
  {Bluvstein}, \citenamefont {Levine}, \citenamefont {Semeghini}, \citenamefont
  {Wang}, \citenamefont {Ebadi}, \citenamefont {Kalinowski}, \citenamefont
  {Keesling}, \citenamefont {Maskara}, \citenamefont {Pichler}, \citenamefont
  {Greiner}, \citenamefont {Vuleti{\'c}},\ and\ \citenamefont
  {Lukin}}]{bluvstein2022quantum}%
  \BibitemOpen
  \bibfield  {author} {\bibinfo {author} {\bibfnamefont {D.}~\bibnamefont
  {Bluvstein}}, \bibinfo {author} {\bibfnamefont {H.}~\bibnamefont {Levine}},
  \bibinfo {author} {\bibfnamefont {G.}~\bibnamefont {Semeghini}}, \bibinfo
  {author} {\bibfnamefont {T.~T.}\ \bibnamefont {Wang}}, \bibinfo {author}
  {\bibfnamefont {S.}~\bibnamefont {Ebadi}}, \bibinfo {author} {\bibfnamefont
  {M.}~\bibnamefont {Kalinowski}}, \bibinfo {author} {\bibfnamefont
  {A.}~\bibnamefont {Keesling}}, \bibinfo {author} {\bibfnamefont
  {N.}~\bibnamefont {Maskara}}, \bibinfo {author} {\bibfnamefont
  {H.}~\bibnamefont {Pichler}}, \bibinfo {author} {\bibfnamefont
  {M.}~\bibnamefont {Greiner}}, \bibinfo {author} {\bibfnamefont
  {V.}~\bibnamefont {Vuleti{\'c}}},\ and\ \bibinfo {author} {\bibfnamefont
  {M.~D.}\ \bibnamefont {Lukin}},\ }\href
  {https://doi.org/10.1038/s41586-022-04592-6} {\bibfield  {journal} {\bibinfo
  {journal} {Nature}\ }\textbf {\bibinfo {volume} {604}},\ \bibinfo {pages}
  {451} (\bibinfo {year} {2022})}\BibitemShut {NoStop}%
\bibitem [{\citenamefont {Ransford}\ \emph {et~al.}(2025)\citenamefont
  {Ransford}, \citenamefont {Allman}, \citenamefont {Arkinstall}, \citenamefont
  {Campora}, \citenamefont {Cooper}, \citenamefont {Delaney}, \citenamefont
  {Dreiling}, \citenamefont {Estey}, \citenamefont {Figgatt}, \citenamefont
  {Hall}, \citenamefont {Husain}, \citenamefont {Isanaka}, \citenamefont
  {Kennedy}, \citenamefont {Kotibhaskar}, \citenamefont {Madjarov},
  \citenamefont {Mayer}, \citenamefont {Milne}, \citenamefont {Park},
  \citenamefont {Reed}, \citenamefont {Ancona}, \citenamefont {Andersen},
  \citenamefont {{Andres-Martinez}}, \citenamefont {Angenent}, \citenamefont
  {Argueta}, \citenamefont {Arkin}, \citenamefont {Ascarrunz}, \citenamefont
  {Baker}, \citenamefont {Barnes}, \citenamefont {Bartolotta}, \citenamefont
  {Berg}, \citenamefont {Besand}, \citenamefont {Bjork}, \citenamefont {Blain},
  \citenamefont {Blanchard}, \citenamefont {{Blume-Kohout}}, \citenamefont
  {Bohn}, \citenamefont {Borgna}, \citenamefont {Botamanenko}, \citenamefont
  {Boutelle}, \citenamefont {Brown}, \citenamefont {Buckingham}, \citenamefont
  {Burdick}, \citenamefont {Burton}, \citenamefont {Carey}, \citenamefont
  {Carron}, \citenamefont {Chambers}, \citenamefont {Children}, \citenamefont
  {Colussi}, \citenamefont {Crepinsek}, \citenamefont {Cureton}, \citenamefont
  {Davies}, \citenamefont {Davis}, \citenamefont {DeCross}, \citenamefont
  {Deen}, \citenamefont {Delaney}, \citenamefont {DelVento}, \citenamefont
  {DeSalvo}, \citenamefont {Dominy}, \citenamefont {Duncan}, \citenamefont
  {Eccles}, \citenamefont {Edgington}, \citenamefont {Erickson}, \citenamefont
  {Erickson}, \citenamefont {Ertsgaard}, \citenamefont {Evans}, \citenamefont
  {Evans}, \citenamefont {Fabrikant}, \citenamefont {Fischer}, \citenamefont
  {Foltz}, \citenamefont {{Foss-Feig}}, \citenamefont {Francois}, \citenamefont
  {Freyberg}, \citenamefont {Gao}, \citenamefont {Garay}, \citenamefont
  {Garvin}, \citenamefont {Gaudiosi}, \citenamefont {Gilbreth}, \citenamefont
  {Giles}, \citenamefont {Glynn}, \citenamefont {Graves}, \citenamefont
  {Hansen}, \citenamefont {Hayes}, \citenamefont {Heidemann}, \citenamefont
  {Higashi}, \citenamefont {Hilbun}, \citenamefont {Hines}, \citenamefont
  {Hlavaty}, \citenamefont {Hoffman}, \citenamefont {Hoffman}, \citenamefont
  {Holliman}, \citenamefont {Hooper}, \citenamefont {Horning}, \citenamefont
  {Hostetter}, \citenamefont {Hothem}, \citenamefont {Houlton}, \citenamefont
  {Hout}, \citenamefont {Hutson}, \citenamefont {Jacobs}, \citenamefont
  {Jacobs}, \citenamefont {Johannsen}, \citenamefont {Johansen}, \citenamefont
  {Jones}, \citenamefont {Julian}, \citenamefont {Jung}, \citenamefont {Keay},
  \citenamefont {Klein}, \citenamefont {Koch}, \citenamefont {Kondo},
  \citenamefont {Kong}, \citenamefont {Kosto}, \citenamefont {Lawrence},
  \citenamefont {Liefer}, \citenamefont {Lollie}, \citenamefont {Lucchetti},
  \citenamefont {Lysne}, \citenamefont {Lytle}, \citenamefont {MacPherson},
  \citenamefont {Malm}, \citenamefont {Mather}, \citenamefont {Mathewson},
  \citenamefont {Maxwell}, \citenamefont {McCaffrey}, \citenamefont
  {McDougall}, \citenamefont {Mendoza}, \citenamefont {Mills}, \citenamefont
  {Morrison}, \citenamefont {Narmour}, \citenamefont {Nguyen}, \citenamefont
  {Nugent}, \citenamefont {Olson}, \citenamefont {Ouellette}, \citenamefont
  {Parks}, \citenamefont {Peters}, \citenamefont {Petricka}, \citenamefont
  {Pino}, \citenamefont {Polito}, \citenamefont {Preidl}, \citenamefont
  {Price}, \citenamefont {Proctor}, \citenamefont {Pugh}, \citenamefont
  {Ratcliff}, \citenamefont {Raymondson}, \citenamefont {Rhodes}, \citenamefont
  {Roman}, \citenamefont {Roy}, \citenamefont {{Ryan-Anderson}}, \citenamefont
  {Sanchez}, \citenamefont {Sangiolo}, \citenamefont {Sawadski}, \citenamefont
  {Schaffer}, \citenamefont {Schow}, \citenamefont {Sedlacek}, \citenamefont
  {Semenenko}, \citenamefont {Shevchuk}, \citenamefont {Shore}, \citenamefont
  {Siegfried}, \citenamefont {Singhal}, \citenamefont {Sivarajah},
  \citenamefont {Skripka}, \citenamefont {Sletten}, \citenamefont {Spaun},
  \citenamefont {Sprenkle}, \citenamefont {Stoufer}, \citenamefont {Tader},
  \citenamefont {Taylor}, \citenamefont {Thompson}, \citenamefont {Tobey},
  \citenamefont {Tran}, \citenamefont {Tran}, \citenamefont {Vittorini},
  \citenamefont {Volin}, \citenamefont {Walker}, \citenamefont {White},
  \citenamefont {Wilson}, \citenamefont {Wolf}, \citenamefont {Wringe},
  \citenamefont {Young}, \citenamefont {Zheng}, \citenamefont {Zuraski},
  \citenamefont {Baldwin}, \citenamefont {Chernoguzov}, \citenamefont
  {Gaebler}, \citenamefont {Sanders}, \citenamefont {Neyenhuis}, \citenamefont
  {Stutz},\ and\ \citenamefont {Bohnet}}]{ransford2025helios}%
  \BibitemOpen
  \bibfield  {author} {\bibinfo {author} {\bibfnamefont {A.}~\bibnamefont
  {Ransford}}, \bibinfo {author} {\bibfnamefont {M.~S.}\ \bibnamefont
  {Allman}}, \bibinfo {author} {\bibfnamefont {J.}~\bibnamefont {Arkinstall}},
  \bibinfo {author} {\bibfnamefont {J.~P.}\ \bibnamefont {Campora}}, \bibinfo
  {author} {\bibfnamefont {S.~F.}\ \bibnamefont {Cooper}}, \bibinfo {author}
  {\bibfnamefont {R.~D.}\ \bibnamefont {Delaney}}, \bibinfo {author}
  {\bibfnamefont {J.~M.}\ \bibnamefont {Dreiling}}, \bibinfo {author}
  {\bibfnamefont {B.}~\bibnamefont {Estey}}, \bibinfo {author} {\bibfnamefont
  {C.}~\bibnamefont {Figgatt}}, \bibinfo {author} {\bibfnamefont
  {A.}~\bibnamefont {Hall}}, \bibinfo {author} {\bibfnamefont {A.~A.}\
  \bibnamefont {Husain}}, \bibinfo {author} {\bibfnamefont {A.}~\bibnamefont
  {Isanaka}}, \bibinfo {author} {\bibfnamefont {C.~J.}\ \bibnamefont
  {Kennedy}}, \bibinfo {author} {\bibfnamefont {N.}~\bibnamefont
  {Kotibhaskar}}, \bibinfo {author} {\bibfnamefont {I.~S.}\ \bibnamefont
  {Madjarov}}, \bibinfo {author} {\bibfnamefont {K.}~\bibnamefont {Mayer}},
  \bibinfo {author} {\bibfnamefont {A.~R.}\ \bibnamefont {Milne}}, \bibinfo
  {author} {\bibfnamefont {A.~J.}\ \bibnamefont {Park}}, \bibinfo {author}
  {\bibfnamefont {A.~P.}\ \bibnamefont {Reed}}, \bibinfo {author}
  {\bibfnamefont {R.}~\bibnamefont {Ancona}}, \bibinfo {author} {\bibfnamefont
  {M.~P.}\ \bibnamefont {Andersen}}, \bibinfo {author} {\bibfnamefont
  {P.}~\bibnamefont {{Andres-Martinez}}}, \bibinfo {author} {\bibfnamefont
  {W.}~\bibnamefont {Angenent}}, \bibinfo {author} {\bibfnamefont
  {L.}~\bibnamefont {Argueta}}, \bibinfo {author} {\bibfnamefont
  {B.}~\bibnamefont {Arkin}}, \bibinfo {author} {\bibfnamefont
  {L.}~\bibnamefont {Ascarrunz}}, \bibinfo {author} {\bibfnamefont
  {W.}~\bibnamefont {Baker}}, \bibinfo {author} {\bibfnamefont
  {C.}~\bibnamefont {Barnes}}, \bibinfo {author} {\bibfnamefont
  {J.}~\bibnamefont {Bartolotta}}, \bibinfo {author} {\bibfnamefont
  {J.}~\bibnamefont {Berg}}, \bibinfo {author} {\bibfnamefont {R.}~\bibnamefont
  {Besand}}, \bibinfo {author} {\bibfnamefont {B.}~\bibnamefont {Bjork}},
  \bibinfo {author} {\bibfnamefont {M.}~\bibnamefont {Blain}}, \bibinfo
  {author} {\bibfnamefont {P.}~\bibnamefont {Blanchard}}, \bibinfo {author}
  {\bibfnamefont {R.}~\bibnamefont {{Blume-Kohout}}}, \bibinfo {author}
  {\bibfnamefont {M.}~\bibnamefont {Bohn}}, \bibinfo {author} {\bibfnamefont
  {A.}~\bibnamefont {Borgna}}, \bibinfo {author} {\bibfnamefont {D.~Y.}\
  \bibnamefont {Botamanenko}}, \bibinfo {author} {\bibfnamefont
  {R.}~\bibnamefont {Boutelle}}, \bibinfo {author} {\bibfnamefont
  {N.}~\bibnamefont {Brown}}, \bibinfo {author} {\bibfnamefont {G.~T.}\
  \bibnamefont {Buckingham}}, \bibinfo {author} {\bibfnamefont {N.~Q.}\
  \bibnamefont {Burdick}}, \bibinfo {author} {\bibfnamefont {W.~C.}\
  \bibnamefont {Burton}}, \bibinfo {author} {\bibfnamefont {V.}~\bibnamefont
  {Carey}}, \bibinfo {author} {\bibfnamefont {C.~J.}\ \bibnamefont {Carron}},
  \bibinfo {author} {\bibfnamefont {J.}~\bibnamefont {Chambers}}, \bibinfo
  {author} {\bibfnamefont {J.}~\bibnamefont {Children}}, \bibinfo {author}
  {\bibfnamefont {V.~E.}\ \bibnamefont {Colussi}}, \bibinfo {author}
  {\bibfnamefont {S.}~\bibnamefont {Crepinsek}}, \bibinfo {author}
  {\bibfnamefont {A.}~\bibnamefont {Cureton}}, \bibinfo {author} {\bibfnamefont
  {J.}~\bibnamefont {Davies}}, \bibinfo {author} {\bibfnamefont
  {D.}~\bibnamefont {Davis}}, \bibinfo {author} {\bibfnamefont
  {M.}~\bibnamefont {DeCross}}, \bibinfo {author} {\bibfnamefont
  {D.}~\bibnamefont {Deen}}, \bibinfo {author} {\bibfnamefont {C.}~\bibnamefont
  {Delaney}}, \bibinfo {author} {\bibfnamefont {D.}~\bibnamefont {DelVento}},
  \bibinfo {author} {\bibfnamefont {B.~J.}\ \bibnamefont {DeSalvo}}, \bibinfo
  {author} {\bibfnamefont {J.}~\bibnamefont {Dominy}}, \bibinfo {author}
  {\bibfnamefont {R.}~\bibnamefont {Duncan}}, \bibinfo {author} {\bibfnamefont
  {V.}~\bibnamefont {Eccles}}, \bibinfo {author} {\bibfnamefont
  {A.}~\bibnamefont {Edgington}}, \bibinfo {author} {\bibfnamefont
  {N.}~\bibnamefont {Erickson}}, \bibinfo {author} {\bibfnamefont
  {S.}~\bibnamefont {Erickson}}, \bibinfo {author} {\bibfnamefont {C.~T.}\
  \bibnamefont {Ertsgaard}}, \bibinfo {author} {\bibfnamefont {B.}~\bibnamefont
  {Evans}}, \bibinfo {author} {\bibfnamefont {T.}~\bibnamefont {Evans}},
  \bibinfo {author} {\bibfnamefont {M.~I.}\ \bibnamefont {Fabrikant}}, \bibinfo
  {author} {\bibfnamefont {A.}~\bibnamefont {Fischer}}, \bibinfo {author}
  {\bibfnamefont {C.}~\bibnamefont {Foltz}}, \bibinfo {author} {\bibfnamefont
  {M.}~\bibnamefont {{Foss-Feig}}}, \bibinfo {author} {\bibfnamefont
  {D.}~\bibnamefont {Francois}}, \bibinfo {author} {\bibfnamefont
  {B.}~\bibnamefont {Freyberg}}, \bibinfo {author} {\bibfnamefont
  {C.}~\bibnamefont {Gao}}, \bibinfo {author} {\bibfnamefont {R.}~\bibnamefont
  {Garay}}, \bibinfo {author} {\bibfnamefont {J.}~\bibnamefont {Garvin}},
  \bibinfo {author} {\bibfnamefont {D.~M.}\ \bibnamefont {Gaudiosi}}, \bibinfo
  {author} {\bibfnamefont {C.~N.}\ \bibnamefont {Gilbreth}}, \bibinfo {author}
  {\bibfnamefont {J.}~\bibnamefont {Giles}}, \bibinfo {author} {\bibfnamefont
  {E.}~\bibnamefont {Glynn}}, \bibinfo {author} {\bibfnamefont
  {J.}~\bibnamefont {Graves}}, \bibinfo {author} {\bibfnamefont
  {A.}~\bibnamefont {Hansen}}, \bibinfo {author} {\bibfnamefont
  {D.}~\bibnamefont {Hayes}}, \bibinfo {author} {\bibfnamefont
  {L.}~\bibnamefont {Heidemann}}, \bibinfo {author} {\bibfnamefont
  {B.}~\bibnamefont {Higashi}}, \bibinfo {author} {\bibfnamefont
  {T.}~\bibnamefont {Hilbun}}, \bibinfo {author} {\bibfnamefont
  {J.}~\bibnamefont {Hines}}, \bibinfo {author} {\bibfnamefont
  {A.}~\bibnamefont {Hlavaty}}, \bibinfo {author} {\bibfnamefont
  {K.}~\bibnamefont {Hoffman}}, \bibinfo {author} {\bibfnamefont {I.~M.}\
  \bibnamefont {Hoffman}}, \bibinfo {author} {\bibfnamefont {C.}~\bibnamefont
  {Holliman}}, \bibinfo {author} {\bibfnamefont {I.}~\bibnamefont {Hooper}},
  \bibinfo {author} {\bibfnamefont {B.}~\bibnamefont {Horning}}, \bibinfo
  {author} {\bibfnamefont {J.}~\bibnamefont {Hostetter}}, \bibinfo {author}
  {\bibfnamefont {D.}~\bibnamefont {Hothem}}, \bibinfo {author} {\bibfnamefont
  {J.}~\bibnamefont {Houlton}}, \bibinfo {author} {\bibfnamefont
  {J.}~\bibnamefont {Hout}}, \bibinfo {author} {\bibfnamefont {R.}~\bibnamefont
  {Hutson}}, \bibinfo {author} {\bibfnamefont {R.~T.}\ \bibnamefont {Jacobs}},
  \bibinfo {author} {\bibfnamefont {T.}~\bibnamefont {Jacobs}}, \bibinfo
  {author} {\bibfnamefont {M.}~\bibnamefont {Johannsen}}, \bibinfo {author}
  {\bibfnamefont {J.}~\bibnamefont {Johansen}}, \bibinfo {author}
  {\bibfnamefont {L.}~\bibnamefont {Jones}}, \bibinfo {author} {\bibfnamefont
  {S.}~\bibnamefont {Julian}}, \bibinfo {author} {\bibfnamefont
  {R.}~\bibnamefont {Jung}}, \bibinfo {author} {\bibfnamefont {A.}~\bibnamefont
  {Keay}}, \bibinfo {author} {\bibfnamefont {T.}~\bibnamefont {Klein}},
  \bibinfo {author} {\bibfnamefont {M.}~\bibnamefont {Koch}}, \bibinfo {author}
  {\bibfnamefont {R.}~\bibnamefont {Kondo}}, \bibinfo {author} {\bibfnamefont
  {C.}~\bibnamefont {Kong}}, \bibinfo {author} {\bibfnamefont {A.}~\bibnamefont
  {Kosto}}, \bibinfo {author} {\bibfnamefont {A.}~\bibnamefont {Lawrence}},
  \bibinfo {author} {\bibfnamefont {D.}~\bibnamefont {Liefer}}, \bibinfo
  {author} {\bibfnamefont {M.}~\bibnamefont {Lollie}}, \bibinfo {author}
  {\bibfnamefont {D.}~\bibnamefont {Lucchetti}}, \bibinfo {author}
  {\bibfnamefont {N.~K.}\ \bibnamefont {Lysne}}, \bibinfo {author}
  {\bibfnamefont {C.}~\bibnamefont {Lytle}}, \bibinfo {author} {\bibfnamefont
  {C.}~\bibnamefont {MacPherson}}, \bibinfo {author} {\bibfnamefont
  {A.}~\bibnamefont {Malm}}, \bibinfo {author} {\bibfnamefont {S.}~\bibnamefont
  {Mather}}, \bibinfo {author} {\bibfnamefont {B.}~\bibnamefont {Mathewson}},
  \bibinfo {author} {\bibfnamefont {D.}~\bibnamefont {Maxwell}}, \bibinfo
  {author} {\bibfnamefont {L.}~\bibnamefont {McCaffrey}}, \bibinfo {author}
  {\bibfnamefont {H.}~\bibnamefont {McDougall}}, \bibinfo {author}
  {\bibfnamefont {R.}~\bibnamefont {Mendoza}}, \bibinfo {author} {\bibfnamefont
  {M.}~\bibnamefont {Mills}}, \bibinfo {author} {\bibfnamefont
  {R.}~\bibnamefont {Morrison}}, \bibinfo {author} {\bibfnamefont
  {L.}~\bibnamefont {Narmour}}, \bibinfo {author} {\bibfnamefont
  {N.}~\bibnamefont {Nguyen}}, \bibinfo {author} {\bibfnamefont
  {L.}~\bibnamefont {Nugent}}, \bibinfo {author} {\bibfnamefont
  {S.}~\bibnamefont {Olson}}, \bibinfo {author} {\bibfnamefont
  {D.}~\bibnamefont {Ouellette}}, \bibinfo {author} {\bibfnamefont
  {J.}~\bibnamefont {Parks}}, \bibinfo {author} {\bibfnamefont
  {Z.}~\bibnamefont {Peters}}, \bibinfo {author} {\bibfnamefont
  {J.}~\bibnamefont {Petricka}}, \bibinfo {author} {\bibfnamefont {J.~M.}\
  \bibnamefont {Pino}}, \bibinfo {author} {\bibfnamefont {F.}~\bibnamefont
  {Polito}}, \bibinfo {author} {\bibfnamefont {M.}~\bibnamefont {Preidl}},
  \bibinfo {author} {\bibfnamefont {G.}~\bibnamefont {Price}}, \bibinfo
  {author} {\bibfnamefont {T.}~\bibnamefont {Proctor}}, \bibinfo {author}
  {\bibfnamefont {M.}~\bibnamefont {Pugh}}, \bibinfo {author} {\bibfnamefont
  {N.}~\bibnamefont {Ratcliff}}, \bibinfo {author} {\bibfnamefont
  {D.}~\bibnamefont {Raymondson}}, \bibinfo {author} {\bibfnamefont
  {P.}~\bibnamefont {Rhodes}}, \bibinfo {author} {\bibfnamefont
  {C.}~\bibnamefont {Roman}}, \bibinfo {author} {\bibfnamefont
  {C.}~\bibnamefont {Roy}}, \bibinfo {author} {\bibfnamefont {C.}~\bibnamefont
  {{Ryan-Anderson}}}, \bibinfo {author} {\bibfnamefont {F.~B.}\ \bibnamefont
  {Sanchez}}, \bibinfo {author} {\bibfnamefont {G.}~\bibnamefont {Sangiolo}},
  \bibinfo {author} {\bibfnamefont {T.}~\bibnamefont {Sawadski}}, \bibinfo
  {author} {\bibfnamefont {A.}~\bibnamefont {Schaffer}}, \bibinfo {author}
  {\bibfnamefont {P.}~\bibnamefont {Schow}}, \bibinfo {author} {\bibfnamefont
  {J.}~\bibnamefont {Sedlacek}}, \bibinfo {author} {\bibfnamefont
  {H.}~\bibnamefont {Semenenko}}, \bibinfo {author} {\bibfnamefont
  {P.}~\bibnamefont {Shevchuk}}, \bibinfo {author} {\bibfnamefont
  {S.}~\bibnamefont {Shore}}, \bibinfo {author} {\bibfnamefont
  {P.}~\bibnamefont {Siegfried}}, \bibinfo {author} {\bibfnamefont
  {K.}~\bibnamefont {Singhal}}, \bibinfo {author} {\bibfnamefont
  {S.}~\bibnamefont {Sivarajah}}, \bibinfo {author} {\bibfnamefont
  {T.}~\bibnamefont {Skripka}}, \bibinfo {author} {\bibfnamefont
  {L.}~\bibnamefont {Sletten}}, \bibinfo {author} {\bibfnamefont
  {B.}~\bibnamefont {Spaun}}, \bibinfo {author} {\bibfnamefont {R.~T.}\
  \bibnamefont {Sprenkle}}, \bibinfo {author} {\bibfnamefont {P.}~\bibnamefont
  {Stoufer}}, \bibinfo {author} {\bibfnamefont {M.}~\bibnamefont {Tader}},
  \bibinfo {author} {\bibfnamefont {S.~F.}\ \bibnamefont {Taylor}}, \bibinfo
  {author} {\bibfnamefont {T.~H.}\ \bibnamefont {Thompson}}, \bibinfo {author}
  {\bibfnamefont {R.}~\bibnamefont {Tobey}}, \bibinfo {author} {\bibfnamefont
  {A.}~\bibnamefont {Tran}}, \bibinfo {author} {\bibfnamefont {T.}~\bibnamefont
  {Tran}}, \bibinfo {author} {\bibfnamefont {G.}~\bibnamefont {Vittorini}},
  \bibinfo {author} {\bibfnamefont {C.}~\bibnamefont {Volin}}, \bibinfo
  {author} {\bibfnamefont {J.}~\bibnamefont {Walker}}, \bibinfo {author}
  {\bibfnamefont {S.}~\bibnamefont {White}}, \bibinfo {author} {\bibfnamefont
  {D.}~\bibnamefont {Wilson}}, \bibinfo {author} {\bibfnamefont
  {Q.}~\bibnamefont {Wolf}}, \bibinfo {author} {\bibfnamefont {C.}~\bibnamefont
  {Wringe}}, \bibinfo {author} {\bibfnamefont {K.}~\bibnamefont {Young}},
  \bibinfo {author} {\bibfnamefont {J.}~\bibnamefont {Zheng}}, \bibinfo
  {author} {\bibfnamefont {K.}~\bibnamefont {Zuraski}}, \bibinfo {author}
  {\bibfnamefont {C.~H.}\ \bibnamefont {Baldwin}}, \bibinfo {author}
  {\bibfnamefont {A.}~\bibnamefont {Chernoguzov}}, \bibinfo {author}
  {\bibfnamefont {J.~P.}\ \bibnamefont {Gaebler}}, \bibinfo {author}
  {\bibfnamefont {S.~J.}\ \bibnamefont {Sanders}}, \bibinfo {author}
  {\bibfnamefont {B.}~\bibnamefont {Neyenhuis}}, \bibinfo {author}
  {\bibfnamefont {R.}~\bibnamefont {Stutz}},\ and\ \bibinfo {author}
  {\bibfnamefont {J.~G.}\ \bibnamefont {Bohnet}},\ }\href
  {https://doi.org/10.48550/arXiv.2511.05465} {\bibinfo {title} {Helios: {{A}}
  98-qubit trapped-ion quantum computer}} (\bibinfo {year} {2025}),\ \Eprint
  {https://arxiv.org/abs/2511.05465} {arXiv:2511.05465 [quant-ph]} \BibitemShut
  {NoStop}%
\bibitem [{\citenamefont {Bravyi}\ \emph {et~al.}(2024)\citenamefont {Bravyi},
  \citenamefont {Cross}, \citenamefont {Gambetta}, \citenamefont {Maslov},
  \citenamefont {Rall},\ and\ \citenamefont {Yoder}}]{bravyi2024highthreshold}%
  \BibitemOpen
  \bibfield  {author} {\bibinfo {author} {\bibfnamefont {S.}~\bibnamefont
  {Bravyi}}, \bibinfo {author} {\bibfnamefont {A.~W.}\ \bibnamefont {Cross}},
  \bibinfo {author} {\bibfnamefont {J.~M.}\ \bibnamefont {Gambetta}}, \bibinfo
  {author} {\bibfnamefont {D.}~\bibnamefont {Maslov}}, \bibinfo {author}
  {\bibfnamefont {P.}~\bibnamefont {Rall}},\ and\ \bibinfo {author}
  {\bibfnamefont {T.~J.}\ \bibnamefont {Yoder}},\ }\href
  {https://doi.org/10.1038/s41586-024-07107-7} {\bibfield  {journal} {\bibinfo
  {journal} {Nature}\ }\textbf {\bibinfo {volume} {627}},\ \bibinfo {pages}
  {778} (\bibinfo {year} {2024})}\BibitemShut {NoStop}%
\bibitem [{\citenamefont {Cain}\ \emph {et~al.}(2026)\citenamefont {Cain},
  \citenamefont {Xu}, \citenamefont {King}, \citenamefont {Picard},
  \citenamefont {Levine}, \citenamefont {Endres}, \citenamefont {Preskill},
  \citenamefont {Huang},\ and\ \citenamefont {Bluvstein}}]{cain2026shors}%
  \BibitemOpen
  \bibfield  {author} {\bibinfo {author} {\bibfnamefont {M.}~\bibnamefont
  {Cain}}, \bibinfo {author} {\bibfnamefont {Q.}~\bibnamefont {Xu}}, \bibinfo
  {author} {\bibfnamefont {R.}~\bibnamefont {King}}, \bibinfo {author}
  {\bibfnamefont {L.~R.~B.}\ \bibnamefont {Picard}}, \bibinfo {author}
  {\bibfnamefont {H.}~\bibnamefont {Levine}}, \bibinfo {author} {\bibfnamefont
  {M.}~\bibnamefont {Endres}}, \bibinfo {author} {\bibfnamefont
  {J.}~\bibnamefont {Preskill}}, \bibinfo {author} {\bibfnamefont {H.-Y.}\
  \bibnamefont {Huang}},\ and\ \bibinfo {author} {\bibfnamefont
  {D.}~\bibnamefont {Bluvstein}},\ }\href
  {https://doi.org/10.48550/arXiv.2603.28627} {\bibinfo {title} {Shor's
  algorithm is possible with as few as 10,000 reconfigurable atomic qubits}}
  (\bibinfo {year} {2026}),\ \Eprint {https://arxiv.org/abs/2603.28627}
  {arXiv:2603.28627 [quant-ph]} \BibitemShut {NoStop}%
\bibitem [{\citenamefont {Zhao}\ \emph {et~al.}(2026)\citenamefont {Zhao},
  \citenamefont {Duckering}, \citenamefont {Gu}, \citenamefont {Maskara},\ and\
  \citenamefont {Zhou}}]{zhao2026ultrahighrate}%
  \BibitemOpen
  \bibfield  {author} {\bibinfo {author} {\bibfnamefont {C.}~\bibnamefont
  {Zhao}}, \bibinfo {author} {\bibfnamefont {C.}~\bibnamefont {Duckering}},
  \bibinfo {author} {\bibfnamefont {A.}~\bibnamefont {Gu}}, \bibinfo {author}
  {\bibfnamefont {N.}~\bibnamefont {Maskara}},\ and\ \bibinfo {author}
  {\bibfnamefont {H.}~\bibnamefont {Zhou}},\ }\href
  {https://arxiv.org/abs/2604.16209} {\bibinfo {title} {Towards ultra-high-rate
  quantum error correction with reconfigurable atom arrays}} (\bibinfo {year}
  {2026}),\ \Eprint {https://arxiv.org/abs/2604.16209} {arXiv:2604.16209
  [quant-ph]} \BibitemShut {NoStop}%
\bibitem [{\citenamefont {Tansuwannont}\ \emph {et~al.}(2025)\citenamefont
  {Tansuwannont}, \citenamefont {Takada},\ and\ \citenamefont
  {Fujii}}]{tansuwannont2025clifford}%
  \BibitemOpen
  \bibfield  {author} {\bibinfo {author} {\bibfnamefont {T.}~\bibnamefont
  {Tansuwannont}}, \bibinfo {author} {\bibfnamefont {Y.}~\bibnamefont
  {Takada}},\ and\ \bibinfo {author} {\bibfnamefont {K.}~\bibnamefont
  {Fujii}},\ }\href {https://doi.org/10.48550/arXiv.2503.19790} {\bibinfo
  {title} {Clifford gates with logical transversality for self-dual {{CSS}}
  codes}} (\bibinfo {year} {2025}),\ \Eprint {https://arxiv.org/abs/2503.19790}
  {arXiv:2503.19790 [quant-ph]} \BibitemShut {NoStop}%
\bibitem [{\citenamefont {Sullivan}\ \emph {et~al.}(2026)\citenamefont
  {Sullivan}, \citenamefont {Amaro},\ and\ \citenamefont
  {Perlin}}]{sullivan2026injection}%
  \BibitemOpen
  \bibfield  {author} {\bibinfo {author} {\bibfnamefont {J.}~\bibnamefont
  {Sullivan}}, \bibinfo {author} {\bibfnamefont {D.}~\bibnamefont {Amaro}},\
  and\ \bibinfo {author} {\bibfnamefont {M.~A.}\ \bibnamefont {Perlin}},\
  }\href@noop {} {\bibfield  {journal} {\bibinfo  {journal} {In preparation}\ }
  (\bibinfo {year} {2026})}\BibitemShut {NoStop}%
\bibitem [{\citenamefont {Gottesman}(1997)}]{gottesman1997stabilizer}%
  \BibitemOpen
  \bibfield  {author} {\bibinfo {author} {\bibfnamefont {D.}~\bibnamefont
  {Gottesman}},\ }\href {https://doi.org/10.48550/arXiv.quant-ph/9705052}
  {\bibinfo {title} {Stabilizer {{Codes}} and {{Quantum Error Correction}}}}
  (\bibinfo {year} {1997}),\ \Eprint {https://arxiv.org/abs/quant-ph/9705052}
  {arXiv:quant-ph/9705052} \BibitemShut {NoStop}%
\bibitem [{\citenamefont {Gottesman}(2024)}]{gottesman2024surviving}%
  \BibitemOpen
  \bibfield  {author} {\bibinfo {author} {\bibfnamefont {D.}~\bibnamefont
  {Gottesman}},\ }\href@noop {} {\bibinfo {title} {Surviving as a {{Quantum
  Computer}} in a {{Classical World}}}} (\bibinfo {year} {2024})\BibitemShut
  {NoStop}%
\bibitem [{\citenamefont {Van~Laarhoven}\ and\ \citenamefont
  {Aarts}(1987)}]{vanlaarhoven1987simulated}%
  \BibitemOpen
  \bibfield  {author} {\bibinfo {author} {\bibfnamefont {P.~J.~M.}\
  \bibnamefont {Van~Laarhoven}}\ and\ \bibinfo {author} {\bibfnamefont
  {E.~H.~L.}\ \bibnamefont {Aarts}},\ }\href
  {https://doi.org/10.1007/978-94-015-7744-1} {\emph {\bibinfo {title}
  {Simulated {{Annealing}}: {{Theory}} and {{Applications}}}}}\ (\bibinfo
  {publisher} {Springer Netherlands},\ \bibinfo {address} {Dordrecht},\
  \bibinfo {year} {1987})\BibitemShut {NoStop}%
\bibitem [{\citenamefont {Sarvepalli}\ and\ \citenamefont
  {Klappenecker}(2010)}]{sarvepalli2010degenerate}%
  \BibitemOpen
  \bibfield  {author} {\bibinfo {author} {\bibfnamefont {P.}~\bibnamefont
  {Sarvepalli}}\ and\ \bibinfo {author} {\bibfnamefont {A.}~\bibnamefont
  {Klappenecker}},\ }\href {https://doi.org/10.1103/PhysRevA.81.032318}
  {\bibfield  {journal} {\bibinfo  {journal} {Physical Review A}\ }\textbf
  {\bibinfo {volume} {81}},\ \bibinfo {pages} {032318} (\bibinfo {year}
  {2010})}\BibitemShut {NoStop}%
\bibitem [{\citenamefont {Gottesman}(2004)}]{gottesman2004co}%
  \BibitemOpen
  \bibfield  {author} {\bibinfo {author} {\bibfnamefont {D.}~\bibnamefont
  {Gottesman}},\ }\href
  {https://www.cs.umd.edu/~dgottesm/CO639-2004/Lecture8.pdf} {\bibinfo {title}
  {{{CO}} 639: {{Quantum Error Correction}}, {{Lecture}} 8}} (\bibinfo {year}
  {2004})\BibitemShut {NoStop}%
\bibitem [{\citenamefont {Matsumoto}(2017)}]{matsumoto2017two}%
  \BibitemOpen
  \bibfield  {author} {\bibinfo {author} {\bibfnamefont {R.}~\bibnamefont
  {Matsumoto}},\ }\href {https://doi.org/10.1007/s11128-017-1748-y} {\bibfield
  {journal} {\bibinfo  {journal} {Quantum Information Processing}\ }\textbf
  {\bibinfo {volume} {16}},\ \bibinfo {pages} {285} (\bibinfo {year}
  {2017})}\BibitemShut {NoStop}%
\bibitem [{\citenamefont {Metropolis}\ \emph {et~al.}(1953)\citenamefont
  {Metropolis}, \citenamefont {Rosenbluth}, \citenamefont {Rosenbluth},
  \citenamefont {Teller},\ and\ \citenamefont
  {Teller}}]{metropolis1953equation}%
  \BibitemOpen
  \bibfield  {author} {\bibinfo {author} {\bibfnamefont {N.}~\bibnamefont
  {Metropolis}}, \bibinfo {author} {\bibfnamefont {A.~W.}\ \bibnamefont
  {Rosenbluth}}, \bibinfo {author} {\bibfnamefont {M.~N.}\ \bibnamefont
  {Rosenbluth}}, \bibinfo {author} {\bibfnamefont {A.~H.}\ \bibnamefont
  {Teller}},\ and\ \bibinfo {author} {\bibfnamefont {E.}~\bibnamefont
  {Teller}},\ }\href {https://doi.org/10.1063/1.1699114} {\bibfield  {journal}
  {\bibinfo  {journal} {The Journal of Chemical Physics}\ }\textbf {\bibinfo
  {volume} {21}},\ \bibinfo {pages} {1087} (\bibinfo {year}
  {1953})}\BibitemShut {NoStop}%
\bibitem [{\citenamefont {Hastings}(1970)}]{hastings1970monte}%
  \BibitemOpen
  \bibfield  {author} {\bibinfo {author} {\bibfnamefont {W.~K.}\ \bibnamefont
  {Hastings}},\ }\href {https://doi.org/10.1093/biomet/57.1.97} {\bibfield
  {journal} {\bibinfo  {journal} {Biometrika}\ }\textbf {\bibinfo {volume}
  {57}},\ \bibinfo {pages} {97} (\bibinfo {year} {1970})}\BibitemShut {NoStop}%
\bibitem [{\citenamefont {Tierney}(1994)}]{tierney1994markov}%
  \BibitemOpen
  \bibfield  {author} {\bibinfo {author} {\bibfnamefont {L.}~\bibnamefont
  {Tierney}},\ }\href@noop {} {\bibfield  {journal} {\bibinfo  {journal} {The
  Annals of Statistics}\ }\textbf {\bibinfo {volume} {22}},\ \bibinfo {pages}
  {1701} (\bibinfo {year} {1994})},\ \Eprint {https://arxiv.org/abs/2242477}
  {2242477} \BibitemShut {NoStop}%
\bibitem [{\citenamefont {Kirkpatrick}\ \emph {et~al.}(1983)\citenamefont
  {Kirkpatrick}, \citenamefont {Gelatt},\ and\ \citenamefont
  {Vecchi}}]{kirkpatrick1983optimization}%
  \BibitemOpen
  \bibfield  {author} {\bibinfo {author} {\bibfnamefont {S.}~\bibnamefont
  {Kirkpatrick}}, \bibinfo {author} {\bibfnamefont {C.~D.}\ \bibnamefont
  {Gelatt}},\ and\ \bibinfo {author} {\bibfnamefont {M.~P.}\ \bibnamefont
  {Vecchi}},\ }\href {https://doi.org/10.1126/science.220.4598.671} {\bibfield
  {journal} {\bibinfo  {journal} {Science}\ }\textbf {\bibinfo {volume}
  {220}},\ \bibinfo {pages} {671} (\bibinfo {year} {1983})}\BibitemShut
  {NoStop}%
\bibitem [{\citenamefont {Geman}\ and\ \citenamefont
  {Geman}(1984)}]{geman1984stochastic}%
  \BibitemOpen
  \bibfield  {author} {\bibinfo {author} {\bibfnamefont {S.}~\bibnamefont
  {Geman}}\ and\ \bibinfo {author} {\bibfnamefont {D.}~\bibnamefont {Geman}},\
  }\href {https://doi.org/10.1109/TPAMI.1984.4767596} {\bibfield  {journal}
  {\bibinfo  {journal} {IEEE Transactions on Pattern Analysis and Machine
  Intelligence}\ }\textbf {\bibinfo {volume} {PAMI-6}},\ \bibinfo {pages} {721}
  (\bibinfo {year} {1984})}\BibitemShut {NoStop}%
\bibitem [{\citenamefont {Barg}\ and\ \citenamefont
  {Forney}(2002)}]{barg2002random}%
  \BibitemOpen
  \bibfield  {author} {\bibinfo {author} {\bibfnamefont {A.}~\bibnamefont
  {Barg}}\ and\ \bibinfo {author} {\bibfnamefont {G.}~\bibnamefont {Forney}},\
  }\href {https://doi.org/10.1109/TIT.2002.800480} {\bibfield  {journal}
  {\bibinfo  {journal} {IEEE Transactions on Information Theory}\ }\textbf
  {\bibinfo {volume} {48}},\ \bibinfo {pages} {2568} (\bibinfo {year}
  {2002})}\BibitemShut {NoStop}%
\bibitem [{\citenamefont {MacWilliams}\ and\ \citenamefont
  {Sloane}(1977)}]{macwilliams1977theory}%
  \BibitemOpen
  \bibfield  {author} {\bibinfo {author} {\bibfnamefont {F.~J.}\ \bibnamefont
  {MacWilliams}}\ and\ \bibinfo {author} {\bibfnamefont {N.~J.~A.}\
  \bibnamefont {Sloane}},\ }\href@noop {} {\emph {\bibinfo {title} {The
  {{Theory}} of {{Error-correcting Codes}}}}}\ (\bibinfo  {publisher}
  {Elsevier},\ \bibinfo {year} {1977})\BibitemShut {NoStop}%
\bibitem [{\citenamefont {Hernando}\ \emph {et~al.}(2019)\citenamefont
  {Hernando}, \citenamefont {Igual},\ and\ \citenamefont
  {{Quintana-Ort{\'i}}}}]{hernando2019fast}%
  \BibitemOpen
  \bibfield  {author} {\bibinfo {author} {\bibfnamefont {F.}~\bibnamefont
  {Hernando}}, \bibinfo {author} {\bibfnamefont {F.~D.}\ \bibnamefont
  {Igual}},\ and\ \bibinfo {author} {\bibfnamefont {G.}~\bibnamefont
  {{Quintana-Ort{\'i}}}},\ }\bibfield  {journal} {\bibinfo  {journal} {ACM
  Transactions on Mathematical Software}\ }\textbf {\bibinfo {volume} {45}},\
  \href {https://doi.org/10.1145/3302389} {10.1145/3302389} (\bibinfo {year}
  {2019})\BibitemShut {NoStop}%
\bibitem [{\citenamefont {Hernando}\ \emph {et~al.}(2026)\citenamefont
  {Hernando}, \citenamefont {{Quintana-Ort{\'i}}},\ and\ \citenamefont
  {Grassl}}]{hernando2026fast}%
  \BibitemOpen
  \bibfield  {author} {\bibinfo {author} {\bibfnamefont {F.}~\bibnamefont
  {Hernando}}, \bibinfo {author} {\bibfnamefont {G.}~\bibnamefont
  {{Quintana-Ort{\'i}}}},\ and\ \bibinfo {author} {\bibfnamefont
  {M.}~\bibnamefont {Grassl}},\ }\href {https://doi.org/10.1145/3795877}
  {\bibfield  {journal} {\bibinfo  {journal} {ACM Transactions on Quantum
  Computing}\ }\textbf {\bibinfo {volume} {7}},\ \bibinfo {pages} {12}
  (\bibinfo {year} {2026})}\BibitemShut {NoStop}%
\bibitem [{\citenamefont {Perlin}(2023)}]{perlin2023qldpc}%
  \BibitemOpen
  \bibfield  {author} {\bibinfo {author} {\bibfnamefont {M.~A.}\ \bibnamefont
  {Perlin}},\ }\href@noop {} {\bibinfo {title} {qldpc}},\ \bibinfo
  {howpublished} {\url{https://github.com/qLDPCOrg/qLDPC}} (\bibinfo {year}
  {2023})\BibitemShut {NoStop}%
\bibitem [{\citenamefont {Greenfield}(1993)}]{greenfield1993simulated}%
  \BibitemOpen
  \bibfield  {author} {\bibinfo {author} {\bibfnamefont {G.}~\bibnamefont
  {Greenfield}},\ }\href@noop {} {\bibfield  {journal} {\bibinfo  {journal}
  {Department of Math \& Statistics Technical Report Series}\ } (\bibinfo
  {year} {1993})}\BibitemShut {NoStop}%
\bibitem [{\citenamefont {Kanchi}(2026)}]{kanchi2026tunnelingaugmented}%
  \BibitemOpen
  \bibfield  {author} {\bibinfo {author} {\bibfnamefont {A.}~\bibnamefont
  {Kanchi}},\ }\href {https://doi.org/10.48550/arXiv.2604.07365} {\bibinfo
  {title} {Tunneling-{{Augmented Simulated Annealing}} for {{Short-Block LDPC
  Code Construction}}}} (\bibinfo {year} {2026}),\ \Eprint
  {https://arxiv.org/abs/2604.07365} {arXiv:2604.07365} \BibitemShut {NoStop}%
\bibitem [{\citenamefont {Freire}\ \emph {et~al.}(2025)\citenamefont {Freire},
  \citenamefont {Delfosse},\ and\ \citenamefont
  {Leverrier}}]{freire2025optimizing}%
  \BibitemOpen
  \bibfield  {author} {\bibinfo {author} {\bibfnamefont {B.~C.~A.}\
  \bibnamefont {Freire}}, \bibinfo {author} {\bibfnamefont {N.}~\bibnamefont
  {Delfosse}},\ and\ \bibinfo {author} {\bibfnamefont {A.}~\bibnamefont
  {Leverrier}},\ }in\ \href {https://doi.org/10.1109/ISIT63088.2025.11195424}
  {\emph {\bibinfo {booktitle} {2025 {{IEEE International Symposium}} on
  {{Information Theory}} ({{ISIT}})}}}\ (\bibinfo {year} {2025})\ pp.\ \bibinfo
  {pages} {1--6}\BibitemShut {NoStop}%
\bibitem [{\citenamefont {Sayginel}\ \emph {et~al.}(2025)\citenamefont
  {Sayginel}, \citenamefont {Koutsioumpas}, \citenamefont {Webster},
  \citenamefont {Rajput},\ and\ \citenamefont
  {Browne}}]{sayginel2025faulttolerant}%
  \BibitemOpen
  \bibfield  {author} {\bibinfo {author} {\bibfnamefont {H.}~\bibnamefont
  {Sayginel}}, \bibinfo {author} {\bibfnamefont {S.}~\bibnamefont
  {Koutsioumpas}}, \bibinfo {author} {\bibfnamefont {M.}~\bibnamefont
  {Webster}}, \bibinfo {author} {\bibfnamefont {A.}~\bibnamefont {Rajput}},\
  and\ \bibinfo {author} {\bibfnamefont {D.~E.}\ \bibnamefont {Browne}},\
  }\href {https://doi.org/10.1103/vf7v-cpq9} {\bibfield  {journal} {\bibinfo
  {journal} {PRX Quantum}\ }\textbf {\bibinfo {volume} {6}},\ \bibinfo {pages}
  {030343} (\bibinfo {year} {2025})}\BibitemShut {NoStop}%
\bibitem [{\citenamefont {Dasu}\ and\ \citenamefont
  {Criger}(2026)}]{dasu2026flagging}%
  \BibitemOpen
  \bibfield  {author} {\bibinfo {author} {\bibfnamefont {S.}~\bibnamefont
  {Dasu}}\ and\ \bibinfo {author} {\bibfnamefont {B.}~\bibnamefont {Criger}},\
  }\href {https://arxiv.org/abs/2603.24573} {\bibinfo {title} {Flagging the
  clifford hierarchy: Fault-tolerant logical $\frac{\pi} {2^l}$ rotations via
  measuring circuit gauge operators of non-cliffords}} (\bibinfo {year}
  {2026}),\ \Eprint {https://arxiv.org/abs/2603.24573} {arXiv:2603.24573
  [quant-ph]} \BibitemShut {NoStop}%
\bibitem [{\citenamefont {Brun}\ \emph {et~al.}(2015)\citenamefont {Brun},
  \citenamefont {Zheng}, \citenamefont {Hsu}, \citenamefont {Job},\ and\
  \citenamefont {Lai}}]{brun2015teleportationbased}%
  \BibitemOpen
  \bibfield  {author} {\bibinfo {author} {\bibfnamefont {T.~A.}\ \bibnamefont
  {Brun}}, \bibinfo {author} {\bibfnamefont {Y.-C.}\ \bibnamefont {Zheng}},
  \bibinfo {author} {\bibfnamefont {K.-C.}\ \bibnamefont {Hsu}}, \bibinfo
  {author} {\bibfnamefont {J.}~\bibnamefont {Job}},\ and\ \bibinfo {author}
  {\bibfnamefont {C.-Y.}\ \bibnamefont {Lai}},\ }\href
  {https://doi.org/10.48550/arXiv.1504.03913} {\bibinfo {title}
  {Teleportation-based {{Fault-tolerant Quantum Computation}} in {{Multi-qubit
  Large Block Codes}}}} (\bibinfo {year} {2015}),\ \Eprint
  {https://arxiv.org/abs/1504.03913} {arXiv:1504.03913 [quant-ph]} \BibitemShut
  {NoStop}%
\bibitem [{\citenamefont {Forlivesi}\ and\ \citenamefont
  {Amaro}(2025)}]{forlivesi2025flag}%
  \BibitemOpen
  \bibfield  {author} {\bibinfo {author} {\bibfnamefont {D.}~\bibnamefont
  {Forlivesi}}\ and\ \bibinfo {author} {\bibfnamefont {D.}~\bibnamefont
  {Amaro}},\ }\href {https://doi.org/10.48550/arXiv.2508.14200} {\bibinfo
  {title} {Flag at origin: A modular fault-tolerant preparation for {{CSS}}
  codes}} (\bibinfo {year} {2025}),\ \Eprint {https://arxiv.org/abs/2508.14200}
  {arXiv:2508.14200 [quant-ph]} \BibitemShut {NoStop}%
\bibitem [{\citenamefont {Peham}\ \emph {et~al.}(2025)\citenamefont {Peham},
  \citenamefont {Schmid}, \citenamefont {Berent}, \citenamefont {M{\"u}ller},\
  and\ \citenamefont {Wille}}]{peham2025automated}%
  \BibitemOpen
  \bibfield  {author} {\bibinfo {author} {\bibfnamefont {T.}~\bibnamefont
  {Peham}}, \bibinfo {author} {\bibfnamefont {L.}~\bibnamefont {Schmid}},
  \bibinfo {author} {\bibfnamefont {L.}~\bibnamefont {Berent}}, \bibinfo
  {author} {\bibfnamefont {M.}~\bibnamefont {M{\"u}ller}},\ and\ \bibinfo
  {author} {\bibfnamefont {R.}~\bibnamefont {Wille}},\ }\href
  {https://doi.org/10.1103/PRXQuantum.6.020330} {\bibfield  {journal} {\bibinfo
   {journal} {PRX Quantum}\ }\textbf {\bibinfo {volume} {6}},\ \bibinfo {pages}
  {020330} (\bibinfo {year} {2025})}\BibitemShut {NoStop}%
\bibitem [{\citenamefont {Criger}\ \emph {et~al.}(2026)\citenamefont {Criger},
  \citenamefont {Hankin}, \citenamefont {Burau},\ and\ \citenamefont
  {Perry}}]{criger2026automated}%
  \BibitemOpen
  \bibfield  {author} {\bibinfo {author} {\bibfnamefont {B.}~\bibnamefont
  {Criger}}, \bibinfo {author} {\bibfnamefont {A.}~\bibnamefont {Hankin}},
  \bibinfo {author} {\bibfnamefont {J.~J.}\ \bibnamefont {Burau}},\ and\
  \bibinfo {author} {\bibfnamefont {A.~R.}\ \bibnamefont {Perry}},\ }\href
  {https://doi.org/10.48550/arXiv.2607.22498} {\bibinfo {title} {Automated
  {{Flag-based Fault-Tolerant State Preparation}} using {{Integer Linear
  Programming}}}} (\bibinfo {year} {2026}),\ \Eprint
  {https://arxiv.org/abs/2607.22498} {arXiv:2607.22498 [quant-ph]} \BibitemShut
  {NoStop}%
\bibitem [{\citenamefont {Henderson}\ and\ \citenamefont
  {Searle}(1981)}]{henderson1981deriving}%
  \BibitemOpen
  \bibfield  {author} {\bibinfo {author} {\bibfnamefont {H.~V.}\ \bibnamefont
  {Henderson}}\ and\ \bibinfo {author} {\bibfnamefont {S.~R.}\ \bibnamefont
  {Searle}},\ }\href@noop {} {\bibfield  {journal} {\bibinfo  {journal} {SIAM
  review}\ }\textbf {\bibinfo {volume} {23}},\ \bibinfo {pages} {53} (\bibinfo
  {year} {1981})}\BibitemShut {NoStop}%
\bibitem [{\citenamefont {Vishwakarma}(2025)}]{vishwakarma2025cholesky}%
  \BibitemOpen
  \bibfield  {author} {\bibinfo {author} {\bibfnamefont {P.~K.}\ \bibnamefont
  {Vishwakarma}},\ }\href {https://doi.org/10.48550/arXiv.2508.04657} {\bibinfo
  {title} {Cholesky decomposition for symmetric matrices over finite fields}}
  (\bibinfo {year} {2025}),\ \Eprint {https://arxiv.org/abs/2508.04657}
  {arXiv:2508.04657 [math]} \BibitemShut {NoStop}%
\bibitem [{\citenamefont {Golub}\ and\ \citenamefont
  {Loan}(2013)}]{golub2013matrix}%
  \BibitemOpen
  \bibfield  {author} {\bibinfo {author} {\bibfnamefont {G.~H.}\ \bibnamefont
  {Golub}}\ and\ \bibinfo {author} {\bibfnamefont {C.~F.~V.}\ \bibnamefont
  {Loan}},\ }\href@noop {} {\emph {\bibinfo {title} {Matrix
  {{Computations}}}}}\ (\bibinfo  {publisher} {JHU Press},\ \bibinfo {year}
  {2013})\BibitemShut {NoStop}%
\bibitem [{\citenamefont {Tansuwannont}\ and\ \citenamefont
  {Nemec}(2024)}]{tansuwannont2024synchronizable}%
  \BibitemOpen
  \bibfield  {author} {\bibinfo {author} {\bibfnamefont {T.}~\bibnamefont
  {Tansuwannont}}\ and\ \bibinfo {author} {\bibfnamefont {A.}~\bibnamefont
  {Nemec}},\ }\href {https://doi.org/10.48550/arXiv.2409.11312} {\bibinfo
  {title} {Synchronizable hybrid subsystem codes}} (\bibinfo {year} {2024}),\
  \Eprint {https://arxiv.org/abs/2409.11312} {arXiv:2409.11312 [quant-ph]}
  \BibitemShut {NoStop}%
\end{thebibliography}%

\appendix
\crefalias{section}{appendix}

\section{CSS equivalence theorem}
\label{sec:css_equivalence}

Here we prove \cref{thm:css_equivalence} of the main text, restated here for reference:
\begin{tcolorbox}
  \CSSEquivalence*
\end{tcolorbox}

To prove this theorem, we first reduce it to that of decomposing an $n$-qubit CSS-preserving gate $\tau$ into CNOT gates.

\begin{tcolorbox}
  \begin{lemma}
    Let $C = \CSS(H_X, H_Z)$ and $C' = \CSS(H_X', H_Z')$.
    If
    \begin{enumerate}[label=(\roman*)]
      \item $\rank(H_X) = \rank(H_X')$, and
      \item $\rank(H_Z) = \rank(H_Z')$,
    \end{enumerate}
    then $C' = \tau(C)$ for some CSS-preserving gate $\tau$.
    \label{lemma:css_equivalence}
  \end{lemma}
  \begin{proof}
    Let $\tau_C$ and $\tau_{C'}$ be CSS-preserving gates that construct $C$ and $C'$ from the trivial CSS code, in the sense of \cref{thm:css_construction}.
    Then $\tau = \tau_{C'} \tau_C^{-1}$ is a CSS-preserving gate for which $C' = \tau(C)$.
  \end{proof}
\end{tcolorbox}
A decomposition of $\tau$ in \cref{lemma:css_equivalence} into CNOT gates recovers the gate sequence for \cref{thm:css_equivalence}.
The Clifford tableau of the CNOT gate on two qubits is
\begin{align}
  \CNOT =
  \left(
    \begin{array}{cc|cc}
      1 & & & \\
      1 & 1 & & \\ \hline
      & & 1 & 1 \\
      & & & 1
    \end{array}
  \right),
\end{align}
where the empty entries are zero.

Letting $\tau\simeq(\tau_X, \tau_Z)$, we now consider the matrix $\tau_X\in\GL(n,\F_2)$.
Left-multiplying $\tau$ by $\CNOT_{i,j}$ adds the $i$-th row of $\tau_X$ to the $j$-th row of $\tau$.
The CNOT gate can thereby be used to row-reduce $\tau_X$ to a permutation matrix via Gaussian elimination.
In turn, this permutation matrix can be reduced to the identity matrix $\1_n$ via SWAP gates, which can be decomposed into CNOT gates.
By construction, the resulting net sequence of CNOT gates $(G_1, G_2, \cdots, G_N)$ has the property that the CSS-preserving gate
\begin{align}
  \sigma = G_N \cdots G_2 G_1 \tau \simeq (\sigma_X, \sigma_Z),
\end{align}
has $\sigma_X = \1_n$, and in turn $\sigma_Z = \sigma_X^{-\T} = \1_n$ by \cref{lemma:css_clifford}.
We can therefore decompose
\begin{align}
  \tau = (G_N \cdots G_2 G_1)^{-1} = G_1 G_2 \cdots G_N,
\end{align}
which proves \cref{thm:css_equivalence}.

\section{CSS construction theorem}
\label{sec:css_construction}

Here we prove \cref{thm:css_construction} of the main text, restated here for reference:
\begin{tcolorbox}
  \CSSConstruction*
\end{tcolorbox}

We prove this theorem constructively.

Without loss of generality, we can permute qubits and perform Gaussian elimination to put $H_X$ and $H_Z$ in the standard forms~\cite[Section 4.1]{gottesman1997stabilizer}
\begin{align}
  H_X &=
  \begin{pmatrix}
    \ell_X & \1_{s_X} & R_X
  \end{pmatrix},
  \\
  H_Z &=
  \begin{pmatrix}
    \ell_Z & R_Z & \1_{s_Z}
  \end{pmatrix},
\end{align}
with matrices $(\ell_X,\ell_Z,R_X,R_Z)$ of appropriate dimensions.
Here $\rank(H_X) = s_X$ and $\rank(H_Z) = s_Z$.
To map the trivial code $\TrivCSS(k, s_X, s_Z)$ (\cref{def:trivial_css}) to $C$, we need to construct a CSS-preserving gate $\tau$ that maps the parity check matrices
\begin{align}
  H_X^{\text{Triv}} &=
  \begin{pmatrix}
    \bm{0} & \1_{s_X} & \bm{0}
  \end{pmatrix},
  \\
  H_Z^{\text{Triv}} &=
  \begin{pmatrix}
    \bm{0} & \bm{0} & \1_{s_Z}
  \end{pmatrix},
\end{align}
to $H_X$ and $H_Z$, meaning
\begin{align}
  H_X &= H_X^{\text{Triv}} \tau_X^\T,
  \\
  H_Z &= H_Z^{\text{Triv}} \tau_Z^\T.
\end{align}
These conditions imply that
\begin{align}
  \tau_X &=
  \begin{pmatrix}
    A_X & H_X^\T & B_X
  \end{pmatrix},
  \\
  \tau_Z &=
  \begin{pmatrix}
    A_Z & B_Z & H_Z^\T
  \end{pmatrix},
\end{align}
for some matrices $A_X,A_Z,B_X,B_Z$, which must in turn satisfy the constraint
\begin{align}
  \tau_X^\T \tau_Z = \tau_Z^{-1} \tau_Z = \1,
\end{align}
or equivalently
\begin{align}
  \begin{pmatrix}
    \cellcolor{red!40} A_X^\T A_Z   & \cellcolor{orange!40} A_X^\T B_Z & \cellcolor{orange!40} A_X^\T H_Z^\T \\
    \cellcolor{green!40} H_X A_Z    & \cellcolor{blue!40} H_X B_Z      & \cellcolor{purple!40} H_X H_Z^\T \\
    \cellcolor{green!40} B_X^\T A_Z & \cellcolor{red!40} B_X^\T B_Z    & \cellcolor{blue!40} B_X^\T H_Z^\T
  \end{pmatrix}
  =
  \begin{pmatrix}
    \1     & \bm{0} & \bm{0} \\
    \bm{0} & \1     & \bm{0} \\
    \bm{0} & \bm{0} & \1
  \end{pmatrix},
\end{align}
where we use colors as a visual aid to identify different sectors of this constraint.
We can satisfy $\mhl{blue!40}{H_X B_Z} = \1$ and $\mhl{blue!40}{H_Z B_X} = \1$ by choosing
\begin{align}
  B_Z =
  \begin{pmatrix}
    \bm{0} \\ \1 \\ \bm{0}
  \end{pmatrix}
  &&
  \text{and}
  &&
  B_X =
  \begin{pmatrix}
    \bm{0} \\ \bm{0} \\ \1
  \end{pmatrix},
\end{align}
and in turn satisfy
\begin{align}
  \mhl{orange!40}{B_Z^\T A_X} &= \bm{0} & \mhl{green!40}{B_X^\T A_Z} &= \bm{0} \\
  \mhl{orange!40}{H_Z A_X}    &= \bm{0} & \mhl{green!40}{H_X A_Z}    &= \bm{0}
\end{align}
by choosing
\begin{align}
  A_X =
  \begin{pmatrix}
    \1 \\ \bm{0} \\ \ell_Z
  \end{pmatrix}
  &&
  \text{and}
  &&
  A_Z =
  \begin{pmatrix}
    \1 \\ \ell_X \\ \bm{0}
  \end{pmatrix}.
\end{align}
As a sanity check, we note that these choices satisfy $\mhl{red!40}{A_X^\T A_Z} = \1$ and $\mhl{red!40}{B_X^\T B_Z} = \bm{0}$, as required, and that $\mhl{purple!40}{H_X H_Z^\T} = \bm{0}$.

Altogether, \cref{thm:css_construction} follows from the fact that $C = \tau(\TrivCSS(k, s_X, s_Z))$ for the CSS-preserving gate $\tau \simeq (\tau_X, \tau_Z)$ with
\begin{align}
  \tau_X =
  \begin{pmatrix}
    \1     & \ell_X^\T & \bm{0} \\
    \bm{0} & \1        & \bm{0} \\
    \ell_Z & R_X^\T    & \1
  \end{pmatrix},
  &&
  \tau_Z =
  \begin{pmatrix}
    \1     & \bm{0} & \ell_Z^\T \\
    \ell_X & \1     & R_Z^\T    \\
    \bm{0} & \bm{0} & \1
  \end{pmatrix}.
\end{align}

\section{SWEL stabilization theorem}
\label{sec:swel_stabilization}

Here we prove \cref{thrm:swel_stabilization} of the main text, restated here for reference:
\begin{tcolorbox}
  \SWELStabilization*
\end{tcolorbox}
The sufficiency of $\tau_X = \tau_Z = \tau_\star$ is obvious: if $C = \CSS(H_\star, H_\star)$ is a SWEL code and $(L_\star, L_\star)$ symplectic basis of logical operators for $C$, then the parity check matrices $H_\star$ and logicals $L_\star$ transform identically under $\tau$, making $\tau(C)$ a SWEL code.

We now need to prove the claim of \emph{necessity} in \cref{thrm:swel_stabilization}: that if
\begin{enumerate}[label=(\roman*)]
  \item\label{assumption:1} $C$ is an arbitrary SWEL code,
  \item\label{assumption:2} $\tau\simeq(\tau_X, \tau_Z)$ is a CSS-preserving gate, and
  \item\label{assumption:3} $\tau(C)$ is a SWEL code,
\end{enumerate}
then $\tau_X = \tau_Z$.

\subsection{Requirements}

To this end, let $C = \CSS(H_\star, H_\star)$ be a SWEL code for some parity check matrix $H_\star$, and let $(L_\star, L_\star)$ be a symplectic basis of logical operators for $C$.
The matrices $H_\star$ and $L_\star$ must satisfy
\begin{align}
  H_\star H_\star^\T = \bm{0},
  &&
  H_\star L_\star^\T = \bm{0},
  &&
  L_\star L_\star^\T = \1.
  \label{eq:swel_conditions}
\end{align}
Let $\tau$ be a CSS-preserving gate for which $\tau(C)$ is a SWEL code, as in Assumptions \ref{assumption:2} and \ref{assumption:2}.
The requirement that $X$-type stabilizers and logical operators of $\tau(C)$ commute with its $Z$-type stabilizers translates, on the level of binary matrices, to the requirement that rows of $H_\star \tau_X^\T$ and $L_\star \tau_X^\T$ live in the kernel of $H_\star\tau_Z^\T$, meaning
\begin{align}
  (H_\star \tau_Z^\T) (H_\star \tau_X^\T)^\T = \bm{0},
  &&
  (H_\star \tau_Z^\T) (L_\star \tau_X^\T)^\T = \bm{0}.
\end{align}
By Assumption \ref{assumption:3}, $\tau(C)$ is a self-dual code, meaning
\begin{align}
  \ker(H_\star\tau_X^\T) = \ker(H_\star\tau_Z^\T).
\end{align}
The rows of $H_\star \tau_X^\T$ and $L_\star \tau_X^\T$ must therefore live in the kernel of $H_\star\tau_X^\T$, so
\begin{align}
  (H_\star \tau_X^\T) (H_\star \tau_X^\T)^\T
  = H_\star \tau_X^\T \tau_X H_\star^\T
  &= \bm{0},
  \label{eq:tau_C_HH} \\
  (H_\star \tau_X^\T) (L_\star \tau_X^\T)^\T
  = H_\star \tau_X^\T \tau_X L_\star^\T
  &= \bm{0}.
  \label{eq:tau_C_HL}
\end{align}
The requirement that $\tau(C)$ be a SWEL code for all SWEL codes $C$ is therefore equivalent to the condition that $\tau \simeq (\tau_X, \tau_Z)$ satisfy \cref{eq:tau_C_HH} and \cref{eq:tau_C_HL} for all matrices $(H_\star, L_\star)$ that satisfy \cref{eq:swel_conditions}.

\subsection{Four possibilities}

Consider the symmetric, invertible matrix $A = \tau_X^\T \tau_X$.
Letting $e_q\in\F_2^n$ with $q\in\Z_n$ denote a standard basis vector for $\F_2^n$, we can pick matrices $H_\star$ and $L_\star$ for which $e_i + e_j \in \rows(H_\star)$ and $e_k \in \rows(L_\star)$ with distinct integers $i,j,k\in\Z_n$, as for example in the trivial SWEL code (\cref{def:trivial_swel})\footnote{
  This choice only fails when $n < 3$, and the only SWEL code with $n < 3$ is the trivial one-qubit code for which the only CSS-preserving gate is the identity gate $\tau$ for which $\tau_X = \tau_Z = 1$.
}.
The condition in \cref{eq:tau_C_HH} then implies that
\begin{align}
  (e_i + e_j)^\T A (e_i + e_j) = A_{ii} + A_{jj} = 0,
\end{align}
where the cross term $A_{ij} + A_{ji}$ vanishes over $\F_2$ because $A_{ij} = A_{ji}$.
The freedom in choosing $i,j\in\Z_n$ thereby implies that all diagonal entries of $A$ are equal.
Similarly, the condition in \cref{eq:tau_C_HL} implies that
\begin{align}
  (e_i + e_j)^\T A e_k = A_{ik} + A_{jk} = 0.
\end{align}
The freedom of choosing $i,k,k\in\Z_n$ together with the fact that $A = A^\T$ thereby implies that all off-diagonal entries of $A$ are equal.
There are therefore only four possibilities:
\begin{enumerate}[label=(\alph*)]
  \item\label{possibility:00} $A$ is an all-zeros matrix.
  \item\label{possibility:11} $A$ is an all-ones matrix.
  \item\label{possibility:01} $A$ is 0 on the diagonal and 1 everywhere else.
  \item\label{possibility:10} $A$ is the identity matrix.
\end{enumerate}
The fact that $A$ is invertible means it has full rank, which rules out possibilities \ref{possibility:00} and \ref{possibility:11}.

\subsection{Ruling out possibility \ref{possibility:01}}

We now consider the requirement that $A$ must admit the decomposition $A = \tau_X^\T \tau_X$ with an invertible matrix $\tau_X \in \F_2^{n\times n}$.
To this end, we observe that if $A = \tau_X^\T \tau_X$, then
\begin{align}
  A_{ii}
  = \sum_j (\tau_X)_{ji}^2
  = \sum_j (\tau_X)_{ji}
  = e_i^\T \tau_X^\T \bm{1}_n,
\end{align}
where $\bm{1}_n \in \F_2^n$ is the all-ones vector.
It follows that
\begin{align}
  \text{($A_{ii} = 0$ for all $i\in\Z_n$)}
  \implies
  \tau_X^\T \bm{1}_n = \bm{0}_n,
\end{align}
where $\bm{0}_n \in \F_2^n$ is the all-zeros vector.
However, $\tau_X^\T \bm{1}_n = \bm{0}_n$ contradicts the requirement that $\tau_X$ (and hence $\tau_X^\T$) is invertible, so it cannot be the case that $A_{ii} = 0$, thereby ruling out possibility \ref{possibility:01}.

Altogether, we are only left with possibility \ref{possibility:10}: that $A = \tau_X^\T \tau_X = \1$, so $\tau_X^\T = \tau_X^{-1}$.
By \cref{lemma:css_clifford}, it is also the case that $\tau_X^\T = \tau_Z^{-1}$, so $\tau_X = \tau_Z$, as stated in \cref{thrm:swel_stabilization}.

\section{SWEL construction theorem}
\label{sec:swel_construction}

Here we prove \cref{thm:swel_construction} of the main text, restated here for reference:
\begin{tcolorbox}
  \SWELConstruction*
\end{tcolorbox}

We prove this theorem constructively.

Without loss of generality, we can permute qubits and perform Gaussian elimination to put $H_\star$ in the standard form~\cite[Section 4.1]{gottesman1997stabilizer}
\begin{align}
  H_\star =
  \begin{pmatrix}
    \ell & R & \1_s
  \end{pmatrix},
  \label{eq:H_lR1}
\end{align}
where $\ell$ has dimensions $s\times k$ and $R$ has dimensions $s\times s$.
Here $\rank(H_\star) = s$.
Performing Gaussian elimination on the second block of $H_\star$ yields an alternative standard-form parity check matrix,
\begin{align}
  \tilde H_\star
  =
  \begin{pmatrix}
    \tilde\ell & \1_s & \tilde R
  \end{pmatrix},
\end{align}
where
\begin{align}
  \tilde R = R^{-1}
  &&
  \text{and}
  &&
  \tilde\ell = R^{-1} \ell.
\end{align}
CSS-compatibility of $H_\star$ with itself implies that
\begin{align}
  H_\star H_\star^\T = \ell \ell^\T + R R^\T + \1_s = \bm{0}_{s\times s}.
  \label{eq:HHT}
\end{align}
A Gottesman-canonical symplectic basis of logical operators for $C$ is then $(L_X, L_Z)$, with
\begin{align}
  L_X =
  \begin{pmatrix}
    \1 & \tilde\ell^\T & \bm{0}
  \end{pmatrix},
  &&
  L_Z =
  \begin{pmatrix}
    \1 & \bm{0} & \ell^\T
  \end{pmatrix}.
  \label{eq:swel_basis}
\end{align}

\subsection{Requirements}

To map the trivial code $\TrivSWEL(k, s)$ (\cref{def:trivial_swel}) to $C$, we need to construct a SWEL-preserving gate $\tau$ that maps the parity check matrix
\begin{align}
  H_\star^{\text{Triv}} =
  \begin{pmatrix}
    \bm{0}_{s\times k} & \1_s & \1_s
  \end{pmatrix}
\end{align}
to $H_\star$, meaning $H_\star^{\text{Triv}} \tau_\star^\T = H_\star$.
Similarly, the inverse gate $\tau^{-1} = \tau^\T$ needs to map $H_\star \to H_\star^{\text{Triv}}$, meaning $H_\star \tau_\star = H_\star^{\text{Triv}}$.
Together, these conditions imply that
\begin{align}
  \tau_\star
  \begin{pmatrix}
    \bm{0} \\ \1 \\ \1
  \end{pmatrix}
  =
  \begin{pmatrix}
    \ell^\T \\ R^\T \\ \1
  \end{pmatrix}
  \label{eq:tau_cond_1}
\end{align}
and
\begin{align}
  \begin{pmatrix}
    \ell & R & \1
  \end{pmatrix}
  \tau_\star
  =
  \begin{pmatrix}
    \bm{0} & \1 & \1
  \end{pmatrix}.
  \label{eq:tau_cond_2}
\end{align}
For $\tau$ to be a SWEL-preserving gate, the matrix $\tau_\star$ additionally needs to be orthogonal, meaning
\begin{align}
  \tau_\star^\T \tau_\star = \tau_\star \tau_\star^\T = \1.
  \label{eq:tau_cond_3}
\end{align}

\subsection{Ansatz and parity check constraints}

We now build a matrix $\tau_\star$ that satisfies all necessary requirements.
To this end, we guess a solution of the form
\begin{align}
  \tau_\star =
  \begin{pmatrix}
    A      & J      & \bm{0} \\
    B      & K      & \bm{0} \\
    \bm{0} & \bm{0} & \1
  \end{pmatrix},
  \label{eq:ansatz}
\end{align}
and find suitable matrices $A,B,J,K$.
Note that the top row and left-most column of the block matrix in \cref{eq:ansatz} have size $k$, whereas all other rows and columns have size $s$.

\cref{eq:tau_cond_1} implies that
\begin{align}
  J = \ell^\T
  &&
  \text{and}
  &&
  K = R^\T.
\end{align}
In turn, \cref{eq:tau_cond_2} reduces to
\begin{align}
  \begin{pmatrix}
    \ell & R
  \end{pmatrix}
  \begin{pmatrix}
    A & \ell^\T \\ B & R^\T
  \end{pmatrix}
  =
  \begin{pmatrix}
    \bm{0} & \1
  \end{pmatrix},
\end{align}
where the right column ($\ell\ell^\T + R R^\T = \1$) is automatically satisfied by \cref{eq:HHT}, while the left column ($\ell A + R B = \bm{0}$) implies that
\begin{align}
  B = R^{-1} \ell A.
\end{align}
In total, enforcing that the ansatz in \cref{eq:ansatz} transforms stabilizers as necessary forces
\begin{align}
  \tau_\star =
  \begin{pmatrix}
    A             & \ell^\T & \bm{0} \\
    R^{-1} \ell A & R^\T    & \bm{0} \\
    \bm{0}        & \bm{0}  & \1
  \end{pmatrix},
\end{align}
for some $k\times k$ matrix $A$.
The choice of $A$ is then constrained only by the orthogonality of $\tau_\star$ (\cref{eq:tau_cond_3}).

\subsection{Orthogonality constraints}

To find a suitable choice $A$, we enforce $\tau_\star^\T \tau_\star = \1$, which implies
\begin{align}
  \begin{pmatrix}
    A^\T & B^\T \\ \ell & R
  \end{pmatrix}
  \begin{pmatrix}
    A & \ell^\T \\ B & R^\T
  \end{pmatrix}
  =
  \begin{pmatrix}
    \1 & \bm{0} \\ \bm{0} & \1
  \end{pmatrix},
  \label{eq:tTt}
\end{align}
and similarly enforce $\tau_\star \tau_\star^\T = \1$, which implies
\begin{align}
  \begin{pmatrix}
    A & \ell^\T \\ B & R^\T
  \end{pmatrix}
  \begin{pmatrix}
    A^\T & B^\T \\ \ell & R
  \end{pmatrix}
  =
  \begin{pmatrix}
    \1 & \bm{0} \\ \bm{0} & \1
  \end{pmatrix}.
  \label{eq:ttT}
\end{align}
The only nontrivial constraints in \cref{eq:tTt} and \cref{eq:ttT} are
\begin{align}
  A^\T A + B^\T B &= \1,
  \label{eq:ATA_BTB}
  \\
  A A^\T + \ell^\T \ell &= \1,
  \label{eq:AAT}
  \\
  A B^\T + \ell^\T R &= \bm{0},
  \label{eq:ABT}
  \\
  BB^\T + R^\T R &= \1.
  \label{eq:BBT}
\end{align}
In fact, we only need to worry about one of these constraints: \cref{eq:AAT} implies \cref{eq:ATA_BTB}, \cref{eq:ABT}, and \cref{eq:BBT}.
We show the recovery of \cref{eq:ATA_BTB}, \cref{eq:ABT}, and \cref{eq:BBT} from \cref{eq:AAT} before returning to the question of finding a solution to \cref{eq:AAT}.

\subsection{Recovering \cref{eq:BBT}}

Substituting $B = R^{-1} \ell A$ and $A A^\T = \1 + \ell^\T \ell$ from \cref{eq:AAT}, we can expand
\begin{align}
  B B^\T
  &= R^{-1} \ell A A^\T \ell^\T R^{-\T} \\
  &= R^{-1} \ell (\1 + \ell^\T \ell) \ell^\T R^{-\T} \\
  &= R^{-1} (\1 + \ell \ell^\T) \ell \ell^\T R^{-\T},
\end{align}
and substitute $\1 + \ell \ell^\T = R R^\T$ from \cref{eq:HHT} to get
\begin{align}
  B B^\T
  = R^{-1} R R^\T (\1 + R R^\T) R^{-\T}
  = \1 + R^\T R,
\end{align}
which recovers \cref{eq:BBT}.

\subsection{Recovering \cref{eq:ABT}}

Substituting $B = R^{-1} \ell A$ and $A A^\T = \1 + \ell^\T \ell$ from \cref{eq:AAT}, we can expand
\begin{align}
  A B^\T
  &= A A^\T \ell^\T R^{-\T} \\
  &= (\1 + \ell^\T \ell) \ell^\T R^{-\T} \\
  &= \ell^\T (\1 + \ell \ell^\T) R^{-\T},
\end{align}
and substitute $\1 + \ell \ell^\T = R R^\T$ from \cref{eq:HHT} to get
\begin{align}
  A B^\T
  = \ell^\T R R^\T R^{-\T}
  = \ell^\T R,
\end{align}
which recovers \cref{eq:ABT}.

\subsection{Recovering \cref{eq:ATA_BTB}}

Substituting $B = R^{-1} \ell A$, we can expand
\begin{align}
  B^\T B
  &= A^\T \ell^\T R^{-\T} R^{-1} \ell A \\
  &= A^\T \ell^\T (R R^\T)^{-1} \ell A,
\end{align}
and substitute $R R^\T = \1 + \ell \ell^\T$ from \cref{eq:HHT} to get
\begin{align}
  B^\T B = A^\T \ell^\T (\1 + \ell \ell^\T)^{-1} \ell A.
\end{align}
Applying the push-through identity~\cite[Eq.~(11) with $(P, Q) = (\ell^\T, \ell)$]{henderson1981deriving}
\begin{align}
  \ell^\T (\1 + \ell \ell^\T)^{-1} &= (\1 + \ell^\T \ell)^{-1} \ell^\T,
  \label{eq:push_through}
\end{align}
we can simplify
\begin{align}
  B^\T B = A^\T (\1 + \ell^\T \ell)^{-1} \ell^\T \ell A,
\end{align}
and substitute $\1 + \ell^\T \ell = A A^\T$ from \cref{eq:AAT} to find
\begin{align}
  B^\T B
  &= A^\T (A A^\T)^{-1} (\1 + A A^\T) A \\
  &= A^\T (A A^\T)^{-1} A (\1 + A^\T A).
  \label{eq:ATA_BTB_recovery}
\end{align}
We know that $R R^\T = \1 + \ell \ell^\T$ has full rank, which means that $\1 + \ell^\T \ell = A A^\T$ must have full rank (making the push-through identity in \cref{eq:push_through} possible), and in turn that $\rank(A A^\T) = \rank(A) = k$.
The matrix $A$ therefore has dimensions $k\times k$ and rank $k$, which implies that $A$ is invertible, so $(A A^\T)^{-1} = A^{-\T} A^{-1}$, and
\begin{align}
  B^\T B
  = A^\T A^{-\T} A^{-1} A (\1 + A^\T A)
  = \1 + A^\T A,
\end{align}
which recovers \cref{eq:ATA_BTB}.

\subsection{A SWEL solution to \cref{eq:AAT}}
\label{sec:swel_solution}

We now seek a matrix $A$ that satisfies \cref{eq:AAT}, or
\begin{align}
  A A^\T = \1 + \ell^\T \ell.
  \label{eq:cholesky}
\end{align}
This equation is essentially the Cholesky decomposition of the (symmetric, invertible) matrix $\1 + \ell^\T \ell$.
Generally speaking, if $M$ is a matrix over a finite field, then the Cholesky decomposition of $M$ exists when every leading principal minor of $M$ is nonsingular~\cite{vishwakarma2025cholesky} (which is analogous to the ordinary conditions for the existence of an $LDL^\T$ decomposition~\cite[Theorem 4.1.3]{golub2013matrix}).

We have thus far only used the fact that $C$ is a self-dual CSS code.
We now invoke the fact that $C$ is SWEL to prove the existence of a matrix $A$ that satisfies \cref{eq:cholesky}.

Let $(L_X, L_Z)$ be the symplectic basis of logical operators in \cref{eq:swel_basis}.
If $C$ is a SWEL code, then there exists a $k$-qubit CSS-preserving gate $\sigma \simeq (\sigma_X, \sigma_Z)$ (that is, a logical change of basis) for which $\sigma_X L_X$ and $\sigma_Z L_Z$ are equal modulo stabilizers; that is,
\begin{align}
  \sigma_X L_X = \sigma_Z L_Z + \rho_\sigma H_\star
\end{align}
for some $k\times s$ matrix $\rho_\sigma$.
Substituting $L_X$, $L_Z$, and $H_\star$ from \cref{eq:H_lR1} and \cref{eq:swel_basis},
\begin{align}
  \sigma_X
  \begin{pmatrix}
    \1 & \tilde\ell^\T & \bm{0}
  \end{pmatrix}
  = \sigma_Z
  \begin{pmatrix}
    \1 & \bm{0} & \ell^\T
  \end{pmatrix}
  + \rho_\sigma
  \begin{pmatrix}
    \ell & R & \1
  \end{pmatrix},
\end{align}
where the third column implies that $\rho_\sigma = \sigma_Z \ell^\T$, and in turn the first column implies that
\begin{align}
  \sigma_X = \sigma_Z + \sigma_Z \ell^\T \ell = \sigma_Z (\1 + \ell^\T \ell),
\end{align}
so
\begin{align}
  \1 + \ell^\T \ell
  = \sigma_Z^{-1} \sigma_X
  = \sigma_X^\T \sigma_X.
\end{align}
Choosing $A = \sigma_X^\T$ therefore satisfies \cref{eq:cholesky}, which means that SWEL codes are guaranteed to have a matrix $A$ that satisfies \cref{eq:cholesky}.
This matrix can be found using a recursive algorithm outlined in the proof of Theorem A(2) in Ref.~\cite[Section 2]{vishwakarma2025cholesky}.
Alternatively, one can find a symplectic basis $(L_X', L_Z')$ for $C$ with $L_X' = L_Z'$ using Algorithm 1 of Ref.~\cite{tansuwannont2025clifford} (first introduced in Ref.~\cite[Algorithm 1]{tansuwannont2024synchronizable}), construct a CSS-preserving gate $\sigma \simeq (\sigma_X, \sigma_Z)$ for which $\sigma_X L_X \sim L_X'$ (where $\sim$ denotes equality modulo stabilizers), and set $A = \sigma_X^\T$.

\subsection{Summary}

In total, if $C = \CSS(H_\star, H_\star)$ is a SWEL code, then without loss of generality (that is, modulo qubit permutations and Gaussian elimination)
\begin{align}
  H_\star =
  \begin{pmatrix}
    \ell & R & \1
  \end{pmatrix},
\end{align}
where $\ell$ has dimensions $s\times k$, $R$ has dimensions $s\times s$, and $R$ is invertible.
In this case, it is possible to find a $k\times k$ matrix $A$ for which $A A^\T = \1 + \ell^\T \ell$ (see \cref{sec:swel_solution}), out of which we can construct the $n\times n$ matrix
\begin{align}
  \tau_\star =
  \begin{pmatrix}
    A             & \ell^\T & \bm{0} \\
    R^{-1} \ell A & R^\T    & \bm{0} \\
    \bm{0}        & \bm{0}  & \1
  \end{pmatrix}
\end{align}
which satisfies
\begin{align}
  \tau_\star^\T = \tau_\star^{-1}
  &&
  \text{and}
  &&
  H_\star^{\text{Triv}} \tau_\star^\T = H_\star.
\end{align}
The SWEL-preserving gate $\tau \simeq (\tau_\star, \tau_\star)$ therefore satisfies
\begin{align}
  C = \tau(\TrivSWEL(k, s)),
\end{align}
which completes our proof of \cref{thm:swel_construction}.

\section{SWEL equivalence theorem}
\label{sec:swel_equivalence}

Here we prove \cref{thm:swel_equivalence} of the main text, restated here for reference:
\begin{tcolorbox}
  \SWELEquivalence*
\end{tcolorbox}

To prove \cref{thm:swel_equivalence}, we reduce it to a statement about decomposing elements of $\O(n,\F_2)$ into transvections of the form
\begin{align}
  t(Q) = \1_n + e_Q e_Q^\T,
  \label{eq:transvection}
\end{align}
where $Q\subset\Z_n$, $e_Q = \sum_{q\in Q} e_q \in \F_2^n$ and $e_q\in\F_2^n$ with $q\in\Z_n$ is a standard basis vector for $\F_2^n$.
The transvection $t(Q)$ induces a SWEL-preserving gate $T(Q) \simeq (t(Q), t(Q))$ that acts nontrivially on $\abs{Q}$ qubits.
Our reduction makes use of the following lemma:
\begin{tcolorbox}
  \begin{lemma}
    If $C$ and $C'$ are both $\params{n,k}$ SWEL codes, then $C' = \tau(C)$ for some SWEL-preserving gate $\tau$.
    \label{lemma:swel_equivalence}
  \end{lemma}
  \begin{proof}
    Let $\tau_C$ and $\tau_{C'}$ be SWEL-preserving gates that construct $C$ and $C'$ from the trivial SWEL code, in the sense of \cref{thm:swel_construction}.
    Then $\tau = \tau_{C'} \tau_C^{-1}$ is a SWEL-preserving gate for which $C' = \tau(C)$.
  \end{proof}
\end{tcolorbox}
If $\tau_\star$ can be decomposed as $\tau_\star = t(Q_N) \cdots t(Q_2) t(Q_1)$ with $\abs{Q_j}\in\set{2,4}$ for all $j$, then $\tau$ decomposes into elementary SWEL-preserving and SWAP gates as $\tau = T(Q_N) \cdots T(Q_2) T(Q_1)$.
The SWAP gates can then be commuted to the end by relabeling gate targets, and merged into a single permutation gate, altogether arriving at a decomposition of the form $\tau = \Pi T(Q_{N'}') \cdots T(Q_2') T(Q_1')$ for some permutation gate $\Pi$, integer $N'\le N$, and $\abs{Q_j'}=4$ for all $j\in\Z_{N'}$.
\cref{lemma:swel_equivalence} thereby reduces \cref{thm:swel_equivalence} to the following $\O(n,\F_2)$ decomposition lemma:
\begin{tcolorbox}
  \begin{lemma}
    If $U \in \O(n,\F_2)$, then
    \begin{align*}
      U = t(Q_1) t(Q_2) \cdots t(Q_N)
    \end{align*}
    for some finite sequence of subsets $Q_j\in\Z_n$ with $\abs{Q_j}\in\set{2,4}$ for all $j\in\Z_N$.
    \label{lemma:ortho_decomp}
  \end{lemma}
\end{tcolorbox}
We prove this theorem by constructing a Gaussian-elimination-like algorithm that uses suitable transvections to reduce an arbitrary element $U\in\O(n,\F_2)$ to the identity matrix.
To this end, we expand
\begin{align}
  U =
  \begin{pmatrix}
    u_1 & u_2 & \cdots & u_n
  \end{pmatrix}
\end{align}
into columns $u_j\in\F_2^n$.
We will use the following fact about the columns of $U$:
\begin{tcolorbox}
  \begin{lemma}
    If $v$ is a column of some $U\in\O(n,\F_2)$ with $n>1$, then $\norm{v}$ is odd and $v\ne\bm{1}_n$.
    \label{lemma:reduction_lemma}
  \end{lemma}
  Here $\bm{1}_n\in\F_2^n$ is the vector of all ones.
  In other words, $v$ has an odd number of ones and at least one zero.
  \begin{proof}
    Let $u_1$ and $u_2$ be arbitrary but different columns of an arbitrary $U\in\O(n,\F_2)$.
    \begin{enumerate}
      \item $U^\T U = \1_n$ implies that $u_1^\T u_1 = 1$, so $\norm{u_1}$ must be odd.
      \item $\bm{1}_n^\T u_1 = u_1^\T u_1 = 1$, but $U^\T U = \1_n$ implies that $u_2^\T u_1 = 0$, so $u_2\ne\bm{1}_n$.
    \end{enumerate}
  \end{proof}
\end{tcolorbox}
Finally, we observe that
\begin{align}
  t(\Z_4)
  \begin{pmatrix}
    0 \\ 1 \\ 1 \\ 1
  \end{pmatrix}
  =
  \begin{pmatrix}
    1 \\ 0 \\ 0 \\ 0
  \end{pmatrix}.
  \label{eq:reduction_unit}
\end{align}
If $n \ge 4$, \cref{lemma:reduction_lemma} and \cref{eq:reduction_unit} let us use a zero and one in the last column $u_n$ of $U\in\O(n,\F_2)$ to catalyze the reduction of $u_n$ to a standard basis vector, removing a pair of ones with each transvection.
Each step in the reduction reduces the weight $\norm{u_n}$ by two until $\abs{u_n} = 1$.
We can then move the last one in $u_n$ to the $n$-th position with a SWAP gate.
This reduction is implemented by \cref{alg:reduce_column}.

\begin{algorithm*}[tb]
  \caption[]{%
    Reducing the last column of an orthogonal matrix $U_{\mathrm{in}}\in\O(n,\F_2)\subset\F_2^{n\times n}$.
    This algorithm comes with the following definitions and promises:
    \begin{enumerate}[label=(\alph*), parsep=0pt, itemsep=1pt]
      \item For any $M\in\F_2^{n\times n}$, $M[r,c]$ is the entry of $M$ at row $r$ and column $c$.
      \item For any $M\in\F_2^{n\times n}$, $M[:,n]\in\F_2^n$ is the last column of $M$.
      \item For any set $\mathcal{S}$, $\oper{pop}(\mathcal{S})$ removes and returns an arbitrary element of $\mathcal{S}$.
      \item For any $Q\subset\Z_n$, $t(Q) = \1_n + e_Q e_Q^\T$ and $e_Q = \sum_{q\in Q} e_q$, where $e_q\in\F_2^n$ with $q\in\Z_n$ is a standard basis vector for $\F_2^n$.
      \item The outputs $U_{\mathrm{out}}$ and $\mathcal{Q} = (Q_1, Q_2, \cdots, Q_N)$ satisfy
        \begin{enumerate}[label=(\roman*), itemsep=1pt, parsep=0pt]
          \item $U_{\mathrm{out}}[:,n] = e_n$,
          \item $\abs{Q}\in\set{2,4}$ for all $Q\in\mathcal{Q}$, and
          \item $U_{\mathrm{out}} = t(Q_N) \cdots t(Q_2) t(Q_1) U_{\mathrm{in}}$.
        \end{enumerate}
    \end{enumerate}
  }
  \label{alg:reduce_column}
  \DontPrintSemicolon
  \SetKwComment{Comment}{\# }{}
  \SetKwProg{Fn}{def}{\string:}{}
  \SetKwInOut{Input}{Input}
  \SetKwInOut{Output}{Output}
  \SetKwFor{For}{for}{:}{}
  \SetKwFor{While}{while}{:}{}
  \SetKwIF{If}{ElseIf}{Else}{if}{:}{elif}{else:}{}
  \BlankLine
  \Input{Matrix $U_{\mathrm{in}} \in \O(n,\F_2) \subset \F_2^{n\times n}$}
  \Output{%
    (1) Matrix $U_{\mathrm{out}} \in \O(n,\F_2) \subset \F_2^{n\times n}$ \\
    (2) Sequence $\mathcal{Q} = (Q_1, Q_2, \cdots)$ of subsets $Q_j\subset\Z_n$
  }
  \BlankLine
  \SetKwFunction{Reduce}{reduce\_last\_column}
  \Fn{\Reduce{$U_{\mathrm{in}}$}}{
    \texttt{set} $U_{\mathrm{out}} \leftarrow \oper{copy}(U_{\mathrm{in}})$\;
    \texttt{set} $\mathcal{Q} \leftarrow$ empty list\;
    \texttt{set} $q_1 \leftarrow$ position of the first zero in $U_{\mathrm{out}}[:,n]$\;
    \texttt{set} $\mathcal{I} \leftarrow \set{j\in\Z_n:U_{\mathrm{out}}[j,n] = 1}$
    \Comment{positions of all ones in the last column}
    \texttt{set} $q_2 \leftarrow \oper{pop}(\mathcal{I})$\;
    \While{$\abs{\mathcal{I}} > 0$}{
      \Comment{eliminate a pair of ones in the last column}
      \texttt{set} $q_3 \leftarrow \oper{pop}(\mathcal{I})$\;
      \texttt{set} $q_4 \leftarrow \oper{pop}(\mathcal{I})$\;
      \texttt{update} $U_{\mathrm{out}} \leftarrow t(\set{q_1, q_2, q_3, q_4}) U_{\mathrm{out}}$\;
      \texttt{append} $\set{q_1, q_2, q_3, q_4}$ \texttt{to} $\mathcal{Q}$\;
    }
    \Comment{move the remaining one in the last column to the last position}
    \If{$U_{\mathrm{out}}[q_1,n] = 1$}{
      \texttt{update} $U_{\mathrm{out}} \leftarrow t(\set{q_1, n}) U_{\mathrm{out}}$\;
      \texttt{append} $\set{q_1, n}$ \texttt{to} $\mathcal{Q}$\;
    }
    \Else{
      \texttt{update} $U_{\mathrm{out}} \leftarrow t(\set{q_2, n}) U_{\mathrm{out}}$\;
      \texttt{append} $\set{q_2, n}$ \texttt{to} $\mathcal{Q}$\;
    }
    \Return{$U_{\mathrm{out}},\mathcal{Q}$}
  }
\end{algorithm*}

After reducing the last column of $U$ with \cref{alg:reduce_column}, its last column $u_n = e_n$.
Moreover, the $n$-th entry of every column $j\ne n$ is zero, since now $e_n^\T u_j = u_n^\T u_j = 0$.
That is, the reduced matrix $U$ has the form
\begin{align}
  U =
  \begin{pmatrix}
    u_{1,1}   & u_{1,2}   & \cdots & u_{1,n-1}   & 0 \\
    u_{2,1}   & u_{2,2}   & \cdots & u_{2,n-1}   & 0 \\
    \vdots    & \vdots    & \ddots & \vdots      & 0 \\
    u_{n-1,1} & u_{n-1,2} & \cdots & u_{n-1,n-1} & 0 \\
    0         & 0         & \cdots & 0           & 1
  \end{pmatrix}.
\end{align}
At this point, we can consider two cases:
\begin{enumerate}
  \item If $n>4$, we can recursively consider the sub-matrix
    \begin{align}
      \tilde{U} =
      \begin{pmatrix}
        u_{1,1}   & u_{1,2}   & \cdots & u_{1,n-1}   \\
        u_{2,1}   & u_{2,2}   & \cdots & u_{2,n-1}   \\
        \vdots    & \vdots    & \ddots & \vdots      \\
        u_{n-1,1} & u_{n-1,2} & \cdots & u_{n-1,n-1} \\
      \end{pmatrix},
    \end{align}
    and again reduce the last column of $\tilde{U}$.
    Note that the transvections used to reduce the last column of $\tilde{U}$ do not, if extended to act on $\F_2^{n\times n}$, affect the last (already reduced) column of $U$, since $t(Q) e_n = e_n$ whenever $n\notin Q$.
  \item If $n=4$, then all elements of $\O(n,\F_2)$ are equal to either $t(\Z_4)$ or $\1_4$ modulo row and column permutations, so any reduced matrix $U$ for which $u_n = e_n$ must be equal to the identity matrix modulo permutations.
\end{enumerate}
Altogether, we can use \cref{alg:reduce_column} to recursively reduce any matrix $U\in\O(n,\F_2)\subset\F_2^{n\times n}$ to a permutation matrix using only transvections $t(Q)$ with weight $\abs{Q}\in\set{2,4}$, and then use weight-2 transvections to further reduce the permutation matrix to the identity matrix.
This procedure, implemented by \cref{alg:reduce_matrix}, constructively proves \cref{lemma:ortho_decomp}, which completes our proof of \cref{thm:swel_equivalence}.

\begin{algorithm*}[tb]
  \caption[]{%
    Decomposing the matrix $U\in\O(n,\F_2)\subset\F_2^{n\times n}$ into elementary SWEL-preserving gates.
    This algorithm comes with the following definitions and promises:
    \begin{enumerate}[label=(\alph*), nosep]
      \item For any $M\in\F_2^{m\times m}$, $\oper{size}(M) = m$.
      \item $\oper{reduce\_last\_column}$ is defined in \cref{alg:reduce_column}.
      \item For sequences $\mathcal{A}$ and $\mathcal{B}$, ``\texttt{extend} $\mathcal{A}$ \texttt{by} $\mathcal{B}$'' appends the elements of $\mathcal{B}$ to $\mathcal{A}$, preserving their order.
      \item For any $M\in\F_2^{m\times m}$, $\oper{submat}(M)$ is the sub-matrix of the first $m-1$ rows and columns of $M$.
      \item For any $Q\subset\Z_n$, $t(Q) = \1_n + e_Q e_Q^\T$ and $e_Q = \sum_{q\in Q} e_q$, where $e_q\in\F_2^n$ with $q\in\Z_n$ is a standard basis vector for $\F_2^n$.
      \item The output $\mathcal{Q} = (Q_1, Q_2, \cdots, Q_N)$ satisfies
        \begin{enumerate}[label=(\roman*), nosep]
          \item $\abs{Q}\in\set{2,4}$ for all $Q\in\mathcal{Q}$, and
          \item $U = t(Q_1) t(Q_2) \cdots t(Q_N)$.
        \end{enumerate}
    \end{enumerate}
  }
  \label{alg:reduce_matrix}
  \DontPrintSemicolon
  \SetKwComment{Comment}{\# }{}
  \SetKwProg{Fn}{def}{\string:}{}
  \SetKwInOut{Input}{Input}
  \SetKwInOut{Output}{Output}
  \SetKwFor{For}{for}{:}{}
  \SetKwFor{While}{while}{:}{}
  \SetKwIF{If}{ElseIf}{Else}{if}{:}{elif}{else:}{}
  \BlankLine
  \Input{Matrix $U \in \O(n,\F_2) \subset \F_2^{n\times n}$ with $n\ge4$}
  \Output{Sequence $\mathcal{Q} = (Q_1, Q_2, \cdots)$ of subsets $Q_j\subset\Z_n$}
  \BlankLine
  \SetKwFunction{Decompose}{decompose\_orthogonal\_matrix}
  \Fn{\Decompose{$U$}}{
    \texttt{set} $\mathcal{Q} \leftarrow$ empty list\;
    \texttt{set} $V \leftarrow \oper{copy}(U)$\;
    \While{$\oper{size}(V) > 4$}{
      \texttt{set} $\tilde{V}, \tilde{\mathcal{Q}} \leftarrow \oper{reduce\_last\_column}(V)$\;
      \texttt{update} $V \leftarrow \oper{submat}(\tilde{V})$\;
      \texttt{extend} $\mathcal{Q}$ \texttt{by} $\tilde{\mathcal{Q}}$\;
    }
    \Comment{identify SWAPs necessary to convert V to the identity matrix}
    \If{$V[3,3] = 0$}{
      \texttt{set} $q \leftarrow$ position of the first one in $V[:,3]$\;
      \texttt{update} $V \leftarrow t(\set{q,3}) V$\;
      \texttt{append} $\set{q, 3}$ \texttt{to} $\mathcal{Q}$\;
    }
    \If{$V[2,2] = 0$}{
      \texttt{append} $\set{1, 2}$ \texttt{to} $\mathcal{Q}$\;
    }
    \Return{$\mathcal{Q}$}
  }
\end{algorithm*}

\onecolumn
\clearpage
\section{Best codes found}
\label{sec:best_codes}

Here we provide weight-optimized stabilizer generators for the codes in \cref{tab:codes}.
These codes can be loaded into a useful format in Python using the \texttt{qldpc} package~\cite{perlin2023qldpc} on the Python Package Index, as follows:
\begin{tcolorbox}[boxsep=-3pt]
\begin{lstlisting}[language=python]
import numpy as np
import qldpc
from qldpc.objects import Pauli

# Construct a CSS code from stabilizers represented by strings.
pauli_strings = """
XXXX___
_XX_XX_
XX__X_X
ZZZZ___
_ZZ_ZZ_
ZZ__Z_Z
""".strip().splitlines()
code = qldpc.codes.CSSCode.from_strings(pauli_strings)
assert code.get_code_params() == (7, 1, 3)

# You can also load a code from qecdb.org by its ID.
code = qldpc.codes.CSSCode.from_qecdb_id("6705224a4315521360862353")

# SWEL codes can be defined from their X or Z stabilizers alone.
# Calling .to_swel() on a CSS code that is SWEL finds a self-dual logical operator basis.
strings_x = """
XXXX___
_XX_XX_
XX__X_X
""".strip().splitlines()
strings_z = [string.replace("X", "Z") for string in strings_x]
code = qldpc.codes.CSSCode.from_strings(strings_x + strings_z).to_swel()
logicals_x = code.get_logical_ops(Pauli.X)
logicals_z = code.get_logical_ops(Pauli.Z)
assert np.array_equal(logicals_x, logicals_z)
\end{lstlisting}
\end{tcolorbox}
The above Python code runs with \texttt{qldpc==0.5.0}.

\subsection{CSS codes}
\label{sec:best_codes_css}
\begin{lstlisting}
# CSS (20, 4, 4)
X______X_X_X_XX_X_X_
_X_____X__XX__XX_XX_
__X_______XXX__XXX_X
___X_____XX_XXX_X__X
____X____X_XX_XX__XX
_____X___XXXXX___XX_
______XX__X__X__XXXX
________X___XXXXXXX_
Z______Z_Z_Z_ZZ_Z_Z_
_Z_____Z__ZZ__ZZ_ZZ_
__Z_______ZZZ__ZZZ_Z
___Z_____ZZ_ZZZ_Z__Z
____Z____Z_ZZ_ZZ__ZZ
_____Z___ZZZZZ___ZZ_
______ZZ__Z__Z__ZZZZ
________Z___ZZZZZZZ_

# CSS (20, 6, 4)
X_X______X__XX___XXX
_X_______XXX__XXX_X_
___X____XXX__XX___XX
____X__X_X_XX__X_XX_
_____X____XXX_X__XXX
X_X___XX__X___XXX___
__X____XX__XXXX_X_X_
Z_Z______Z__ZZ___ZZZ
_Z___Z___Z__Z__ZZZ_Z
___Z_Z____Z_ZZZZZ___
____ZZ___Z___ZZ_ZZZ_
Z_Z___ZZ__Z___ZZZ___
________ZZZZ__ZZZZ__
__Z__Z_Z___ZZ_ZZ_ZZ_

# CSS (22, 4, 5)
__X__X__X___XX________
_X_X__X____XX____X___X
____XXX________X_X_X_X
X__XX_X______XX__X____
__X_X_X__X_X_____X__X_
XXXX__XX_______X______
______X__X__X_XXXX____
__X___XX__X_____XX_X__
_X_X___X_X___X____XX__
Z___Z____Z_____Z__Z___
__Z___Z__Z__Z____Z_Z_Z
___Z___Z________ZZZZZ_
___Z_Z_Z__Z__Z____Z__Z
__Z_ZZZ_______Z_ZZ____
__ZZ___ZZ_Z____Z_Z____
Z_Z___ZZ_____Z___Z__Z_
___Z___ZZZ_Z_Z__Z_____
ZZ_Z_Z_ZZ__________Z__

# CSS (22, 6, 4)
X_X_____XX____XX_X____
_X_X______X_XX__X_XX__
_XXX____X___X____X__XX
___X_X_____X_X__X__XXX
_XXXX____X_X___XX_____
X____X__XX__XX____X_X_
______XX__XXXX__X____X
X______X__XX___XXX_X__
__Z____ZZZ_____ZZ__Z__
Z_Z___Z___Z____Z_ZZ___
Z__Z_____Z___Z____ZZ_Z
ZZ_____Z__Z___Z_Z_Z__Z
_Z__Z_____Z_ZZZ__Z___Z
___ZZ__Z__ZZ____Z_Z_Z_
____Z____ZZ___ZZZZ__Z_
Z_Z__Z_Z___Z______ZZZ_

# CSS (24, 4, 5)
__X_______X___XX_X____X_
X_______X___XX_X______XX
X__X____X__X_____X_XX___
____X___X____XX_X__X___X
_____X___X_X_XX__XX_____
_______X_X_XX__XX____X__
XXX__XX______XX_________
X____XX_____X__X__XX____
_______X_XX___X___XX___X
X_____X_X_X____XX___X___
_Z________Z__Z___Z__Z__Z
_Z___Z_____Z___Z____Z_Z_
__Z__Z_Z___Z_____ZZ_____
Z_____ZZ________Z___Z__Z
Z_Z__________ZZ_Z__Z_Z__
__Z_____Z_ZZ__Z______ZZ_
____Z_Z______Z__Z_Z__Z_Z
Z__Z__Z____Z_____Z___ZZ_
_________ZZZ___Z___Z_Z_Z
____________Z__ZZZZZ_Z__

# CSS (24, 6, 4)
X_________XXX_XXX___X___
_X___XX_____X_X____XX__X
X_X__X_____X____XX____XX
_X_X__X______X__X_X___XX
X___XX________XX_XXX____
______X_X_XXX_X____X__X_
__XX___XX_____X____X__XX
_________X_XXX_X__X___XX
_______X__XX___X__XXXX_X
__Z___Z__Z_Z_Z__Z______Z
Z_____ZZZZ____Z______ZZ_
_Z_________Z__Z__ZZZ__ZZ
___Z____Z__Z_Z____ZZZ_Z_
____Z_______ZZ__ZZ___ZZZ
_____ZZ__ZZ________ZZZZ_
__ZZZ___ZZ__Z__Z_______Z
Z_ZZ_Z___Z_______ZZ_Z___
__Z_____Z___ZZZ_Z__ZZ___

# CSS (25, 5, 5)
__X__X________X___X_XX_X_
X_X____XX__XX______X_____
_XX______X___X___X_XX____
___X_X___XX___X_X______X_
XX_________X___X_X_X_X_X_
___X______X____XX__XX_X_X
XXX_______X_X__X__X_____X
_X__XX_____X___XX___X__X_
______X____XXXX__X___X_X_
__X__X_X__X____XXXX______
Z_______ZZ____Z__ZZ______
__Z___Z_________Z__Z__ZZZ
Z___ZZ____ZZ_Z______Z____
___Z__Z_Z____ZZ____Z_Z___
Z_Z_Z_________Z__Z___Z_Z_
_Z_ZZ__Z__Z________Z__Z__
_Z___ZZ___Z______ZZ_____Z
Z_________Z_______ZZZZ_ZZ
_Z_______Z___ZZZ_Z___ZZ__
___Z____Z_Z_Z__Z_ZZ_Z____

# CSS (26, 4, 6)
XX_______X___X__X_____X___
X__X_______X_X_X______X_XX
_XX_________XXX__X____X__X
__X__________XXX___X_X_XX_
__XX________X____XXXX__X__
____X_____X__X__XXX_____XX
__X__X________XXX_X___X__X
______X_____X_XXX__X__X_X_
_______X____XX_X___XXX___X
__X_____X__X__X_XX__X___X_
__________XXXXXXX______X__
___Z__Z_Z_Z____________ZZ_
ZZ_______ZZ__Z_______Z____
Z___________Z__ZZZ_Z_Z_Z__
_Z____Z____Z__ZZZ_Z_Z_____
__Z________ZZZ_Z__ZZ__Z___
___ZZ______Z_ZZ_Z__Z____Z_
Z___Z______Z_Z____Z_ZZ___Z
_Z___ZZ______Z__Z_____Z_ZZ
__Z___Z____ZZ____Z__ZZ___Z
_Z____ZZ_______Z___Z_ZZZ__
________Z__Z____Z_ZZ__Z_ZZ

# CSS (26, 6, 5)
X_______X__X_X_XX_X_______
_X_________XX_____XX_XXX__
___X_X____X_X____X__X_X_X_
__XXX________X_____XX_X__X
____XX_____X__X_XXX___X___
____XXX_____XX_X____X__X__
_______X___XX____XXXX____X
XX_X___XX__________XX___X_
________XX___XX_X___XX_X__
_X_X_X___X___X_X__X____X__
_Z__ZZ________ZZ__Z_Z_____
_______Z___Z_Z___Z_____ZZZ
Z_________Z__Z_ZZ__Z__Z__Z
_____Z___Z__ZZZZ_____Z___Z
__Z____Z______Z_ZZZ_Z__Z__
___Z__Z____Z__ZZ_____Z__ZZ
____Z__Z___Z___ZZ_ZZ_Z____
ZZ____Z___Z_Z_Z___Z__Z____
Z_______Z_ZZ__Z_ZZ___Z____
ZZ_____Z_Z___Z_____Z_Z_Z__

# CSS (30, 4, 6)
____X___X___X__X_X_X____X_____
__XX__X__X____X________X_X____
__________X____X_X____XXX_X___
___XXXX___________X__X______X_
_______X_XXX______X___X__X____
_____________X__X_XX___X__X_X_
__X___XXX____X________X__X____
XX_____X_X___X___X___________X
X_X________________X__XXX__XX_
_X_____________XXXX__X__X__X__
__X_____X_______X___X__XXX___X
X__X________X___X_X_X___X___X_
______X_____X__X____X___X_XX_X
____Z_Z__Z_______Z____Z____Z__
Z____Z__________ZZ______Z__ZZ_
__Z___Z_____________Z_ZZ_Z__Z_
_____Z____Z___Z___Z_Z__Z___Z__
Z______Z__Z_______Z__Z_Z_Z____
ZZ__Z_Z___Z_____________ZZ____
____________Z_Z__Z___Z_Z____ZZ
Z______ZZ____Z__ZZ____Z_______
_Z______Z_Z_____Z__ZZ_Z______Z
___Z______Z______ZZZ___Z__Z__Z
Z___Z______Z__ZZ__Z____Z_____Z
______Z_Z______ZZ___ZZ_Z___Z__
___ZZ____Z_Z____Z__Z_______Z_Z

# CSS (30, 6, 6)
X___________X___XX___X_XXX____
_X_____X______XXX_X__X_____X__
__X______X__X_X_X_______X__X_X
__X_X____X___X_X___XX_______X_
_____X_____XXX______XXX___X___
______X___X___X_X_X_X__XX_____
_______X____X__X__XX_____XX_X_
___X___X_X_X_____X______XX_X__
XX______XX_X_X__________X____X
______X____X_____XX__X_X__XX__
__X_______X___XX___X__X_XX____
___X__X___XX_X____XX_____X____
Z____Z________Z___Z___ZZZZ____
__Z_______Z___Z____ZZZ_Z____Z_
___ZZ_______Z__ZZ______Z_ZZ___
____Z___Z_____ZZ_Z____Z_Z_Z___
_____Z____Z__Z__ZZ__ZZ__Z_____
__Z____Z__Z_Z_____ZZ_____ZZ___
_________ZZ___Z____Z__Z___ZZ_Z
___Z__Z_ZZ_______Z__ZZ_____Z__
__Z____ZZZ____Z_Z____Z____Z___
__Z_______ZZZ_ZZ_Z___________Z
_Z______Z__Z___ZZ_Z_Z___Z_____
____Z___Z____Z______Z___ZZ__ZZ

# CSS (32, 6, 6)
__X____X_X___X_XX____________X__
_X_____X________X_X_X_______X_XX
___XX___________X_X__X_____X_XX_
X_X__X________X_X________X_XX___
______X______XXX_X__X_____X_X___
_______X___________XXXXXX____X__
_____X_XX______X___X_____X__X__X
_________XX_X_______X__X_XX___X_
XX_X_______X____XX_X_______X____
_X_______XX___XX______X____X_X__
_____X_____XX___X______XXX___X__
__XX_________X_________XXXX____X
X_____XX__X___XX________X____X__
Z____________Z_Z_Z__Z_Z___Z_Z___
__Z__________Z_______Z_Z_ZZ___ZZ
_______ZZ____________ZZ__ZZ_ZZ__
__ZZ____Z_________Z_____Z__Z_Z_Z
____ZZ__Z________Z____ZZ__ZZ____
Z_Z__ZZ_Z_______Z_________Z___Z_
ZZZZ___Z_______________Z_Z_Z____
_________Z___Z__ZZ______ZZ___Z_Z
_Z____Z___ZZ_____Z_Z___Z_______Z
___Z__Z___Z__Z___Z______Z_Z__Z__
_________Z__ZZ_____Z___Z___ZZ_Z_
__Z__Z__Z___ZZ____Z_ZZ__________
______________Z_Z_ZZ______Z__ZZZ

# CSS (32, 8, 6)
X_____X____XX___X____X_X___X____
_X_X______X_X__X___________X_XX_
X____X__X___X___X______X_X___X__
__XXX___X__X_X____X_X___________
_XXX_______X____XX_X__X_X___X___
___X____X_XX_____XX__XX__X__X___
____X______XX____XXX_X_X__X_X___
X____XX___X_______XXXX______X_X_
______XX___X____XXX____XXX__X___
_______X_X_____XX_X_____X___XXXX
_______XXX_X_______XX_X___X_XX__
______________X_X__X_XXXX____XXX
Z________Z__ZZ_Z__Z___Z_Z_______
_______Z___Z__ZZ__ZZ___Z_____Z__
__ZZZ____Z_Z__Z_________Z__Z____
ZZZZ___________Z_Z_____Z______Z_
_ZZ__Z__ZZZ___Z________________Z
_Z____________Z__Z___ZZZZZZ___Z_
__Z____Z_____ZZ__Z____Z__Z__ZZZ_
Z__Z_____Z___ZZ_Z____Z_Z_____Z_Z
____Z_Z_____Z_Z___Z___Z__Z__Z_ZZ
_____Z_______Z__Z__ZZ___ZZ_ZZZ__
______Z______Z__Z__ZZZZZZ_Z_____
_Z______Z___Z__ZZ_ZZ___ZZ__Z____

# CSS (36, 6, 6)
__X______X_________X______X__X__XX_X
____X__X___X___X_X________X___X____X
___XX_______X______X__XX_X___X______
X_X___X______X_________X____X_XX____
__X_____X_X_______XX__X_____X_____X_
_X______X______XX__X_X____X_______X_
X________________X____X__X__X__XXXX_
XX___X_______XX___X____X____X____X__
___X__________XX___X__X____XX_XX____
X__X_X___X________X_X_________XX__X_
_X__X________X_______XX____X__X_XX__
_______X______X_____XXX_XX________XX
___XX___XX___XX___X_X___X___________
X________X__________X____X_XXXX____X
_X_____X___X__X___X________X___XX_X_
_Z___Z__Z___________Z__Z_Z_ZZ_______
___Z_Z___________ZZ______ZZZ____Z_Z_
___ZZZ__________Z_____ZZ___Z__Z___Z_
__________Z___Z___Z_Z___ZZ_Z_Z__Z___
__Z___Z__ZZZ_ZZ_________Z_____Z_____
___________Z____Z__Z___Z__Z___Z__ZZZ
_____Z__Z__Z___Z_ZZ________Z__ZZ____
_Z___Z___Z_ZZZ___Z_Z________Z_______
ZZ___Z______ZZ_______Z__Z____Z__Z___
Z______Z___ZZ________Z___Z_ZZ_____Z_
_Z__Z__Z__ZZ__Z_Z__Z______Z_________
__Z___Z_____Z________ZZ___ZZ__ZZ____
____Z_______ZZ_____Z_ZZZ_Z_________Z
____ZZ________Z___Z__ZZ_Z_Z______Z__
_____ZZ_____Z_ZZ____Z_Z_______Z___Z_

# CSS (36, 8, 6)
__________X_______X___X_XX_X_X_X____
______X___X___XXXX_X_______________X
____X________XX___X__XX_____X_X_X___
_XXX___X___X_____XX______X_____X____
__XXX_______X__X_____X_____X_X_X____
_____________X___XX____X_X___X__X_XX
X_________X___XX_______XX______X_XX_
_XX___X__XX__________X_______X____XX
_________X__X_X_______X___X__X_XXX__
______________X___XX___XX_X_X___X__X
___XX__X_X_______XX_____________XX_X
___XXX____XX___X______X__________X_X
_XX____________X____X__X__X___XX_X_X
__X_____X____X__X_X___X_XX_______XX_
_Z_____ZZ_____Z___________Z_____ZZ_Z
Z___Z______________ZZZ_Z____Z_Z____Z
ZZ_Z_Z____Z____Z______Z_______Z__Z__
Z___Z_Z_Z________________ZZ___ZZ___Z
_________Z__Z_____ZZZ_______Z__Z_Z_Z
_Z______Z_Z____Z_Z___ZZ___Z________Z
Z_______Z_Z______Z____ZZ_____Z_ZZ___
_Z____Z____Z_________ZZZ___Z_____Z_Z
_Z_____ZZ_ZZ______ZZ___Z____Z_______
_ZZ____Z____ZZ___Z_Z______Z___Z_____
Z____Z__Z_Z_ZZ_____Z________ZZ______
___ZZ____Z_____Z_________Z__ZZ____ZZ
___Z______Z__Z_____Z__Z_Z__Z_Z_Z_Z__
____Z_Z_Z_______Z___ZZ_____Z_Z___ZZ_

# CSS (38, 8, 6)
________X_______X__XX__XX____X_X_____X
____X___X____X_______XX__X_______X_XX_
__________XXXXX__X___X________X____X__
X_X_X__________X_XXX__X____X________X_
______X_________X_____XX_X__X_X_XXX___
_X__X_X________X___X___XX_X__X___X____
_______X_X_____XX_X______X_X__X___X_X_
_X______XX_______X_X_____XX___X_X____X
X____X__X_____X_X____X____XX__X______X
____X_X__X_X__X__X_______X_X_____X_X__
____________X__XXXX__X_XXX__________X_
_______X_X___X___X_______X___X_XXX_X__
_X_X__X______X_____XX_____XXX_X_______
_XX_X_X__X___X__X___X__X____________X_
X___X_X_X_____X____X__X_X______XX_____
__Z_Z_________ZZ_Z__________Z_Z____Z__
ZZ_Z__Z__Z_____Z___________Z_____Z__Z_
_ZZ_____Z______Z___Z__Z_Z_________Z__Z
__Z______Z_Z_____Z_ZZ___________ZZ__Z_
_ZZ______Z_Z_Z________Z______Z_Z__Z___
Z_____Z___________ZZ________Z_ZZ__ZZZ_
_Z____Z___________Z____ZZ____Z_ZZZ__Z_
___Z___Z________Z__Z_Z_ZZ__Z__Z__Z____
_____ZZZ_Z_______Z______Z_ZZ_ZZ_______
__ZZZ__ZZ______Z_________Z_Z_Z___Z____
__Z___Z___Z__ZZ__________Z___Z_____ZZZ
__ZZ______Z___Z_Z_ZZ__Z_________ZZ____
__Z__Z_____Z_____Z____ZZZZ_ZZ_________
___Z____Z___Z__Z_ZZZ______Z_Z____Z____
_ZZ_Z_____________Z__ZZ__Z_Z_______ZZ_

# CSS (38, 10, 6)
___X_X_X__XX____________XX_______X__X_
____X_____X_X___X____X______XXX__X____
_XX__X____XX____XX_______XX___________
_____X___XX_XX_____X__X____X__X_______
__X___X_______XX__X___X_________XXX_X_
____X__XX_____XX___X__X____X_____XX___
_____X_X___X_____XXX_X_X__X__X________
______X______X__X_XXX_XX_____X____X___
________X______XX_X___XXX_XX_______X__
_______XXX_____X_XX_______X_____X___XX
XX___X_X_______X__X_X_X_____X____X____
____X__X_X_X__X___X_X__________X__X_X_
XX_X_____X__X_X_X___X____________X__X_
_X____X____X___X___XX____X________XXX_
_Z_Z_____________Z_Z__Z_Z_______Z_____
_____Z____Z_Z_Z__ZZZ____ZZ____________
____Z___Z_Z___Z_Z_____Z_Z___Z________Z
ZZ_Z_Z_Z_Z____Z_____________________ZZ
___Z___Z_______Z_____ZZ_Z__Z_____ZZ___
__Z___Z__Z_____ZZ____ZZZ_______Z______
Z___Z____________Z_____Z_Z__Z__Z__ZZZ_
_____Z______Z_____Z_Z_Z_Z_ZZ_Z_____Z__
__Z_Z____ZZ________Z_ZZ_________ZZ_Z__
_Z________Z__Z____ZZ_Z_Z___Z_______ZZ_
Z_______Z_ZZ_____Z____Z_Z_Z__Z______Z_
___Z_Z_____Z______Z____Z___ZZ____Z_Z_Z
____ZZ_ZZ_Z____________ZZ__Z_ZZ_______
__ZZ_____Z__Z_Z_ZZ________Z_______Z_Z_

# CSS (38, 12, 6)
X____X___________XX_X_X__X_XX_X_______
___X_____XX_______X_X_____XXX_X___X___
X_X_X_____X___X____X____X__X__X___X___
XX___XXX____________X________X_XX__X__
X___X__X____X___X__X__X__X_____XX_____
X______X_____X_XX___X_____X____X__X_X_
_X_____X______X___X_X_XXX____XX____XX_
__X_____________XX__XX______X_XXXX_XX_
____X__________X__X__XX____XX__XX_XX_X
_______X__X___X_XXX_XX_X__XX_________X
_XX__X_XX_X___X______X____X_X__X___X__
_________XX___XXX_X_______X_XX_X___X_X
X____X_____X___X___X____X___XXXXX_X___
__Z_ZZ________Z___Z_ZZZ_Z___________Z_
Z_Z___Z_______ZZ___Z______Z_Z____Z__Z_
______ZZ_____ZZ__Z__Z____Z___ZZ_____Z_
_________Z__Z___Z___ZZ________Z___ZZZZ
_________ZZ_Z_ZZ__Z____Z____ZZ_Z______
Z____Z_Z__Z_____Z_Z______Z_______Z_ZZ_
___________ZZZZ______Z__Z____Z__ZZ__Z_
_Z_Z_ZZ_____Z_Z______Z_____Z____Z____Z
_Z__ZZ______ZZ_Z_Z___Z______Z_Z_______
Z________ZZ____Z___Z______Z___ZZZ__Z__
___Z___ZZZ____Z__Z_Z_____Z_____Z_____Z
Z___Z______________ZZ___Z__ZZ___ZZZZZ_
Z___Z__Z_____ZZZZ______ZZ_Z_Z________Z

# CSS (40, 10, 6)
________X__XXX________X___X_______XX__X_
_X____________X_X__X_X_____XX__X_X___X__
______XX________XX__X_______XX__X_X__X__
___XX____X____________XX_X____XXX____X__
_____________XX_X__XXX_______X_XX__X____
_______________XX____XX__X_X_X_X______XX
X__X___XX_X____X_____X_______XX____X____
X_XX____X________X_____XX_______XXX_____
X____X__X__X____XXX___XX_______________X
X__X______X________X____XX_XXX_____X___X
___XX_____X__X_X__X__X_X__X______X____X_
____X____X______XX______X_XXX_____X_XX__
________XXX____X_XX__XX_X_____X_____X___
_XX__X_______X_X__X_______XXX___X_X_____
__XXX_XX_________X________XX__XX_______X
____Z_ZZ____Z____________Z___Z_Z_____ZZ_
Z_______Z_______Z__Z_Z_ZZ__________ZZZ__
_Z_____________ZZ_Z_Z_______Z___Z_ZZ_Z__
__Z__Z_______ZZ_______Z__ZZ_Z_____Z_Z___
ZZZ_____ZZZ____Z____Z__________Z__Z_____
_____Z_Z____ZZ_Z_____Z_ZZ___Z_Z_________
Z________Z___Z_Z__________Z__Z__Z__Z__ZZ
_Z__ZZ_Z_Z___Z____Z_______ZZ_Z__________
_Z____Z_______Z_______Z__Z_Z__Z____Z_Z_Z
__Z___________Z__Z_ZZ_Z_Z_Z____Z_Z__Z___
____Z___Z_______ZZ__Z__Z___Z_Z_Z__Z__Z__
_______Z____Z_________Z_Z_ZZ_Z__ZZZ____Z
________Z_________Z_ZZZ______ZZ_Z___ZZ_Z
_Z_ZZZ_________Z_Z___Z__Z_____Z_Z___Z___
Z_Z___Z_Z__Z_Z________Z____Z___Z__Z___Z_

# CSS (40, 12, 6)
_X____________X_X__X______X_XXX___X_X___
____X__X________XXX___X_X___X_X_____X___
X_X_X_XXX_________X_____X________XX_____
____X_____XXX________X_X____XX________XX
__XX______X_X_____XX__X______X___X__X___
_XX__X__X__X________X____X____X__X____X_
XX__X__X__X________X_X__________XX___X__
X_____X______X_X__X_X_X_____X____X____X_
___XX_X_______XX___X_X__X______X_X______
X____X_______X___X__X______X__X_XXX___XX
__X__________X__X_XXX____XX__X__X____X_X
___X__X__________X_X___X_XXXX____X__X__X
_____X___X_X_____X__XX__X_X_X______X_X_X
______X_X__X_X___X_______XX__X_XX__X__X_
___Z_______Z______ZZ______Z_______ZZZZZ_
____________Z__Z___Z_Z_______ZZZ___ZZ_Z_
ZZ______Z_Z____Z__Z____Z_Z_______Z__Z___
__Z_ZZ____Z________Z____Z_Z______Z____ZZ
__Z___Z__ZZ______Z___Z___ZZ_Z__________Z
____Z____ZZZ______ZZZ_Z_____Z___Z_______
_Z___ZZ___Z______Z____Z_Z_Z__Z______Z___
___Z___Z__Z____Z___Z_ZZZ__________Z____Z
Z__Z___Z_ZZ____Z_Z__Z________________ZZ_
Z___Z___Z______ZZ____Z__Z_____Z____Z_Z__
__Z___Z______ZZZ__Z____ZZ___ZZ______Z_Z_
____Z_Z________Z__Z____Z_Z_ZZ_Z__Z_Z__Z_
______Z______Z__ZZ___ZZZ____ZZ__Z_ZZ____
________Z___Z_Z_Z_____Z__Z____Z_ZZ_ZZ__Z

# CSS (42, 8, 6)
__XX__X__X_____X____X__________XX_X_______
________X_X___X__X_______X___X__X_X__X____
__X_____X__________X_X__X___X__X_____X_XX_
___X___X_________X_______X________XXXX__XX
______X_______________X___X_X__XX__XX___XX
_____X___X__X_____X_XX___X_X_______X__X___
____________X_X__X__X__X__X______XX___X__X
_XXX____X___X_X___XX_____________X_____X__
XX___X___X_X__X_____X__X____X___X_________
__XX_______X_____X____X__X_X_XX__X________
____X__XX__X_____XX_______XX__X_X_________
_X__X_______XX______XX_X_X_______X______X_
__X____X_____X_X__XX_X__X_________X______X
_X______XXX____XXX________X________X__X___
__X__________XXX_____X____X_X_XXX_________
_____X_____X___X__X_X____X_____X_X____X__X
X____X_XX______X_X____X________X_____X___X
___Z_Z__Z_______Z__________ZZ__Z_____Z____
__________ZZ_____________Z__Z_Z__Z__Z_ZZ__
Z___Z___ZZ___Z___________Z_____Z_Z_______Z
Z_Z__Z________Z_____ZZ_____Z_ZZ___________
Z_______________Z__Z___ZZZ_______Z_Z_Z__Z_
____ZZ__________Z___Z_Z__Z__Z___Z_____Z_Z_
_Z____Z__Z_Z_______ZZ_________ZZ____Z____Z
_______Z_Z_____Z____Z___Z____ZZ___Z_Z___Z_
_Z______ZZ____Z__Z_Z____Z__Z____Z_______Z_
_____Z__ZZZZ_Z__Z___ZZ___________Z________
__Z_____Z_Z_Z_Z______Z_______Z__Z___ZZ____
_____Z___Z_____Z_Z________Z__Z_____ZZ_Z__Z
____Z_Z__ZZ_____Z___Z_____Z____Z____ZZ____
___Z_ZZ_Z___Z_____Z____ZZ____Z_____Z______
__Z_Z______Z_Z_____Z________Z__Z___Z__Z__Z
____ZZ________Z____Z_Z_Z_______ZZ___Z___Z_
__Z_Z____Z____Z__Z___Z_________Z__Z__Z___Z

# CSS (42, 10, 6)
____X________X_____XX__X_____X___X____X__X
_X__XXX_X____X____XX____X__________X______
_________X_______XX_X__X______X_X______XXX
____XX_XX___X_______X_X___X___X________X__
__X__X_X_______X__X_______XX_X__X____X____
_X________X_X__XX_X__X______X_X_____X_____
X___X________X____X__X_XX__X____X__X__X___
__X________X__X_________X_XXX_____X__XX_X_
_____X______X______X___XX__XX__X__XX____X_
______X__________X_X_X____X_XXX_____X_X_X_
______XX______X______X_____X_X_X_X_X_X__X_
_____X__X___X_X__XX__XX_X__X_____X________
_X___X____X_______X_______XX_______X__XXXX
X__X____X_X__________X_X_X__XXX___X_______
___X__X_X__X__X__________XXX_X____X______X
X__X_X________________X_X_X___X__X_X__X_X_
_________________Z_______ZZZ_Z_Z_____ZZZ__
______Z___Z_______Z_ZZZ____Z___Z____Z_Z___
Z________Z_Z______Z____Z__Z______Z_ZZ__Z__
_Z_ZZ_Z___Z__Z___________Z________Z____ZZ_
__Z_____Z__Z__Z________ZZ_Z____Z_Z______Z_
__Z_____ZZZ_Z__Z_____Z__Z_________Z_____Z_
_ZZ______Z__Z_________Z__Z_____ZZ_ZZ______
____Z__Z__Z___ZZ___Z_____Z_____Z_Z____Z___
_Z_____Z_____Z_Z__ZZ___Z______Z_ZZ________
Z___Z________Z__Z_Z_Z___Z_______Z_Z______Z
ZZ_Z__ZZ___Z__Z___ZZZ__________Z__________
___Z_ZZ_______ZZZ______Z_Z____Z_____Z_Z___
ZZ_____Z_Z___________ZZ_Z_ZZ_________Z_Z__
_Z_Z____ZZ_____Z_____Z_Z______Z___Z__Z___Z
Z__Z_Z________Z_Z___________ZZ_Z___Z___Z_Z
_Z________ZZ______Z___Z__Z___ZZ_Z____Z___Z

# CSS (42, 12, 6)
X__________________X_X_X_X_X___X___X_X__XX
______X_X______X__X_X_________XXXX___X___X
_________X_X______XX___X____X_X_X_XX_X____
__X___X_XX_X____________XX____X___X____X_X
__XX________X___XXXX____X_______X__XX_____
____X_XX____X__X___XX____X________X__XX___
X____________X____X_X____X___X__X__X__XX_X
__X___X__X_X_X____X____X__XX___X_____X____
___XX___X___XX__X__X_X_X_XX_______________
X____X____X___X_________XXX__X_____XX___X_
__X___X_X_X___XX____XXX_XX________________
__XX__XX____X_X______X________XX___X____X_
______X__X__X_XX________X______XX_X_X____X
___XX___X___XXX__X__X_X_____________X_X___
_X__X__X___X_______X__XX__X__X_____X_XX___
____Z_Z___Z_____Z__________Z_Z___Z__Z__ZZ_
Z_______Z______ZZZ__Z___Z__Z___Z____Z___Z_
__Z________Z________Z__________Z_ZZZ_ZZZZ_
____Z_Z_______Z________Z____ZZ__ZZ_Z____ZZ
______Z________ZZ___Z____Z______Z____ZZZZZ
___Z_ZZ_Z_____Z______Z_Z_____Z____Z_Z__Z__
ZZZ____Z___Z___ZZ___Z____Z________Z___Z___
_ZZ__Z_Z__Z_Z________ZZ__________ZZ__Z____
_____Z___ZZ______Z___ZZ__Z__Z___Z____Z__Z_
_Z_______Z_ZZ____Z__ZZ_Z__Z_____Z___Z_____
Z_Z___Z___Z_Z____Z_Z_____ZZZ__Z___________
_Z__Z____Z___ZZ_______Z____Z_Z_Z_ZZ_______
Z_Z_____Z____Z__Z__Z_________ZZ_ZZZ_______
__________________Z____ZZZ_Z__ZZZ_ZZ____Z_
____Z__ZZZZ___Z___Z_Z___________Z_Z____Z__

# CSS (44, 8, 7)
_X___X________X_XX_____X_______X_X_______X__
_______XX____X__XX___________X___X__X_X_____
X_________X___________X___X____X____XXX_X_X_
_XX_____X___________XX______X__X_____X_X_X__
_______X___________________XX_XXX____X_XX_X_
_____________X____X__X_X_XX_XX___X_________X
___X_XX___X____X_X_________X_________XX___X_
_XXXXX_X__XX________X_____________X_________
_X_X____________X______X_X__X_________XXX__X
_X_X_X_____XX______X_____XX_____________X_X_
_XX___X__X_XX____XX_________X______________X
_XX______X_X____X___X________XX_X____X______
____X_______XXX___X___XX______X______X__X___
X__X____X____X______X___X___X_X_______X___X_
X____________X_XX___X_____X_XX____X__X______
X___X_____X________XX_X____XXX________X____X
XX__X_____X____X_____________X___X_X_XXX____
__XXX___X___________X_X___XX____XX______X___
__Z_____Z______Z_Z___Z_Z______Z___Z____Z____
Z_______Z_________Z___Z_______Z____Z_ZZ____Z
_ZZZ_Z__Z___Z_____Z______________Z_______Z__
__Z_________Z__Z__ZZZ_______________Z_Z____Z
____Z___Z___________Z__Z_ZZ____ZZZ_______Z__
_______Z_Z__Z__Z__________Z______ZZZ_Z___Z__
_Z___Z___Z_____Z_Z______Z__ZZ____Z_Z________
ZZ_____Z_Z_Z_Z______Z_Z_____Z____________Z__
Z__Z______Z_Z_____________ZZ__Z_Z_________ZZ
ZZ_____________ZZ_Z___Z________Z_ZZ_______Z_
________ZZ____ZZ__Z__Z____Z_____ZZ________Z_
Z____ZZZ_Z______Z_______Z___________Z_Z___Z_
_________Z_____ZZZ____Z___Z__Z__Z___Z___Z___
____Z____Z___Z_ZZ___Z______Z____Z________Z_Z
_Z_Z______ZZ__ZZ___Z____Z_______Z____Z______
_______ZZ__Z___Z__ZZ__Z______Z________ZZ____
_Z__Z______Z____Z__ZZZ_______Z_______Z____Z_
___Z_Z_________Z____Z__Z__ZZ____Z_Z____ZZ___

# CSS (48, 6, 8)
_____XX___________X__________X_X_XX__________X__
_X____________X_______X_X_X__X_XX_________X_____
__XX_______________X__________XX_X_____X_X_____X
___X_______________X____XX_X____X___XX_________X
______X____X____X___________XX__XX_X_____X______
__X_X__X___X__X____________X____X______X____X___
__X_X__X________X___XX__X____________________X_X
X_X__X_X__X________________X___X______X________X
_____X________X_____X___X____X_X____________X_XX
X___________X____________XX___X_X_XXX______X____
______X___________________XX__X____X_X___XX_X_X_
__XX_X_X_________________________XX_____X_XXX___
___X___XX__________X__X__X__X_____XX__________X_
_____________X_X_____X_X_XX____X_______XXX______
_X____X______X_X______X_______X_X_____XX_____X__
X________X_____X____X_X__XXX_______X_X__________
____X___X________X________XX_X_X___X_________XX_
__X___________XXXX_______XX__X_X_______X________
_____X_____X_X____XX_______X___________XX_____XX
___________X____X__X_X______X_XX__X________X__X_
_____X_X_X_X____X_______X___X____X___X_X________
_____Z____Z________________Z___ZZ__Z___Z____Z_Z_
_____Z___Z__Z_Z___________ZZ_____Z___Z___Z______
__Z__Z__Z__Z_Z___Z_______________________Z___ZZ_
_____Z__________Z____Z___Z____Z____Z_______Z_Z_Z
Z____Z_________Z__Z___Z______Z___Z_Z_____Z______
_Z___________________Z___Z_ZZZ__ZZ________Z____Z
__Z______Z___________Z__Z____Z____ZZ____Z___Z__Z
Z_____Z_________Z___Z____ZZ_ZZ______Z_Z_________
________Z_____________ZZ_Z___Z__Z_Z_______ZZZ___
___________Z_____Z______ZZZ______Z_____ZZ____ZZ_
_____________Z_______Z_Z__Z__Z___Z___ZZ____Z___Z
______Z_______Z______Z__Z___ZZ__Z_______Z__Z__Z_
Z______Z_Z__Z________Z___Z____Z______Z_ZZ_______
ZZ_ZZ__Z________________Z_Z__Z___________Z___Z__
___ZZ________Z_ZZ______Z_____________Z___Z__Z_Z_
__Z_________ZZ__Z______ZZ_______Z___Z______Z___Z
ZZ_Z__________Z____________Z___Z__Z___Z_Z_Z_____
_Z_Z_____Z__Z___Z_________Z_______Z_Z____Z___Z__
______Z_Z__Z_Z_Z___Z______Z_____ZZ______Z_______
Z__Z___Z___________ZZ____________ZZ_____ZZ__Z___
_______________________ZZ__ZZZ_ZZ__Z__Z________Z

# CSS (48, 8, 8)
__X____________________________XX__X_X_X_X__XX_X
___X_______________X____X__XX______XX___X_X____X
_______X__X___________X_______X_X_____X_X_X__XX_
___X____X______________XX___XX______XX______X_X_
________XX__________________X_XX_X______X__X_X_X
____________X______XX___X_X_____X_____XX_X_____X
__X_X__X_____X__X__X____________X_X______X__X___
__XXX_____XX__X________X_X___X_________X________
_______X___X___XX_________XXX____X__X_____X_____
___X_____X___X__X______X____X______X______XXX___
_X_____X___X_X__X__________X_X___X_____XX_______
_XXX__________X_X__X____X_____________X__XX_____
______XX_X___X_X________X______XX____X________X_
___X____X_____X_XX_________X___X_X__X__X________
_XX_X_X_____X____X_________X_________XXX________
__X___X____XX_____X______________________XXXX__X
X__X___X__X_______X_____X____X_X______XX________
__X_X_____X____X__XX___X________________XX_____X
________X__________X_X___X_______X_XX_X__X____X_
_____X_X______________X_X__X__X______X_X__XXXX__
Z_Z_Z_Z______________Z__________Z_Z___Z_________
Z____Z________________Z_ZZZZ_________ZZZ________
ZZ___Z__________________Z______Z_ZZ___Z___Z_Z___
__Z________Z__________Z__Z_Z__Z___Z___ZZZ_______
___Z_Z______________Z__Z_Z___Z_Z_____ZZ_Z_______
________Z____________________ZZZ____Z_ZZ__ZZ__Z_
Z____Z____Z_________Z____Z__Z____________ZZZZ___
____________Z___Z_____ZZ_____Z_Z_____Z____Z_Z__Z
_____________ZZ___Z___Z______Z___ZZ_____Z_Z___Z_
_Z______Z_____Z________Z_Z___ZZ_______Z__Z__Z___
_Z_______ZZZ_____Z____Z_Z__________ZZ__________Z
_______________ZZZ___Z_Z_____ZZ_Z_____Z______Z__
______Z___ZZ___ZZ_______________Z__ZZ_Z_______Z_
_Z__Z__Z____Z_____________ZZ_Z__Z___Z____Z______
____Z____Z_Z________Z________Z______Z__ZZ___Z_Z_
_______ZZZ_Z_________Z_______Z___Z_______ZZ____Z
_____ZZ_______Z___Z__Z__________Z______ZZZ_____Z
Z__Z__ZZ________Z_Z____Z___________Z_Z________Z_
_Z_Z_____ZZZ___Z_ZZ__________Z__________Z_______
Z___Z______________ZZ__Z______Z___Z___Z_____Z__Z

# CSS (50, 6, 8)
X__________X_XX____X_____XXX______X_______________
______X________XXX____________X_____X__X____XX____
______X___X_______X____X____X_X____XX_____X_______
___X__X_____X__X_X___X______XX_X__________________
______X___XXXX____________X__X_____X__X____X______
________XX__X____X______X_X___X________X_X_____X__
__X__X______________X_______X__X_____X__XX_____XX_
__X_____X______X___XX_X______________XXX________X_
_X______X________XX_XX__________X____XX_______X___
_____________________X_XX__X___X_X_____X__X__X___X
____XX__X__X___________X______X____X_XX_________X_
____X___X_X_____X__X_____X_X______X___________X__X
________X_X___X__X_________X_______X__XX________XX
___X_X_X__X______XX___________X______XX__X________
_X____X_X________X______X______X____X__X__X___X___
___X___X______XX_____X__X___________X______XX_X___
_X_X_X_X____X____________X_________XX_XX__________
____X______XXX____________X___X__XX_____X_X_______
____XX____X_XX____X_____________X_X______X___X____
XX__X_X______X_______X___X______X_________XX______
__X_X_XX__X_____X___X____X____X____X______________
___X______X_______XXX___X_X_____XX__X_____________
___Z_Z__Z___________________Z_______Z__Z____ZZZ___
___Z_______ZZ___Z___Z_____________Z__Z__Z___Z__Z__
_Z__ZZ_________ZZ______Z_____Z_______Z____ZZ______
___Z__ZZ_____Z____Z_______________Z______Z__Z_ZZ__
Z____Z______Z______Z___________Z_Z___Z_____Z__ZZ__
_______Z_______Z____Z__Z_ZZ___ZZ___________Z__Z___
__ZZ_ZZ_______Z___________Z_______________Z__Z_ZZ_
Z__Z___Z____Z_Z__________ZZ_________________ZZ___Z
Z_Z___Z___Z_Z___Z_________Z____________Z_Z___Z____
____Z____ZZ________Z_____Z___ZZZ_______Z_______Z__
Z_Z_____Z_Z_______ZZ____Z__________________Z____ZZ
______Z__Z_ZZ__________ZZZZ_______Z_________Z_____
_______________ZZ___Z______Z_Z__Z_Z_____Z__Z_____Z
___Z___Z___Z_____ZZ___Z___________ZZ___Z_________Z
_______ZZ_____ZZZ__Z____Z____Z________Z__________Z
_Z______Z_______Z__Z__ZZ_Z__________Z__Z____Z_____
________Z_ZZZ______Z__Z________Z___Z_Z___________Z
_Z________Z____________ZZ____ZZZZ___Z___Z________Z
_______Z___Z_____________Z______Z___Z__Z_ZZZ__Z_Z_
_________Z___________Z____Z____Z_ZZ_____Z__ZZZZ___
__________Z______Z_______ZZ_ZZ_ZZ___Z_____ZZ______
Z___________ZZZ_________Z_Z____Z_____ZZ_Z______Z__

# CSS (50, 8, 8)
X__X_________________X_____XXX_X___________XXX____
___X_X_______________X_______X___X____X____XX_X__X
_______X________X______X_X___X_______XXX____XX____
_X__XX_______X__________X__X_X____XX_______X______
___X__X_X_______X_____X___X__X___X________X______X
_____X_X__X_X__________X______X______X__X____X___X
XX__X____X______X__________________X___XX______X_X
____X_____X_X__XX_________XX_X_______________X__X_
____X__X__X__X________X__XX__X_____X______X_______
__X__X____X__________X_______X___X_X____X_X___X___
XX_____X____XX_________XX_X__X________________X___
X__X______X_X_____X________X________XX________X__X
____X____X_____X__X__X__XXX____X_X________________
_______XX__X______X__X_X_X____X_X________________X
__________X______XX_____X____X___X_X_XXX__________
X_________XX_X_____X_X_X_X___________X____X_______
______________X__X_X_____X_____X___X______X__X_X_X
_________X__X_____XX__X___X________XX_XX__________
____X_XX___X________XX_X____X__________XX_________
_X_____X_____________X__X____XX____XX__X_X__X____X
__X___X______________X_____XX_X_____XX__X____XXX__
_____Z______________________ZZ_______ZZZ__Z___ZZZ_
______Z______________Z___Z_Z____Z__Z__Z_____ZZ___Z
__________Z_Z__________Z_Z____Z__Z___Z_Z_Z_______Z
__Z__Z______Z____Z______Z______Z__Z_Z______Z____Z_
_____________Z_____Z__Z________Z_ZZ_Z_Z_____Z_Z___
__Z___ZZ________________Z__Z___ZZ____Z____Z_____Z_
_________________Z_______Z____Z__ZZ______ZZZ_Z__Z_
Z___Z____Z____Z_______ZZ__________Z_________Z___ZZ
_Z__Z_____ZZ____Z_________ZZ____Z__Z_________Z____
___Z______Z_______Z____Z__ZZ____________Z__Z_Z_Z__
___Z____Z_Z__________Z_____Z___Z___ZZ___ZZ________
_______ZZ_____Z______ZZ___Z___Z___________Z_Z___Z_
_Z___Z__Z_______Z_______Z__Z_____Z__________Z__Z_Z
_Z________Z________Z_______________Z__ZZZZ____Z_Z_
_____Z____Z___Z____Z____Z_Z___Z__________Z_____Z_Z
_________Z___Z__ZZ_____ZZ_Z_Z_ZZ__________________
____Z______ZZ____Z__Z____Z________Z___Z_Z_____Z___
_____Z______ZZ_____ZZZ___Z_Z_______________ZZ_____
_______Z____Z_______Z_Z_____Z_________ZZ_Z___Z___Z
____Z_Z___Z__Z_Z__Z_Z_____Z_____Z_______Z_________
____Z_______Z________ZZ_Z_ZZ_Z_______ZZ______Z___Z

# CSS (50, 10, 7)
___X___X_______X_____X_X_______X______X__X____X__X
_XX___X_________XX_____X_______X____________XX___X
X______X_____XX______X__X__X__X_________________XX
XXXXX___X________X________________X_____XX________
__X_____X_____________XXX___X_X___X________X___XX_
_____X___X_______________X__XX______X_XX__XX__X___
__________X______X_X_______XX_____X__X_XX___X___X_
___________X________XXX_____X__X_X_X________XX_X__
____X_______X__X_________X__X____XX_XXX__X________
_X_X_______X__X__________XX___XX_________X____X_X_
______X_______XX_X_X________X____X_X____X__X____X_
___X_____XX_____X__X______X__XX____X__X_________X_
_XX_X_X______X______X_____X___XX_______X_______X__
_________X__XX______X__XX__XX__X___X____X_________
____X_______X_X_X___________X_X_X__________X__XX_X
_XX__XX_XX_________X___X_XX__________________X____
_XX_XX________X_____X___X__XX_____X____X__________
XXX___X______X__X_____X_____X__________X_X____X___
______X____________X_X___X_X___X____X__XX_X__XX___
___X__X________X_XX___X____X_____X__X___X_____X__X
__ZZZ___Z__________Z_____________Z__ZZ_______ZZ___
__Z_Z__Z__Z____Z___________ZZZ_______Z_______Z____
___ZZZZ_______Z_________Z__ZZ______Z_____________Z
____________ZZ_Z___________ZZZ__Z_____Z___Z____Z__
_________Z_Z___Z__Z_Z___ZZ_____Z____________Z___Z_
___Z__Z__ZZ_Z____Z_______ZZ_______________Z___Z___
Z___________________ZZ__Z___ZZZ______Z_ZZ_____ZZ__
_Z_______________Z_______ZZZ__ZZZ___Z______Z_ZZ___
___Z____________Z_______Z__Z_Z__ZZZ____Z_______ZZZ
____Z__Z____________Z____Z_Z_Z_Z_ZZZ_________Z_Z__
_____Z_Z_________________ZZ_ZZ_Z___________ZZZZ__Z
Z_____Z__Z___Z__________Z_Z_Z__Z_Z______Z____Z___Z
_______Z___Z___________ZZ___ZZ___Z____Z_ZZ_Z_Z____
_Z______Z________Z_Z___Z_Z_______Z__Z_ZZZ____Z____
_________Z___Z__________Z_Z______ZZZZZZ__ZZ_______
___________ZZ________Z_Z__ZZZZ_Z_____Z__ZZ________
____________Z___Z____ZZ_Z_Z__Z_____Z___ZZZ___Z____
_______Z___Z_Z_______Z_________ZZ_Z____Z_Z____ZZ_Z
______ZZ___Z__Z________Z___Z________Z_Z____Z__ZZZ_
____Z_Z____Z_______Z__Z____Z_Z_____Z_____Z__Z_Z_Z_

# CSS (52, 10, 8)
______X___X__X_X________X___X____XX___X________X____
_______X_XX__XX________X_________________X_______XXX
________X__X___X_X_______X_________XX___XX__X_______
_______X_X________X__X_X_X_X_X__X___X_______________
__X___________X___X____X____X_____X_X_____X__X___X__
__________X_X__X___X_____X__X_____X______XX_X_______
__X_____X____XX____X_____X______X_XX__X_____________
__XX________X_______X_X_______________XX__X___X_X___
X_______________________X___X___XXX_XX_X__X_X___X___
__X______________X___X_X___XX__X_XX_X___X_________X_
___X_____________X_______XX__X__X__X___X____X__XXX__
___X__X__________X_____X_____X__XX__X_XXX_____X_____
________X___X________X_XXXX______X_____XXX_________X
_X______XX_________________XX_X_____X_X__XXXX_______
____X_XX_X__X_______________XX_____XXX______X_____X_
____XX__X_X__________XX___X_X_______X_____X__X_X____
_XXX_X_XXX_X____________X______________X________XX__
X____X______X____________X_XX_________XXX_X__X__X___
_XX____X_________X____X__X__X_________X_XX____XX____
_____________X_______XX_____X_XXX_______X_X_X__X_X__
XX_X_______XX___XX____________X__X_X_X______X_______
__________Z______Z_____Z_Z_____ZZ__ZZ_Z_______Z_____
_____Z___Z_______Z_____Z_ZZ__Z___ZZ________Z________
___ZZZ___Z__ZZ__________Z____Z_____Z________Z_______
ZZ___Z___________Z___Z______________Z___Z_Z_Z_Z_____
Z_____________________Z_____Z___ZZ__Z_Z_Z_Z__ZZZ____
_Z__________Z__________Z___Z_Z_Z____Z___Z_Z______ZZZ
__Z_________Z_________ZZ_____Z________ZZZ___Z_ZZ__Z_
_ZZZ_Z___________Z_______Z_____Z____Z__Z___ZZ___Z___
______Z__Z________Z_______Z___Z_Z_Z_Z__Z____Z___ZZ__
________Z________Z________Z_ZZ____ZZZ_Z______Z_ZZ___
___Z__Z_ZZ______________Z_____Z_____Z_Z____Z_Z_Z__Z_
_________Z_Z___Z____Z________ZZ___________ZZ__ZZZZ__
___________Z______Z_Z_______Z_Z_Z_Z_____Z____Z__Z_ZZ
Z__Z___Z_______Z_Z____ZZ____Z_ZZ_______Z________Z___
_Z_____Z_Z_______________Z__Z______ZZ___Z____Z_Z_ZZ_
__Z_________ZZ______________Z_ZZ_____Z__Z_Z_Z_Z__Z__
______ZZ______Z___Z_________Z___Z___Z____Z_Z__Z_Z__Z
______Z_Z_ZZ________________Z_Z_______ZZZZZ_____Z___
Z________Z__Z__ZZ_______________Z_Z__Z___Z____Z__ZZ_
_Z__Z_______Z____Z______Z_Z__ZZ_Z_____Z__Z________Z_
_ZZ_Z_______Z_____ZZ___Z_______________Z___Z_Z__Z__Z

# CSS (52, 12, 7)
X_________________X______X___X__X__X___XX___X_X___XX
___X__XX________X_X___X__X_X________X_XX___________X
__X___X_X_X____________X____X___X_X__XXX______X_____
_X_________XX_______X____XX___X___X___XX___X_____X__
XX_X____X__XX_____X__X__XX____________X________X____
_X_X__X__X_XX______________X_XXX_____________X____X_
___X__X______X_________X______X__X_X_X__X______XXX__
______X_XXX__X____________________X_X_X___XX_X_____X
__X_________XX_________X_____X_XXXXX_____X______X___
X_______X_____X___X______X__X__X__XXX___________XX__
_____X___X____X__________X_____X__X___X_X_X_X____X_X
___X_________XX_____XX_X_X________X____XXX________X_
_____XX__X___XX___________X____X_______X_XX___XX____
__X_____X__X_XX_____XX_______XX___________X__X___X__
_________X_X___X_____X________X___XXX_XX_____X_X____
________________X_X_____X___X___X_X_X_XXX_XX________
____X_XXX___X________X_X________XX_X_______X_______X
___X__X__XXX_____X___X__X__X___________X___X_______X
X___X_X______________X_X_________X__X_______XXX_X_X_
_XX_X_X_____X______XX________X___X__X__X_X__________
__Z____ZZZ___Z___Z__________Z__Z______Z__Z_______Z__
_______ZZ__Z_Z_Z___________ZZ___Z_ZZ______Z_________
Z____ZZ__Z__Z________Z________ZZ______Z____Z______Z_
____Z__Z_________Z__ZZ_____Z_________ZZ_Z___Z_____ZZ
_______ZZ___________Z___Z___ZZ_Z_Z_Z___Z__Z___Z_____
_____ZZZ________ZZ_____Z_____Z____ZZ__Z___Z____Z____
_Z_____Z_______ZZ___Z_____Z__Z_____ZZ_Z_Z_____Z_____
_Z__Z________Z__Z_____Z__Z______Z_____Z_Z______Z_ZZ_
___Z______Z__Z______Z______Z__Z____ZZ_Z___ZZ______Z_
_Z_____ZZ_Z___Z___________ZZ__Z__Z_Z_____________ZZ_
___Z__Z__Z______Z_____ZZZ______Z_Z_Z_Z________Z_____
_Z_Z__Z______Z________Z______Z______Z__Z_Z__Z___Z__Z
__Z_Z_Z__ZZ__Z____________Z_____Z___Z___Z__Z____Z___
____ZZZ_Z_Z_________ZZ____Z__Z_____Z__________Z____Z
_ZZZ_____Z_Z________ZZZZ____Z_____Z______Z__________
______Z_________Z_Z__Z__Z__Z_Z_Z______Z_Z__Z_______Z
ZZ____Z____ZZ_Z__Z_________Z___Z__Z_____ZZ__________
ZZZ___Z__Z________________Z_Z______ZZ_______ZZZ_____
___Z_____________Z_Z__Z_Z_Z__ZZZZ_Z________Z________
Z__ZZ_Z__________________Z__Z_ZZ_____ZZ_________ZZ__

# CSS (56, 12, 8)
___X___________X_XX____X___X____X______X_____X_X________
_X__________________X____XX______X__X___XX_____XX____XX_
__X__X_______________X_X______X___X__XX_XX_______XX_____
____X__X________________XX_X________XX___X_X_______X_X_X
_____X__X_______________XX_____________X_X____X_XXX__XX_
____XX_____X__X__X__________X__________X___XX________XXX
____X__X______X_X_______________X__X______X__X__X___X_XX
____XX_______X_XX__X________X_______X_____X_X____X___X__
_____X__X__X_X_____X_X____X__X___________X____X__X__X___
________X__X__X______X__X____XX__X_X___XX___X___________
_________X____X__X___X____________X_____XXXX_X__X__X____
____XX_X____X__X_X________X_______________X_X__X____X__X
_________X__X_____X_________XX______XX_X__X_X______X___X
X_XX_____X__X________________________X__XXXX________XX__
_______X_XX_____X_X_______XX___X____X_____X__X__X_______
X_____X__X_____X_______X_X___X__X____XX___X________X____
XX_X_______X________X____XXX_X_______X__________XX______
X_____X_XXX_____XX______XX______X____________X_______X__
______X__XX_X______X_______X________X_X_____X___X___X_X_
______X_X_XX__X__X_X______X______XX______X____________X_
X_____X__X__X____XX_X________X__X___X___X___X___________
______X______X_____X__XX__X_____X___XX______X__XX_______
Z______________Z_Z__Z___________Z_____Z__Z_____Z____ZZ__
_ZZ___________Z_________Z_____Z___ZZ__________Z__Z___Z__
____Z___________Z___Z__________Z____Z_Z_______Z_Z_ZZZZ__
_____Z__Z__________ZZ_Z________Z____ZZ_Z__Z__ZZ_________
____Z_Z_________Z_________Z__ZZ___Z___Z__Z_Z__Z____Z____
_____________Z______Z______ZZ__Z____Z_Z________ZZZ_Z___Z
__Z_Z_Z____Z___Z______Z___Z_Z____________Z__ZZZ_________
Z____________ZZ____________ZZ___Z___Z______ZZ_Z_______ZZ
ZZ______Z__ZZZZ___________Z________Z______Z___Z_____Z___
ZZ_______ZZ__Z_____Z_______Z___ZZ__Z_Z__Z_______________
Z_____Z_______Z__ZZ_Z___________Z_Z__Z_Z________Z______Z
_______ZZZ______Z___Z___Z__Z______Z__Z_________Z_ZZ_____
__Z___Z_Z_____Z____Z__Z___Z__Z_____Z__Z___Z___________Z_
_Z_Z__Z___Z_Z___Z______________ZZ__Z__Z_Z_Z_____________
_Z_Z_Z_________ZZ___Z__Z_Z_____ZZ_____________Z______Z__
____ZZZ_Z___Z________________Z_____Z________Z_Z__Z__Z__Z
_____ZZ_____ZZ__Z_________Z___Z__Z________Z___Z_Z_____Z_
Z_Z_Z__Z_Z_____ZZ___Z__________ZZ_______Z___Z___________
Z______Z____Z_______Z____Z____Z_____Z__Z______ZZ_Z_____Z
_______Z_Z____Z_____ZZ__Z_____ZZ__________Z______ZZ___Z_
_ZZ___Z_____Z______Z_______Z___Z__Z______ZZ__Z____Z_____
_____ZZZ________Z___________Z_____ZZ__Z____ZZ_____Z_Z___

# CSS (56, 14, 7)
____X_______X______________X____X_XX________X_X_X__X____
___X_X_________________XX_X_____X______XXX__X_X________X
______X________________XXX__XX__________X_____X_X____XXX
_X_X__X____X______X____X___________X_______X__X__X__X__X
______________X__________X_XX_____X___X_X___X__XX_X_X___
__XX_X_____________________XX____X__X__X_____X_X__X__X__
___X_XX__X_____X_____________X________X_X_X__X___X_____X
XX__X___________X______X__XX_X______________XX_X_X______
_XX__________X__X_________X____X________X_XX____XX____X_
_________XX_X______X__________XX_X__X______XX___X______X
X____X____X______XXX_____X________X__XX__X_____________X
____X_XX____X_______X___XX_X________XX___X___________X__
X_X_X_X____________XX__________X_____X_X__X_XX__________
__X___X______X_______X________XX_X______X_X_____X__X_X__
___X_X_________XX_XX_X___XXX__X__________________X______
X_________________X___X__X_XX____X____XX__XX___X____X__X
X_XXXXX_______________X___X___X_X____X_____X___X______X_
XX_X_______________________X__X_X__X_X_X_X___X_____XX__X
_________XX_______X_____X__X_X_X___X________X____XX_XX_X
___X__X__XX___________XX__XXX____XXX_________X__X_______
_______XX_X_______X______X___XXX_X____X____X_XX_______X_
Z___Z____Z_________ZZ_Z_ZZ______Z_____Z_________________
_____Z__________Z____Z_____Z____ZZ_Z__Z_________Z___ZZ__
________Z____ZZZ___Z_Z_______Z_________Z_________ZZ____Z
__________ZZ_____Z_________Z__Z______Z___Z_ZZ________ZZ_
_ZZ_____Z_____Z____Z_________ZZ_______ZZZ_Z_________Z___
Z____Z______Z_Z______Z_Z____Z__Z_______Z__Z___Z______Z__
____Z__Z_Z____Z________Z____Z____Z_______Z__Z__Z_____Z_Z
___________Z__Z_Z_Z___Z__Z______Z_Z_Z____________ZZ____Z
Z________ZZ_Z___Z_ZZZZZ____________Z______Z_____________
_Z_Z___Z_____Z______Z_____Z____Z____Z____Z_Z__________ZZ
Z___ZZ_Z__Z__Z___Z_Z__Z______Z____Z_____Z____Z__________
_____________ZZ_Z__Z_____ZZ__Z____Z___Z_____Z_____Z__ZZ_
_________________Z___ZZ_____Z__Z_Z_____Z__ZZ_Z___Z__Z__Z
___ZZ_______Z_Z______Z____Z___Z__Z_Z__Z___Z___Z________Z
____Z__Z_____Z____Z_ZZ__Z___ZZ_____Z_______Z______Z____Z
__Z_ZZZZ___Z________Z__ZZ_____ZZ_____Z________Z_________
__Z_ZZ_________Z_Z__Z_ZZ____Z__Z_Z_________________ZZ___
Z_____ZZ____Z_________ZZ_ZZ_ZZ__________Z___ZZ__________
__Z___ZZ____Z____Z___Z________ZZ__Z_Z_____Z_________Z_Z_
_______Z____Z_Z____Z__Z__ZZ___________ZZ_Z____Z__ZZ_____
Z__Z_Z_Z_____Z________Z_Z_____Z___________ZZ__Z__Z_Z____
\end{lstlisting}

\subsection{SWEL codes}
\label{sec:best_codes_swel}
\lstinputlisting{figures/listing_swel.txt}

\clearpage
\section{Refined estimate of logical error rates}
\label{sec:refined_estimate}

Here we derive the refined estimate of logical error rates in \cref{fig:error_estimates}, for the code-capacity depolarizing noise model and lookup table decoder of \cref{fig:error_rates}.
In the refined estimate, each qubit has an $X$, $Y$, or $Z$ error, each with probability $q = \frac13 p_{\text{dep}}$.
Let $\mu_X$ be the number of distinct half-supports (with weight $w_d = \left\lceil\frac{d+1}{2}\right\rceil$) of minimum-weight nontrivial logical $X$ operators.
That is, $\mu_X$ is obtained by enumerating all ${d \choose w_d}$ half-supports of all $2^{O(n)}$ nontrivial weight-$d$ logical $X$ operators, and removing duplicate half-supports from this list.

To lowest order in $q$, a logical $X$ error occurs when every qubit of a half-support of a minimum-weight nontrivial logical $X$ operator $\overline{X}$ has an $X$ or $Y$ error, and no other qubit in $\overline{X}$ has an $X$ or $Y$ error, so that the decoder completes the logical operator.
The $Y$ errors on the half-support also form a $Z$-type error.
If there are more than $t_d = \left\lfloor\frac{d-1}{2}\right\rfloor$ of them, the syndrome of this $Z$-type error is generally not in the lookup table, so the decoder discards the sample rather than decoding it incorrectly.
Of the $2^{w_d}$ ways to place $X$ or $Y$ errors on a half-support,
\begin{align}
  N_d = \sum_{j=0}^{t_d} {w_d \choose j}
\end{align}
have at most $t_d$ $Y$ errors, so the probability of a logical $X$ error is
\begin{align}
  p_L^X = N_d \mu_X q^{w_d} (1-2q)^{n-w_d} + O(q^{w_d+1}).
  \label{eq:error_estimate_x}
\end{align}
With $\mu_Z$ and $p_L^Z$ defined analogously for logical $Z$ operators, and $Z$ or $Y$ errors, the probability of a logical error is
\begin{align}
  p_L \approx 1 - (1 - p_L^X) (1 - p_L^Z),
  \label{eq:error_estimate_refined}
\end{align}
which treats logical $X$ and $Z$ errors as independent.
Dropping the $O(q^{w_d+1})$ terms in \cref{eq:error_estimate_x} and substituting into \cref{eq:error_estimate_refined} yields the refined estimate in \cref{fig:error_estimates}.

With $p = 2q$, the simpler estimate in \cref{eq:error_estimate} is $m {d \choose w_d} (2q)^{w_d}$, which counts all $2^{w_d}$ ways to place $X$ or $Y$ errors on each half-support.
Apart from counting distinct half-supports and the factor $(1-2q)^{n-w_d} \approx 1$, the refined estimate therefore differs from the simpler estimate by the factor $N_d/2^{w_d}$, which depends only on the distance $d$.
The refined estimate neglects terms of higher order in $q$, correlations between logical $X$ and $Z$ errors, and the few patterns with more than $t_d$ $Y$ errors that the decoder nonetheless decodes incorrectly, such as $Y$ errors on every qubit of a half-support of a SWEL code.

\clearpage
\section{Automorphisms of the \texorpdfstring{$\params{20,6,4}$}{[[20,6,4]]} SWEL code}
\label{sec:automorphisms}

Here we provide automorphisms of the $\params{20,6,4}$ SWEL code that we identified with a non-exhaustive AI-assisted search for matrix automorphisms~\cite{sayginel2025faulttolerant}, and verify that these automorphisms generate 2304 SWAP-transversal logical Clifford gates (modulo Pauli gates).
This Python code runs with \texttt{qldpc==0.5.0}.
\begin{tcolorbox}[boxsep=-3pt]
\begin{lstlisting}[language=python]
import itertools
import numpy as np
import qldpc
import stim
import sympy.combinatorics

strings_x = """
X__X____X__X__X_X___
_X___X________X_XX_X
__X_____X_X_X_XX_XX_
___X___XX_X__X___XXX
____X__XXX___XXX_X__
_____X__XXXXXX___X__
______XX__X__XXXX_X_
""".strip().splitlines()
strings_z = [string.replace("X", "Z") for string in strings_x]
code = qldpc.codes.CSSCode.from_strings(strings_x + strings_z).to_swel()
assert code.get_code_params() == (20, 6, 4)

generators = [
    stim.Circuit("H 0 1 2 3 4 5 6 7 8 9 10 11 12 13 14 15 16 17 18 19"),
    stim.Circuit("S 0 1 2 4 6 8 11 12 14 16 17 18 19 \n S_DAG 3 5 7 9 10 13 15"),
    stim.Circuit("SWAP 0 15 1 14 2 9 3 8 5 16 6 11 13 18 17 19"),
    stim.Circuit("SWAP 0 17 1 8 2 7 3 5 4 18 6 15 9 13 10 12 11 19 14 16"),
    stim.Circuit("SWAP 0 11 1 5 7 13 10 18"),
]
logical_tableaus = [qldpc.circuits.get_logical_tableau(code, gen) for gen in generators]

def get_permutation(tableau: stim.Tableau) -> sympy.combinatorics.Permutation:
    """Convert a Clifford tableau into an (exponentially large) Sympy Permutation."""
    # Convert the tableau into a 2k x 2k symplectic matrix, grouping blocks in such a
    # way that the embedding is a homomorphism with respect to matrix multiplication:
    # M(a.then(b)) == M(a) @ M(b).
    x2x, x2z, z2x, z2z, _, _ = tableau.to_numpy()
    matrix = np.block([[x2x, x2z], [z2x, z2z]]).astype(np.uint8)

    # Collect all length-2k bitstrings into a matrix of column vectors.  Each
    # column represets one Pauli string; find how the tableau permutes them.
    bitstring_iterator = itertools.product([0, 1], repeat=len(matrix))
    old_bitstrings = np.array(list(bitstring_iterator), dtype=np.uint8).T
    new_bitstrings = (matrix @ old_bitstrings) % 2

    # Convert each bitstring into an integer (big-endian, matching itertools.product).
    bit_to_int = 1 << np.arange(len(matrix), dtype=np.int64)[::-1]
    permutation = bit_to_int @ new_bitstrings  # the permutation in one-line notation
    return sympy.combinatorics.Permutation(permutation.tolist())

permutations = [get_permutation(tab) for tab in logical_tableaus]
group = sympy.combinatorics.PermutationGroup(permutations)

for tableau in logical_tableaus:
    print(tableau)
    print()
print("group size:", group.order())
\end{lstlisting}
\end{tcolorbox}

\clearpage
\section*{Disclaimer}

This paper was prepared for informational purposes with contributions from the Global Technology Applied Research center of JPMorgan Chase \& Co. This paper is not a product of the Research Department of JPMorgan Chase \& Co. or its affiliates. Neither JPMorgan Chase \& Co. nor any of its affiliates makes any explicit or implied representation or warranty and none of them accept any liability in connection with this paper, including, without limitation, with respect to the completeness, accuracy, or reliability of the information contained herein and the potential legal, compliance, tax, or accounting effects thereof. This document is not intended as investment research or investment advice, or as a recommendation, offer, or solicitation for the purchase or sale of any security, financial instrument, financial product or service, or to be used in any way for evaluating the merits of participating in any transaction.

\end{document}